\documentclass[12pt]{article}
\usepackage{amsmath}
\usepackage{graphicx}
\usepackage{enumerate}
\usepackage{enumitem}
\usepackage{natbib}
\usepackage{url} 

\usepackage{amsfonts,amssymb,amsthm}
\usepackage{mathrsfs}
\usepackage{bm}
\usepackage{bbm}
\usepackage{apptools}
\usepackage{makecell}
\usepackage{chngcntr}
\usepackage[most]{tcolorbox}
\usepackage{booktabs}
\usepackage{xr}
\usepackage{xcolor}

\newcommand{\blind}{1}

\def\bmA{\bm{A}}
\def\bmB{\bm{B}}
\def\bmC{\bm{C}}
\def\calA{\mathcal{A}}
\def\calB{\mathcal{B}}
\def\calZ{\mathcal{Z}}
\def\bmD{\bm{D}}
\def\bme{\bm{e}}
\def\bmE{\bm{E}}
\def\bmI{\bm{I}}
\def\bmM{\bm{M}}

\def\bmY{\bm{Y}}

\def\bmy{\bm{y}}

\def\bmU{\bm{U}}
\def\bmV{\bm{V}}
\def\bmP{\bm{P}}
\def\bmQ{\bm{Q}}
\def\bmG{\bm{G}}
\def\bmH{\bm{H}}
\def\bmK{\bm{K}}

\def\bmgm{\bm{\gamma}}
\def\bmZ{\bm{Z}}
\def\bmtt{\bm{\theta}}

\def\tp{\mathsf{T}}
\def\tr#1{\mathrm{tr}{\left[#1\right]}}

\def\vecc#1{\mathrm{vec}{\left(#1\right)}}
\def\vech#1{\mathrm{vech}{\left(#1\right)}}
\def\veccq#1{\mathrm{vec}_{-1}{\left(#1\right)}}
\def\vechq#1{\mathrm{vech}_{-1}{\left(#1\right)}}
\def\E#1{\mathbb{E}{\left(#1\right)}}
\def\Var#1{{\rm Var}\left(#1\right)}

\newtheorem{Lemma}{Lemma}
\newtheorem{Theorem}{Theorem}
\newtheorem{Proposition}{Proposition}

\newtheorem{Corollary}{Corollary}
\newtheorem{Remark}{Remark}

\DeclareMathOperator*{\argmin}{arg\,min}
\begin{document}

\def\spacingset#1{\renewcommand{\baselinestretch}%
{#1}\small\normalsize} \spacingset{1}


\if1\blind
{
  \title{\bf Mixture Matrix-valued Autoregressive Model}
  \author{Fei Wu\thanks{
    Fei Wu was partially supported by the {NSF CAREER grant 2045016 during the PhD studies at the University of Iowa.}}\hspace{.2cm}\\
   Department of Statistics and Actuarial Science, University of Iowa\\
   fei-wu-1@uiowa.edu\\
    and \\
   Kung-Sik Chan\thanks{Corresponding author.}
   \\
      Department of Statistics and Actuarial Science, University of Iowa\\
  kung-sik-chan@uiowa.edu}
\date{}
  \maketitle
} \fi

\if0\blind
{
  \bigskip
  \bigskip
  \bigskip
  \begin{center}
    {\LARGE\bf Title}
\end{center}
  \medskip
} \fi

\bigskip
\begin{abstract} Time series of matrix-valued data are increasingly available  in various areas including  economics,  finance, social science, among others. These data may shed light on the inter-dynamical relationships between two sets of attributes, for instance, countries and economic indices. The matrix autoregressive (MAR) model provides a parsimonious approach for analyzing such data.  However, the MAR model, being a linear model with parametric constraints, cannot capture the nonlinear patterns in the data, such as regime shifts in the dynamics. We propose a mixture matrix autoregressive (MMAR) model for analyzing potential regime shifts in the dynamics between two attributes, for instance, due to recession versus expansion, or stable period versus pandemic. We propose an EM algorithm for maximum likelihood estimation. We derive some theoretical properties of the proposed method including consistency and asymptotic distribution, and illustrate its performance via simulations and real applications.
\end{abstract}

\noindent%
\vfill

\newpage
\spacingset{1.3}

\section{Introduction}
\label{sec:intro}
Recent technological advances facilitate the collection of time series data with complex structures, for instance, matrix-valued time series data from various fields, including economics, finance, political science, and social science. In economics, important national economic indices are reported regularly over time, naturally forming a sequence of matrices cross classified by country and index. In finance, matrix-valued time series data are commonly encountered when dealing with monthly portfolio returns. These returns can be represented as a sequence of matrices, where stocks are grouped into portfolios based on their market capital levels and book-to-equity ratio. Other financial applications include bond yields
at different maturities and in different countries. One could also consider different credit ratings, etc.
Dynamic graphs are commonly used in political science, social science, and other related fields, where a matrix can represent the graph or network at each time point. In addition, matrices can also represent 2D grayscale images, and a sequence of images can form a matrix time series.

One approach to modeling matrix-valued time series data is to vectorize the matrices and fit a multiple time series model, e.g., the vector autoregressive (VAR) model or some state space model \citep{hannan1970multiple,lutkepohl2005new}. However, the vectorization approach suffers from the ``curse of dimensionality" even with moderately large matrices. Alternative approaches have been developed to address this issue, for instance, the regularized VAR models \citep{basu2015regularized,nicholson2020high} and the factor models \citep{lam2012factor,pena2019forecasting,fan2020factor}.
Nonetheless, these methods may not be appropriate for matrix-valued time series data because they ignore the information contained in the matrix structures. 

The matrix autoregressive (MAR) model, proposed by \cite{chen2021autoregressive}, is a parsimonious model which preserves the matrix structure. It is also referred to as the bilinear model. Hoff (2015) proposed the bilinear model to study matrix-valued longitudinal relational data, and he also developed multi-linear models for tensor-valued data. \cite{ding2018matrix} studied the bilinear regression model under the envelope framework. \cite{hsu2021matrix} introduced the spatio-temporal MAR model. Multi-linear autoregressive models for tensor-valued time series were proposed by \cite{li2021multi}, and tensor decomposition methods have also been applied to model matrix-valued or tensor-valued time series \citep{wang2021high,han2021cp,chang2022modelling}.

It can be shown that the MAR model and the multi-linear autoregressive model can be expressed as some parametrically constrained VAR model. However, time series data may be generated from some nonlinear process, which displays nonlinear patterns, for instance, conditional or marginal multimodality
in which case linear Gaussian models are inappropriate.
For example, economic data may follow different dynamics over different growth phases -- either in a fast or slow growth phase \citep{hamilton1989new}.
Various models have been developed for nonlinear time series data \citep[see, e.g.,][]{tong1990non,fan2003nonlinear}. One well-established nonlinear model is the mixture autoregressive model, first introduced by \cite{wong2000mixture} as a generalization of the mixture transition distribution model \citep{le1996modeling}. This model has several interesting properties.  It may contain a non-stationary AR component, but remains overall stationary; it is able to capture conditional heteroscedasticity.
Many extensions have been proposed for the mixture autoregressive model. For example, \cite{fong2007mixture} introduced the mixture VAR model. \cite{kalliovirta2015gaussian, kalliovirta2016gaussian} proposed the time-inhomogeneous mixture autoregressive models, where the mixing weights  may vary with time.
Note that the mixture autoregressive model is a special case of the threshold autoregressive model and the Markov-switching autoregressive model \citep{tong1990non}.

Here, we propose a mixture matrix autoregressive (MMAR) model, an extension of both the MAR model and the mixture autoregressive model. This model enables us to segment the matrix time series into different regimes. 
Our extension is motivated by the need to analyze the economic indicator dataset (https://data.oecd.org) displayed in Fig.~\ref{fig_econo}. This dataset contains four economic indicators: quarterly short-term interest rate (first difference),
quarterly GDP (annual percentage growth), quarterly industrial production (first difference of the logarithm of the data), and annual growth rate of quarterly CPI (first difference), from five countries: United States, Germany, France, United Kingdom and Canada, from Q1 1990 to Q4 2022. 
\cite{chen2021autoregressive} applied the MAR model to analyze a similar dataset. Although the dataset is generally stabilized by the logarithmic transformation and/or differencing, some synchronized irregular patterns are observed in the plot. Notably, nearly all indicators experienced a sharp decline followed by a rapid recovery during 2008 and 2009 across all five countries, which may be attributed to the global economic crisis in 2008. Even more dramatic fluctuations were observed between 2020 and 2022, presumably due to the pandemic. 
In summary, the economic indicator dataset appears to be nonlinear, displaying simultaneity in regimes across countries, and hence a nonlinear time series model would be better suited for analyzing and interpreting this dataset. In particular, the proposed method may be suitable for analyzing matrix time series displaying simultaneity in regimes across certain matrix dimensions. 
Moreover, segmenting this dataset into different dynamical regimes can provide valuable insights into the global economic dynamics. For instance, we conjecture that under certain conditions, regimes may be 
more accurately dated when pooling multiple countries in this way rather than studying them
one-by-one. 

Recently, some mixture models have been developed to cluster matrices \citep{gao2021regularized} and tensors \citep{mai2022doubly}.  Those models,  however, assumed a fixed mean structure for each component, which cannot capture shifts in temporal dynamics.

Our contributions are three-fold. First, we build a nonlinear autoregressive model for matrix-valued time series data. Our model expands the scope of regime-switching autoregressions, making the methods applicable to more complex time series data.
Compared to some recently emerged models on matrix-valued time series, the proposed model not only offers a more comprehensive characterization of nonlinear patterns, but it can also classify the data into different regimes, which can  enhance our understanding of the dataset.
Second, both strict and weak stationarity conditions for the model are given, and 
an EM algorithm for maximum likelihood estimation is implemented.
Third, we establish some 
asymptotic properties of the maximum likelihood estimator.

This paper is organized as follows. The proposed MMAR model is elaborated in Section 2. Strict and weak stationarity conditions  of the MMAR model are given in Section 3. An EM algorithm for parameter estimation is described in Section 4. The theoretical properties of the maximum likelihood estimator are investigated in Section 5. Model selection is discussed in Section 6. Section 7 presents simulation studies and real data analysis. Finally, Section 8 concludes the paper and suggests avenues for future research. Proofs of the main results and additional numerical results are provided in the \emph{Supplemental Materials}.

\begin{figure}[ht!]
	\centering
	\includegraphics[width=1\textwidth]{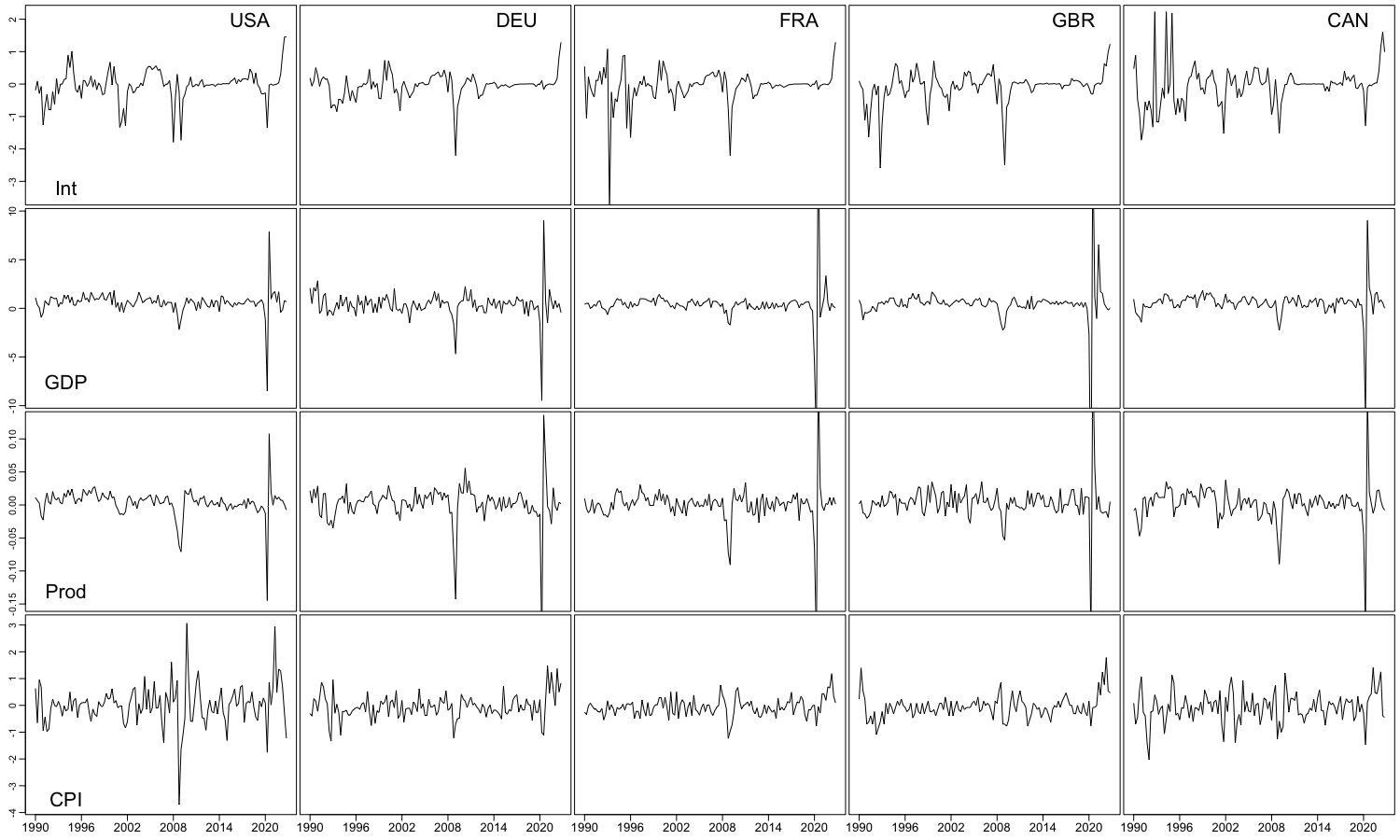}
	\caption[Time series of four economic indicators from five countries.]{\label{fig_econo} Time series of four economic indicators from five countries.}
\end{figure}

\section{Model Formulation}
\label{sec:models}
\subsection{The MAR Model}
Let $\bmY_t\in\mathbb{R}^{m\times n},~1\leq t\leq T$ be the matrix-valued time series data. The $p$th-order matrix autoregressive model, denoted by MAR($p$), specifies the relationship,
\begin{equation}
	\bmY_t = \bmC_0+\sum_{i=1}^p\bmA_i\bmY_{t-i}\bmB_i^\tp+\bmE_t,
	\label{MAR_model}
\end{equation}
where $\bmA_i\in\mathbb{R}^{m\times m}$ and $\bmB_i\in\mathbb{R}^{n\times n}$ are parameter matrices, and $\bmE_t$ is the matrix of random errors. The parameter matrix $\bmC_0$ is the intercept matrix, which
is generally absent for centered data. This model admits some interesting interpretations. For example, in an MAR(1) model, the parameter matrices $\bmA_1$ and $\bmB_1^\tp$ reflect row-wise and column-wise interactions, respectively; it can also be viewed as a factor regression model, with the factor being $\bmY_{t-1}\bmB_1^\tp$ \citep{chen2021autoregressive}.

Let $\vecc{\cdot}$ denote the vectorization of the enclosed matrix via stacking its columns. Also, let the operator $\vech{\cdot}$ denote the half-vectorization of a symmetric matrix.
The MAR($p$) model can be expressed as
\begin{equation}
	\vecc{\bmY_{t}}=\vecc{\bmC_0}+\sum_{i=1}^p\left(\bmB_i\otimes\bmA_i\right)\vecc{\bmY_{t-i}}+\vecc{\bmE_t},
	\label{MAR_V}
\end{equation}
where $\otimes$ represents the Kronecker product of matrices. Hence the MAR($p$) model is intrinsically a constrained $p$th-order VAR model. It is assumed that 
$\{\bmE_t~|~1\leq t\leq T\}$ is a sequence of independent and identically distributed (i.i.d.) random matrices such that $\bmE_t$ is independent of $\{\bmY_{t-1},\bmY_{t-2},\dots\}$. Also, $\E{\bmE_t}=\bm{0}$ and $\Var{\vecc{\bmE_t}}=\bm{\Sigma}$, where
$\bm{\Sigma}$ is positive definite.
Throughout, we denote $\bm{0}$  as either a zero matrix or a zero vector with a suitable dimension. 

We can further specify that $\bm{\Sigma}$ is separable:
$\bm{\Sigma}=\bmV\otimes\bmU$, where $\bmU\in\mathrm{R}^{m\times m}$ and $\bmV\in\mathbb{R}^{n\times n}$ are all positive definite matrices. This covariance structure has gained significant attention in multivariate analysis, especially in cases where variables can be cross-classified by two factors, such as spatiotemporal data. Hypothesis tests have also been developed for this separable covariance structure, see, e.g., \cite{lu2005likelihood}. Let $\mathscr{F}_{t}$ be the $\sigma$-algebra generated by $\bmY_{t-j}$, $j\geq0$. 
Under the separability assumption of $\bm{\Sigma }$, if $\{\bmE_t~|~1\leq t\leq T\}$ is normally distributed, then the conditional distribution of $\bmY_{t}$ given $\mathscr{F}_{t-1}$ follows a matrix normal distribution with mean $\bmC_0+\sum_{i=1}^p\bmA_i\bmY_{t-i}\bmB_i^\tp$ and variance-covariance matrices $\bmU$ and $\bmV$, in symbol, $\bmY_t\sim \mathcal{MN}_{m,n}(\bmC_0+\sum_{i=1}^p\bmA_i\bmY_{t-i}\bmB_i^\tp,\bmU,\bmV)$, whose joint probability density function is given by, 
\begin{equation}
	f_{\mathcal{MN}}(\bmY_t|\bmC_0+\sum_{i=1}^p\bmA_i\bmY_{t-i}\bmB_i^\tp,\bmU,\bmV)
	=\frac{\exp\left(-\frac{1}{2}\tr{\bmV^{-1}\bm{E}_t^\tp\bmU^{-1}
			\bm{E}_t}\right)
	}
	{(2\pi)^{mn/2}\det(\bmV)^{m/2}\det(\bmU)^{n/2}},
\end{equation}
where $\bm{E}_t=\bmY_t-\bmC_0-\sum_{i=1}^p\bmA_i\bmY_{t-i}\bmB_i^\tp$, and $\det(\cdot)$ denotes the determinant of the enclosed matrix.
It can be shown that, if $\bmY_t\sim \mathcal{MN}_{m,n}(\bmC_0+\sum_{i=1}^p\bmA_i\bmY_{t-i}\bmB_i^\tp,\bmU,\bmV)$, then $\vecc{\bmY_t}$  follows an $mn$-dimensional multivariate normal distribution with mean $\vecc{\bmC_0+\sum_{i=1}^p\bmA_i\bmY_{t-i}\bmB_i^\tp}$ and variance-covariance matrix $\bmV\otimes\bmU$, denoted by,
$$\vecc{\bmY_t}\sim \mathcal{N}_{mn}(\mathrm{vec}(\bmC_0+\sum_{i=1}^p\bmA_i\bmY_{t-i}\bmB_i^\tp), \bmV\otimes\bmU),$$ and its joint probability density function is denoted by
$$f_\mathcal{N}(\vecc{\bmY_t}|\mathrm{vec}(\bmC_0+\sum_{i=1}^p\bmA_i\bmY_{t-i}\bmB_i^\tp),\bmV\otimes\bmU).$$

Note that the MAR($p$) model is not identifiable  as the model is unchanged by multiplying $\bmA_i$ by some non-zero constant and dividing $\bmB_i$ by the same constant, for any $i\in\{1,\dots,p\}$, and so do $\bmU$ and $\bmV$. Thus, the model requires some identifiability constraints, for example, $\|\bmB_i\|_F=\|\bmV\|_F=1$, and the first non-zero element of $\vecc{\bmB_i}$ is positive for $1\leq i\leq p$, where  $\|\cdot\|_F$ denotes the Frobenius norm of the enclosed matrix. The subscript $F$ is dropped if there is little risk of ambiguity.

\subsection{The MMAR Model}	
The MMAR($K; p_1,\dots,p_K)$ model consists of a probabilistic mixture of
$K$ normal MAR sub-processes, which specifies that the conditional density of $\bmY_t|\mathscr{F}_{t-1}$ is equal to,
\begin{equation}
	\label{mixture density 1}
	\sum_{k=1}^K 
	\alpha_k f_\mathcal{MN}\left(\bmY_{t}|\bmC_k+\sum_{i=1}^{p_k}\bmA_{k,i}\bmY_{t-i}\bmB_{k,i}^\tp,\bmU_k,\bmV_k\right),
\end{equation}
where $p_k$ is the autoregressive order of the $k$th component, 
$0<\alpha_k<1$ is the mixing weight of the $k$th component such that
$\sum_{k=1}^K\alpha_k=1$, $\bmC_k\in\mathbb{R}^{m\times n}$ is the intercept matrix,
$\bmA_{k,i}\in\mathbb{R}^{m\times m}$ and $\bmB_{k,i}\in\mathbb{R}^{n\times n}$ are the non-zero coefficient matrices of the $k$th component, and $\bmU_k\in\mathbb{R}^{m\times m}$ and $\bmV_k\in\mathbb{R}^{n\times n}$ are the corresponding positive definite variance-covariance matrices. The conditional density \eqref{mixture density 1} is equal to,
\begin{align*}
	\sum_{k=1}^{K}\alpha_k
	\left\{
	\frac{\exp\left(-\frac{1}{2}\tr{\bmV_k^{-1}	\bm{E}_{t,k}^\tp\bmU_k^{-1}
			\bm{E}_{t,k}}\right)
	}
	{(2\pi)^{mn/2}\det(\bmV_k)^{m/2}\det(\bmU_k)^{n/2}}
	\right\},
\end{align*}
where 
\begin{align}
	\bm{\bm{E}}_{t,k} = \bmY_{t}-\bmC_k-\sum_{i=1}^{p_k}\bmA_{k,i}\bmY_{t-i}\bmB_{k,i}^\tp.
\end{align}
Since each MAR component in the mixture has a vector representation in the form of $\eqref{MAR_V}$,
the mixture density \eqref{mixture density 1} has the following representation:
\begin{equation}
	\sum_{k=1}^K 
	\alpha_k f_\mathcal{N}\left(\vecc{\bmY_{t}}|\vecc{\bmC_k}+\sum_{i=1}^{p_k}(\bmB_{k,i}\otimes\bmA_{k,i})\vecc{\bmY_{t-i}},
	\bmV_k\otimes\bmU_k\right). \label{mixture density2}
\end{equation}
For comparison, the mixture VAR model introduced by \cite{fong2007mixture} specifies the conditional density as,
\begin{equation}
	\sum_{k=1}^K 
	\alpha_k f_\mathcal{N}\left(\vecc{\bmY_{t}}|\bm{\Psi}_{k,0}+\sum_{i=1}^{p_k}\bm{\Psi}_{k,i}\vecc{\bmY_{t-i}},
	\bm{\Omega}_k\right), \label{miVAR density2}
\end{equation}
where for each $k\in\{1,\dots,K\}$, $\bm{\Psi}_{k,0}$ is an $mn$-dimensional vector, $\bm{\Psi}_{k,i}\in\mathbb{R}^{mn\times mn}$ is a coefficient matrix for $1\leq i\leq p_k$, and $\bm{\Omega}_k\in\mathbb{R}^{mn\times mn}$ is a variance-covariance matrix.
Hence, the proposed MMAR model can be viewed as a constrained version of the mixture VAR model with the restrictions,
\begin{align}
	\bm{\Psi}_{k,0}=\vecc{\bmC_k}, ~~\bm{\Psi}_{k,i}=\bmB_{k,i}\otimes\bmA_{k,i},~~\bm{\Omega}_k=\bmV_k\otimes\bmU_k.
	\label{eq: basic parametric constraints}
\end{align}
For each $k$ and $i$,
the parameter matrix $\bm{\Psi}_{k,i}$ in the unconstrained mixture VAR model contains $m^2n^2$ parameters, while its counterpart in the MMAR model $\bmA_{k,i}$ and $\bmB_{k,i}$ only require $m^2+n^2$ parameters in total. Similarly, $\bm{\Omega}_k$ contains 
$(mn+1)mn/2$
parameters while $\bmU_{k}$ and $\bmV_{k}$ only require $(m+1)m/2+(n+1)n/2$ parameters. It is evident that the number of unknown parameters in the mixture VAR model could be significantly greater than that of the MMAR model, especially as $mn$ increases  and when the model consists of many mixture components with high AR orders. Therefore, comparing with the mixture VAR model, the proposed MMAR model not only preserves the matrix structure, but also results in a substantial reduction in dimensionality.

Similar to the mixture VAR model, the MMAR model has the following interesting properties. First, it can contain both stationary and non-stationary MAR components while maintaining overall model stationarity. An intuitive way to understand this is that the stationary components exhibit contraction patterns (i.e., the spectral radius of its $\bm{\Psi}_{k,1}$ is less than 1, for the case that $p_k=1$), whereas non-stationary components display expansion patterns (i.e., the spectral radius of its $\bm{\Psi}_{k,1}$ is greater  than or equal to 1, for the case that $p_k=1$). The overall model achieves stationarity when the contraction patterns more than offset the expansion patterns. 
Second, it has the capability to model the multi-modality of matrix-valued time series, and we will illustrate these properties through examples.

The MMAR($K;p_1,\dots,p_K$) model has similar identifiability issues
as the MAR model, as for each $k\in\{1,\dots,K\}$ and $i\in\{1,\dots,p_k\}$,
$\bmA_{k,i}$ and $\bmB_{k,i}$ are identifiable up to a constant, if neither are zero matrices, and so are $\bmU_k$ and $\bmV_k$ which are assumed to be positive definite matrices. Therefore, the following constraints are imposed: 
\begin{align}
	& \bm{A}_{k,i}\not = \bm0; \bm{B}_{k,i}\not = \bm0; \mbox{first non-zero element of } \vecc{\bmB_{k,i}} \mbox{ is positive}, \label{iden_c1_0}\\
	& \|\bmB_{k,i}\|_F=1,~~~k\in\{1,2,\dots,K\}, ~i\in\{1,\dots,p_k\},\label{iden_c1} \\
	&\|\vech{\bmV_{k}}\|_F=1,~~~k\in\{1,2,\dots,K\}.	\label{iden_c1_1}
\end{align}

In addition, to circumvent the label-switching problem for the mixture models \citep{mclachlanfinite}, the following constraints are required:
\begin{equation}
	\label{iden_c2}
	0<\alpha_1<\alpha_2<\cdots<\alpha_K<1,
\end{equation}
\begin{align}
	\label{iden_c33}
	&f_\mathcal{N}\left(\vecc{\bmY_{t}}|\vecc{\bmC_k}+\sum_{i=1}^{p_k}(\bmB_{k,i}\otimes\bmA_{k,i})\vecc{\bmY_{t-i}},
	\bmV_k\otimes\bmU_k\right)\nonumber\\
	\neq&
	f_\mathcal{N}\left(\vecc{\bmY_{t}}|\vecc{\bmC_j}+\sum_{i=1}^{p_j}(\bmB_{j,i}\otimes\bmA_{j,i})\vecc{\bmY_{t-i}},
	\bmV_j\otimes\bmU_j\right),~~\forall k\neq j.
\end{align}

\section{Stationarity}
\label{sec:stationarity}
The stationarity and ergodicity of the MMAR model can be studied using tools from two complementary frameworks, namely, Markov chain and stochastic difference equation (SDE) model. The SDE approach may yield more general results for stationarity while Markov chain techniques are useful for studying stationarity and rate of convergence to the stationary distribution, at the expense of requiring more restrictive parametric conditions. Within the Markov chain literature, \cite{fong2007mixture}  derived sufficient conditions for second-order stationarity for the mixture VAR model. \cite{feigin1985random}  studied the geometric ergodicity of a specialized random coefficient VAR model and \cite{saikkonen2007stability} considered the geometric ergodicity of a general class of random coefficient processes. These known results can be specialized to derive sufficient conditions for the stationarity and ergodicity of the MMAR model, as the MMAR model is a constrained random coefficient VAR model. However, to our knowledge, the SDE approach has not been employed to study the stationarity of the random coefficient VAR model. Here, we apply the SDE framework to obtain new sufficient conditions for the stationarity and existence of finite-order moments of the MMAR model that are less restrictive than the conditions deduced from known results in the Markov chain literature. 

A mixture autoregressive model can be embedded in an SDE model, which is also known as the random coefficient autoregression \citep{douc2014nonlinear}. Let $p_{\max}
=\max\{p_1,\dots,p_K\}$.
For $p_k\neq p_{\max}$, define 
\begin{align*}
	\bmA_{k,i}=\bm{0},\quad\bmB_{k,i}=\bm{0},\quad p_k<i\leq p_{\max}.
\end{align*}
Let $\bmy_{t}=\vecc{\bmY_{t}}$, $t\in\{1,\dots,T\}$, and
\begin{align}
	\label{calX}
	\mathcal{X}_t = \begin{pmatrix}
		\bmy_t\\
		\bmy_{t-1}\\
		\vdots\\
		\bmy_{t-p_{\max}+1}
	\end{pmatrix},
\end{align}
\begin{align}
	\bm{\Phi}_k=
	\begin{pmatrix}
		\bmB_{k,1}\otimes\bmA_{k,1}&\bmB_{k,2}\otimes\bmA_{k,2}&\dots&\bmB_{k,p_{\max}-1}\otimes\bmA_{k,p_{\max}-1}&\bmB_{k,p_{\max}}\otimes\bmA_{k,p_{\max}}\\
		\bmI_{mn}&\bm{0}&\dots&\bm{0}&\bm{0}\\
		\bm{0}&\bmI_{mn}&\dots&\bm{0}&\bm{0}\\
		\vdots&\vdots&\ddots&\vdots&\vdots\\
		\bm{0}&\bm{0}&\dots&\bmI_{mn}&\bm{0}	
	\end{pmatrix}, \label{eq: matrix1}
\end{align}
\begin{align}
	\mathcal{C}_k=\begin{pmatrix}
		\vecc{\bmC_k}\\
		\bm{0}\\
		\vdots\\
		\bm{0}
	\end{pmatrix},\,
	\bm{\Lambda}_k=\begin{pmatrix}
		\bm{V}_k^{1/2}\otimes \bm{U}_k^{1/2} &\bm{0} & \bm{0} & \cdots & \bm{0} \\
		\bm{0} & \bm{I}_{mn} & \bm{0} & \cdots &\bm{0}\\
		\vdots & \vdots &  \vdots & \vdots & \vdots & \\
		\bm{0} & \bm{0} & \bm{0} & \bm{0} & \bm{I}_{mn}
	\end{pmatrix},\,
	\mathcal{E}_{t}=\begin{pmatrix}
		\bme_{t}\\
		\bm{0}\\
		\vdots\\
		\bm{0}
	\end{pmatrix},\, \mathcal{E}_{t,k}= \bm{\Lambda}_k \mathcal{E}_{t} \label{eq: matrix2}
\end{align}
where $\bmI_{mn}$ is the $mn\times mn$ identity matrix, 
and $\{\bme_{t}\}$ is a sequence of i.i.d.   $mn$-dimensional random vector  with  i.i.d. standard normal components. Also, $\bme_{t}$ is independent of $\{\bmY_{t-1},\bmY_{t-2},\dots\}$.
Then the MMAR$(K; p_1,\dots,p_K)$ has the following representation as a first-order mixture VAR model:
\begin{align*}
	\mathcal{X}_t = \mathcal{C}_k+ \bm{\Phi}_k	\mathcal{X}_{t-1} + \mathcal{E}_{t,k},\quad \mathrm{with~probability}~\alpha_k,~1\leq k\leq K.
\end{align*}
Let  $\{(\bmD_t,\bm{\eta}_t)\}$ be a sequence of strictly stationary and ergodic random elements. The SDE model for $\mathcal{X}_t$ is defined as,
\begin{equation}
	\mathcal{X}_t= \bmD_t\mathcal{X}_{t-1}+\bm{\eta}_t,
	\label{sre}
\end{equation}
If $\{(\bmD_t,\bm{\eta}_t)\}$ is set to be a sequence of i.i.d. random elements such that 
\begin{eqnarray}
	&&\bm{\eta}_t=\bm{L}_{0,t}+\bm{L}_t \mathcal{E}_t \nonumber \\
	&&\Pr(\bmD_t=\bm{\Phi}_k, \bm{L}_{0,t}=\mathcal{C}_{k}, \bm{L}_{t}=\bm{\Lambda}_k)=\alpha_k,~~~~1\leq k\leq K, \label{dt_mmar}
\end{eqnarray}
then the MMAR($K,p_1,\dots,p_K$) model \eqref{mixture density 1} coincides with the SDE model \eqref{sre}. Let $\|\cdot\|$ denote an arbitrary but fixed matrix norm.
For the SDE model \eqref{sre}, if $\E{\log^{+}(\|\bmD_1\|)}<\infty$, then its top-Lyapunov exponent is defined as
\begin{align}
	\label{top_lya}
	\gamma=\lim_{t\to\infty}\frac{1}{t}\E{\log\|\bmD_t\bmD_{t-1}\dots\bmD_1\|}=
	\inf_{t\in\mathbb{N}^{*}}\frac{1}{t}\E{\log\|\bmD_t\bmD_{t-1}\dots\bmD_1\|}.
\end{align}
Assume that $\{(\bmD_t,\bm{\eta}_t)\}$ is i.i.d., then the $q$th norm Lyapunov coefficient is defined as,
\begin{equation}
	\label{p_lya}
	\gamma_q = \lim_{t\to\infty}\frac{1}{t}\log\left(\mathbb{E}^{1/q}\left(\|\bmD_t\bmD_{t-1}\dots\bmD_1\|^q\right)\right)=\inf_{t\in\mathbb{N}^{*}}\frac{1}{t}\log\left(\mathbb{E}^{1/q}\left(\|\bmD_t\bmD_{t-1}\dots\bmD_1\|^q\right)\right),
\end{equation}
where $q>0$.
Neither $\gamma$ nor $\gamma_q$ depends on the choice of the matrix norm $\|\cdot\|$ \citep{douc2014nonlinear}.
\subsection{Strict Stationarity}
The strict stationarity of the MMAR model is established by the following proposition.
\begin{Proposition}\label{mmar_str_cond}
	Assume that $\{(\bmD_t,\bm{\eta}_t)\}$ is a sequence of i.i.d. random elements such that \eqref{dt_mmar} holds. If the top-Lyapunov exponent, defined by \eqref{top_lya},
	is strictly negative, then the MMAR model has a unique strictly stationary solution, whose vectorization is given by,
	\begin{equation}
		\tilde{\mathcal{X}_t}=\sum_{j=0}^\infty\left(\prod_{i=t-j+1}^t\bmD_i\right)\bm{\eta}_{t-j}\label{sre_solu}.
	\end{equation}
\end{Proposition}
A sufficient condition for the top-Lyapunov exponent $\gamma$ to be strictly negative is that,
\begin{equation*}
	\E{\log\|\bmD_1\|}=\sum_{k=1}^K\alpha_k\log(\|\bm{\Phi}_k\|)<0.
\end{equation*}
Let $\rho(\bm{\Phi}_k)$ denote the spectral radius of $\bm{\Phi}_k$. By the relationship between the spectral radius and matrix norms, for any $\varepsilon>0$, there exists a matrix norm $\|\cdot\|_{*}$, such that,
\begin{equation}
	\label{sp-norm-relation}
	\rho(\bm{\Phi}_k)\leq \|\bm{\Phi}_k\|_{*}\leq \rho(\bm{\Phi}_k)+\varepsilon.
\end{equation}
By the arbitrariness of $\varepsilon$, we derive the following corollary:
\begin{Corollary}
	\label{mmar_ss_cor}
	A sufficient condition for the MMAR($K;p_1,\dots,p_K$) model to have a strictly stationary and ergodic solution is
	$\sum_{k=1}^K\alpha_k\log(\rho(\bm{\Phi}_k))<0$.
	For an MMAR($K;1,\dots,1$) model, the condition can be simplified to,
	\begin{equation*}
		\sum_{k=1}^K\alpha_k\log(\rho(\bmB_{k,1})\rho(\bmA_{k,1}))<0.
	\end{equation*}
\end{Corollary}
\begin{Remark}
	If $\rho(\bm{\Phi}_k)<1$, then the $k$th component MAR process is stationary. Therefore, by Corollary \ref{mmar_ss_cor} if all the components are stationary, then the MMAR model is also stationary.   
\end{Remark}

The ergodicity of the MMAR model is established by the following proposition.
\begin{Proposition}
	\label{mmar_ergod}
	Let $\{\bmY_t\}$ be an MMAR process, and $\mathcal{X}_t$ defined in \eqref{calX}. If $\{\mathcal{X}_t\}$ is strictly stationary, and the initial values are generated from the stationary distribution, then it is also ergodic.
\end{Proposition}

The MMAR model is a $p_{\max}$-order Markov chain. By imposing stronger parametric conditions, the preceding ergodicity result can be strengthened to geometric ergodicity, specifically, under some conditions elaborated below, the MMAR model admits a unique stationary distribution to which the $m$-th transition probability distribution of the MMAR model converges geometrically fast, as $m\to\infty$, in total variation norm, for any fixed initial values.  Geometric ergodicity implies that the stationary MMAR process is $\beta$-mixing with a geometric decaying mixing rate \citep[p. 89]{doukhan1995mixing}. The following result generalizes Theorem 3 of \cite{feigin1985random} and is a special case of Theorem 1 of \cite{saikkonen2007stability} for a wide class of random coefficient processes.  For completeness, we include in the \emph{Supplementary Materials} a brief proof of the following  proposition~\ref{mmar_gergod}, by modifying the proof of  \citep[Theorem 3]{feigin1985random}.

\begin{Proposition}\label{mmar_gergod}
	Assume that (i) $\{(\bmD_t,\bm{\eta}_t)\}$ is a sequence of i.i.d. random elements such that \eqref{dt_mmar} holds, (ii) $\{\bme_{t}\}$ is a sequence of i.i.d.   $mn$-dimensional random vector  with  i.i.d. standard normal components and  $\bme_{t}$ is independent of $\{\bmY_{t-1},\bmY_{t-2},\dots\}$, and (iii) the spectral radius of $\mathscr{J} =\E{\bm{D}_t\otimes \bm{D}_t}$
	is strictly less than 1. Then the MMAR model is geometric ergodic. 

\end{Proposition}
\begin{Remark}
	Since $\E{\bm{D}_t\otimes \bm{D}_t}=\sum_{k=1}^K \alpha_k \bm{\Phi}_k \otimes \bm{\Phi}_k$, its spectral radius is upper-bounded by  $\sum_{k=1}^K \alpha_k \rho^2(\bm{\Phi}_k)$. Thus, the MMAR model is geometric ergodic if $\sum_{k=1}^K \alpha_k \rho^2(\bm{A}_k)\rho^2(\bm{B}_k)<1$. The latter condition implies $\sum_{k=1}^K\alpha_k\log(\rho(\bm{A}_k)\rho(\bm{B}_k))<0$, thanks to Jensen's inequality.     
\end{Remark}

\subsection{Weak Stationarity}
The tails of the stationary solutions are heavier than those of $\bm{\eta}_t$, and may not have finite second-order moments even if $\bm{\eta}_t$ is Gaussian \citep[pp.~91--92]{douc2014nonlinear}. Thus, it is possible that the MMAR model is strictly stationary but not second-order (weakly) stationary. For the MMAR($K;1,\dots,1$) model, 
its first-order and the second-order stationarity conditions can be established based on the results in \cite{fong2007mixture}.
\begin{Proposition}\label{mmar_1st}
	The MMAR($K;1,\dots,1$) model is stationary in the mean if and only if 
	all the eigenvalues of $\sum_{k=1}^K\alpha_k(\bmB_{k,1}\otimes\bmA_{k,1})$
	have modulus less than 1.
\end{Proposition}
\begin{Proposition}\label{mmar_2nd}
	Assume the MMAR($K;1,\dots,1$) model is stationary in the mean. Then it is second-order stationary if and only if 
	all the eigenvalues of $\sum_{k=1}^K\alpha_k\{(\bmB_{k,1}\otimes\bmA_{k,1})
	\otimes(\bmB_{k,1}\otimes\bmA_{k,1})\}$
	have modulus less than 1.
\end{Proposition}
Next, we consider the conditions for the existence of $q$th-order stationary solutions to the MMAR ($K;p_1,\dots,p_K$) model.
The following proposition gives the conditions for the stationary solutions of the MMAR model to admit moments of order $q\geq1$.
\begin{Proposition}\label{mmar_wk_cond}
	Assume that $\{(\bmD_t,\bm{\eta}_t)\}$ is a sequence of i.i.d. random elements such that \eqref{dt_mmar} holds. If the $q$th norm Lyapunov coefficient, defined by \eqref{p_lya},
	is strictly negative. Then the MMAR model has a unique strictly stationary solution, whose vectorization is given in \eqref{sre_solu}, such that $\E{\|\tilde{\mathcal{X}_t}\|^q}<\infty$. Moreover, the right-hand-side of  \eqref{sre_solu} converges in the $q$th norm.
\end{Proposition}
Similar to the top-Lyapunov coefficient $\gamma$, a sufficient condition for $\gamma_q<0$ is
$
\log\left(\mathbb{E}^{1/q}(\|\bmD_1\|^q)\right)<0,
$
which is equivalent to
\begin{equation*}
	\mathbb{E}(\|\bmD_1\|^q)=\sum_{k=1}^K\alpha_k\|\bm{\Phi}_k\|^q<1.
\end{equation*}
Using \eqref{sp-norm-relation} again, we can derive the following corollary.
\begin{Corollary}
	\label{mmar_ws_cor}
	A sufficient condition for the MMAR($K;p_1,\dots,p_K$) model to have a stationary and ergodic solution with finite $q$th moment is
	$
	\sum_{k=1}^K \alpha_k(\rho(\bm{\Phi}_k))^q<1.
	$
	For the MMAR($K;1,\dots,1$) model, the condition can be expressed as,
	\begin{equation*}
		\sum_{k=1}^K\alpha_k(\rho(\bmB_{k,1})\rho(\bmA_{k,1}))^q<1.
	\end{equation*}
\end{Corollary}
Below, we exhibit an MMAR model comprising both stationary and nonstationary components, while the overall model is strictly stationary. 
\\ \textbf{Example 1:} Consider an MMAR($2;1,1$) model.
Let $\bmY_t\in\mathbb{R}^{2\times2}$, $\alpha_1=0.4$, $\alpha_2=0.6$, and 
\begin{align*}
	&\bmB_{1,1}\otimes\bmA_{1,1}=	\begin{pmatrix}
		0.3&0.4\\0.6&0.3
	\end{pmatrix}\otimes\begin{pmatrix}
		0.5&0.7\\0.55&0.4\end{pmatrix},\\
	&\bmB_{2,1}\otimes\bmA_{2,1}=\begin{pmatrix}
		0.6&0.3\\0.2&0.4
	\end{pmatrix}\otimes\begin{pmatrix}
		1.1&0.2\\0.4&1.2
	\end{pmatrix},
	\nonumber
\end{align*}
Also, we assume that $\bmC_1=\bmC_2=\bm{0}$ and $\bmV_1\otimes\bmU_1 = \bmV_2\otimes\bmU_2 = \bmI_4$. By Proposition 1 in \cite{chen2021autoregressive}, the first MAR component is second-order stationary as $\rho(\bmB_{1,1}\otimes\bmA_{1,1})=0.847<1$, 
while the second MAR component is not because $\rho(\bmB_{2,1}\otimes\bmA_{2,1})=1.099>1$. But the overall model is strictly stationary as $\sum_{k=1}^2\alpha_k\log(\rho(\bmB_k\otimes\bmA_k))
=-0.010<0$. But it is neither first-order nor second-order stationary, as the spectral radii
of $\sum_{k=1}^2\alpha_k(\bmB_k\otimes\bmA_k)$ and $\sum_{k=1}^2\alpha_k\{(\bmB_k\otimes\bmA_k)
\otimes(\bmB_k\otimes\bmA_k)\}$ are all larger than 1. Fig.~\ref{Simu1200} shows a simulated dataset of size 1200 of Example 1.
\begin{figure}[ht!]
	\centering
	\includegraphics[width=.65\textwidth]{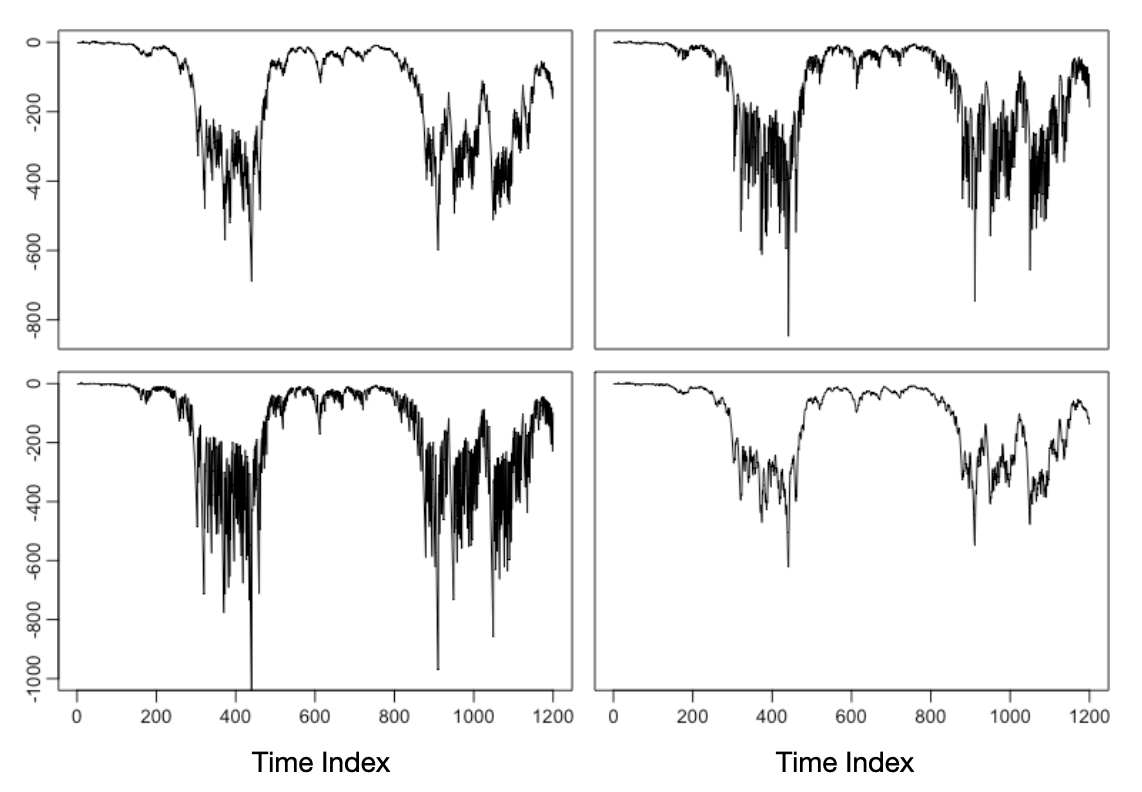}
	\caption[Simulated data of Example 1]{\label{Simu1200} Simulated data of Example 1.}
\end{figure}
We note that the overall model can be made second-order stationary by making minor adjustments to the example's parameters, while preserving the non-stationarity of the second component process. 

\section{Parameter Estimation}
\label{sec:para est}
Maximum likelihood estimation of the MMAR model can be implemented via an Expectation–Maximization (EM) algorithm \citep{dempster1977maximum}.
Let $\bmZ_t=(Z_{t,1},\dots,Z_{t,K})$ be the latent variable, such that $Z_{t,k}=1$ if $\bmY_t$ is from the $k$th component, and equals 0 otherwise. 
For simplicity, define
\begin{equation*}
	\mathcal{A}_k=\begin{pmatrix}
		\bmA_{k,1},\dots,\bmA_{k,p_k}
	\end{pmatrix},\quad	
	\mathcal{B}_k=\begin{pmatrix}
		\bmB_{k,1},\dots,\bmB_{k,p_k}
	\end{pmatrix},
\end{equation*}
and $\mathcal{Z}_{t-1,k}=\mathrm{Bdiag}(\bmY_{t-1},\dots,\bmY_{t-p_k})$, a $(p_km)\times (p_kn)$ block-diagonal matrix, with $\{\bmY_{t-1},\dots,\bmY_{t-p_k}\}$ comprising the diagonal blocks. 
The density of $(\bmY_t,\bmZ_t)$ given $\mathscr{F}_{t-1}$ is 
\begin{equation*}
	\prod_{k=1}^{K}\alpha_k^{Z_{t,k}}
	\left\{
	\frac{\exp\left(-\frac{1}{2}\tr{\bmV_k^{-1}	\bm{E}_{t,k}^\tp\bmU_k^{-1}
			\bm{E}_{t,k}}\right)
	}
	{(2\pi)^{mn/2}\det(\bmV_k)^{m/2}\det(\bmU_k)^{n/2}}
	\right\}^{Z_{t,k}}.
\end{equation*}
\textbf{E-step:} Let $\tau_{t,k}$ be the conditional expectation of the $Z_{t,k}$ given $\mathscr{F}_{t}$ and the current parameter value. Then
\begin{equation*}
	\tau_{t,k}=
	\frac{
		\alpha_k\det(\bmV_k)^{-m/2}\det(\bmU_k)^{-n/2}
		\exp\left(-\frac{1}{2}\tr{\bmV_k^{-1}\bm{E}_{t,k}^\tp\bmU_k^{-1}
			\bm{E}_{t,k}}\right)}
	{\sum_{j=1}^K	\alpha_j\det(\bmV_j)^{-m/2}\det(\bmU_j)^{-n/2}
		\exp\left(-\frac{1}{2}\tr{\bmV_j^{-1}	\bm{E}_{t,k}^\tp\bmU_j^{-1}
			\bm{E}_{t,k}}\right)
	}.
\end{equation*}
\textbf{M-step:} Update the estimates of $\alpha$'s as follows:
\begin{equation*}
	\hat{\alpha}_k=
	\frac{1}{T-p_{\max}}\sum_{t=p_{\max}+1}^T\tau_{t,k}.
\end{equation*}
The estimates of $\calA_k$, $\calB_k$, $\bmC_k$, $\bmU_k$ and $\bmV_k$ must satisfy the following gradient conditions:
\begin{align}
	&\calA_k=\left(\sum_{t=p_{\max}+1}^T\tau_{t,k}(\bmY_t-\bmC_k)\bmV_{k}^{-1}\calB_k\calZ_{t-1,k}^\tp
	\right)\left(\sum_{t=p_{\max}+1}^T\tau_{t,k}
	\calZ_{t-1,k}\calB_k^\tp\bmV_{k}^{-1}\calB_k\calZ_{t-1,k}^\tp\right)^{-1}\label{eq update A},
	\\
	&\calB_k=	\left(\sum_{t=p_{\max}+1}^T\tau_{t,k}
	(\bmY_t-\bmC_k)^\tp\bmU_{k}^{-1}\calA_k\calZ_{t-1,k}
	\right)\left(\sum_{t=p_{\max}+1}^T\tau_{t,k}
	\calZ_{t-1,k}^\tp\calA_k^\tp\bmU_{k}^{-1}\calA_k\calZ_{t-1,k}\right)^{-1},\\
	&\bmC_k=\frac{\sum_{t=p_{\max}+1}^T\tau_{t,k}
		(\bmY_t-\calA_{k}\calZ_{t-1,k}\calB_{k}^\tp)}{\sum_{t=p_{\max}+1}^T\tau_{t,k}},\\
	&\bmU_k=\frac{\sum_{t=p_{\max}+1}^T\tau_{t,k}(\bmY_t-\bmC_k-\calA_{k}\calZ_{t-1,k}\calB_{k}^\tp)\bmV_{k}^{-1}
		(\bmY_t-\bmC_k-\calA_{k}\calZ_{t-1,k}\calB_{k}^\tp)^\tp}{n\sum_{t=p_{\max}+1}^T\tau_{t,k}},\\
	&\bmV_k=\frac{\sum_{t=p_{\max}+1}^T\tau_{t,k}(\bmY_t-\bmC_k-\calA_{k}\calZ_{t-1,k}\calB_{k}^\tp)^\tp\bmU_{k}^{-1}
		(\bmY_t-\bmC_k-\calA_{k}\calZ_{t-1,k}\calB_{k}^\tp)}{m\sum_{t=p_{\max}+1}^T\tau_{t,k}}.
	\label{eq update V}
\end{align}
Closed-form solutions for these parameter estimates do not exist. However, the optimization problem in each M-step can be solved by a blockwise coordinate descent algorithm. To be specific, we use Eqns.~\eqref{eq update A} -- \eqref{eq update V}
to iteratively update one of $\{\mathcal{A}_k, \mathcal{B}_k, \bmC_k,\bmU_{k},\bmV_{k}\}$ with all of the others being fixed.
Note that the target function in each of the {M}-steps is multimodal, and the blockwise coordinate descent algorithm may converge to a local maximum. Due to the identifiability issues,  the estimated parameters are normalized such that constraints \eqref{iden_c1} and \eqref{iden_c1_1} are satisfied. 

The EM algorithm may converge to a local maximum. Nevertheless, given the intricate structure of the target function, numerous local maxima can exist, particularly in high-dimensional scenarios, making it necessary to repeat the process many times. The speed of the proposed EM algorithm could be very slow, as it involves an iterative process to find the maximum within each of the M-step. 

The likelihood function of a mixture model may be unbounded \citep{mclachlanfinite}, hence it may not have a global maximum. Nevertheless, the MLE that corresponds to a local maximum around the true value could be consistent, efficient and asymptotically normal under some regularity conditions \citep{peters1978iterative,redner1984mixture}.

In the longitudinal relational data, \cite{hoff2015multilinear} observed that two scalar time series can be positively correlated even if they are in different rows and columns.
This pattern is not limited to longitudinal relational data but is also observed in other matrix-valued time series datasets, such as the economic indicators dataset displayed in Fig.~\ref{fig_econo} and the simulated dataset shown in Fig.~\ref{Simu1200}. Further investigation in the simulations reveals that the correlations could also be negative. Note that a scalar time series extracted from a matrix-valued time series may reveal the underlying regime-switching behavior of the matrix-valued process. Motivated by this observation, we propose the following initialization procedure for the EM algorithm:
\begin{enumerate}[label=\arabic*)]
	\item[Step 1:] Select an arbitrary scalar time series from the matrix-valued time series data, and fit a scalar mixture autoregressive model with $K$ components. 
	\item[Step 2:] Use the fitted scalar mixture autoregressive model to segment the entire process into $K$ regimes, via a univariate version of Eqn.~\eqref{regime_decision} described in Section \ref{sec:empirical}.
	\item[Step 3:] For each regime so found, fit a VAR model with data in that regime. The  ML estimate of the fitted VAR model serves as the initial value for one component of the MMAR model. The relative frequency of data in that regime is the initial value of the corresponding probability mixing weight.
\end{enumerate}
This procedure can be repeated multiple times to implement the EM algorithm with different sets of initial values, and the estimate yielding the highest likelihood at convergence will be taken as the ML estimate of the MMAR model. Step 1 for the univariate time series is carried out via the EM algorithm, with equal initial mixing weights and randomly generated initial other parameter values. The process is repeated multiple times (20 in our simulations), and the one resulting in the largest log-likelihood is selected, which will be used for segmentation. All numerical results reported below are obtained using the preceding initialization scheme.  

\section{Large Sample Properties}
\label{sec: asym}

In this section, we derive the large-sample properties of the maximum likelihood estimator (MLE)  of the MMAR model. As mentioned earlier, the MMAR model is a constrained mixture VAR($p$) model. For conciseness, within this section, we assume identical VAR orders $p_1=\cdots=p_K=p$, across the mixing VAR components. We consider matrix-valued time series, hence it is assumed that $\min(n,m)\ge 2$.  Combine all the AR coefficient matrices of the $k$th mixing VAR component into $\bm{\Psi}_k=(\bm{\Psi}_{k,1}, \cdots, \bm{\Psi}_{k,p})$, thereby the conditional pdf  is given by
\begin{equation}
	f_t(\bmy_{t}|\mathscr{F}_{t-1};\bm{\bm{\Delta}})=\sum_{k=1}^K 
	\alpha_k f_\mathcal{N}\left(\bmy_{t}|
	\bm{\Psi}_{k,0}+\bm{\Psi}_k \mathcal{X}_{t-1},
	\bm{\Omega}_k\right), \label{mVAR density}
\end{equation}
where $\bmy_t=\vecc{\bmY_t}$, $\bm{\Delta}= \left(\vecc{\bm{\Delta}_1}^\top, \cdots, \vecc{\bm{\Delta}_K}^\top, \alpha_1,\cdots, \alpha_{K-1}\right)^\top$,  $\bm{\Delta}_k=\left(\vecc{\bm{\Psi}_k}^\tp,\right.$ $\left.\vecc{\bm{\Psi}_{k,0}}^\tp, \vech{\bm{\Omega}_k}^\tp\right)^\tp$,    for $k=1,\cdots, K$, $\mathcal{X}_{t-1}=(\bmy_{t-1}^\top, \cdots, \bmy_{t-p}^\top)^\top$,  $\alpha_K=1-\sum_{k=1}^{K-1} \alpha_k$, and  the $\bm{\Omega}$'s are assumed to be positive definite. Below, we sometimes write $f_t(\bmy_{t}|\mathscr{F}_{t-1};\bm{\bm{\Delta}}_k)$ for $ f_\mathcal{N}\left(\bmy_{t}|\bm{\Psi}_{k,0}+\bm{\Psi}_k, \mathcal{X}_{t-1},\bm{\Omega}_k\right)$. 

The MMAR model imposes the Kronecker product representation constraints \eqref{eq: basic parametric constraints}. Thus,  $\bm{\Delta}_k$ are  functions of $\bmtt_k$ where \begin{align*}
	\bmtt_k=&\big(\vecc{\bmA_{k,1}}^\tp, \vecc{\bmB_{k,1}^\top}^\tp,
	\dots,\vecc{\bmA_{k,p}}^\tp,\\ &\vecc{\bmB_{k,p}^\top}^\tp,
	\vecc{\bmC_k}^\tp,\vech{\bmU_k}^\tp, \vech{\bmV_k}^\tp\big)^\tp,
\end{align*} 
where the parameter space of $\bm{\theta}_k$ is assumed to be an open subset of some Euclidean space. 
Indeed they are differentiable functions, with their derivatives obtained as follows: Suppose $\bm{\Psi}=\bm{B}\otimes \bm{A}$ with $\bm{A}$ and $\bm{B}$ are of dimensions $m\times m$ and $n\times n$, respectively. Then, 
\begin{equation}
	\frac{\partial \vecc{\bm{\Psi}}}{\partial\vecc{\bm{B}}^\tp}=(\bm{I}_n\otimes\bm{K}_{mn}\otimes \bm{I}_m)
	(\bm{I}_{n^2}\otimes \vecc{\bm{A}}),
	\label{dPsidB}
\end{equation}
where $\bm{I}_n$ is the identity matrix of dimension $n\times n$ and $\bm{K}_{mn}$ is the commutation matrix which converts the vectorized form of an $m\times n$ matrix to the vectorized form of its transpose, i.e, $\bm{K}_{mn}\vecc{\bm{M}}=\vecc{\bm{M}^\tp}$ for any $m\times n$ matrix $\bm{M}$, c.f., \cite[Theorem 11]{magnus1985matrix}. (Note  that $\bmI_m$ and $\bmI_r$ in  Eqn.~(14) there  should read as $\bmI_{mp}$ and $\bmI_{rs}$, respectively, as can be deduced from the proof  there.) Furthermore, we have 
\begin{equation}
	\frac{\partial \vecc{\bm{\Psi}}}{\partial\vecc{\bm{B}^\top}^\tp}=(\bm{I}_n\otimes\bm{K}_{mn}\otimes \bm{I}_m)
	(\bm{I}_{n^2}\otimes \vecc{\bm{A}})\bm{K}_{nn}.
	\label{dPsidBtop}
\end{equation}
Similarly,
\begin{eqnarray}
	\frac{\partial \vecc{\bm{\Psi}}}{\partial\vecc{\bm{A}}^\tp}&=&(\bm{I}_n\otimes\bm{K}_{mn}\otimes \bm{I}_m)(\vecc{\bm{B}}\otimes\bm{I}_{m^2})
	\label{dPsidA} 
\end{eqnarray}
Consider $\bm{\Omega}=\bm{V}\otimes \bm{U}$ where $\bm{\Omega}, \bm{V}, \bm{U}$ are symmetric matrices of dimensions $mn\times mn, n\times n, m\times m$ respectively.
Let  $\bm{\mathscr{D}}_p$ be the duplication matrix  such that for any $p\times p$ symmetric matrix $\bm{M}$, $\vecc{\bm{M}}=\bm{\mathscr{D}}_p \vech{\bm{M}}$.  As $\bm{\mathscr{D}}_p$ is of full-rank, $\vech{\bm{M}}=\bm{\mathscr{L}}_p \vecc{\bm{M}}$ where $\bm{\mathscr{L}}_{p}=(\bm{\mathscr{D}}_p^\top \bm{\mathscr{D}}_p)^{-1} \bm{\mathscr{D}}_p^\top$. Similarly,
\begin{eqnarray}
	\frac{\partial \vech{\bm{\Omega}}}{\partial\vech{\bm{V}}^\tp}&=&\bm{\mathscr{L}}_{mn}(\bm{I}_n\otimes\bm{K}_{mn}\otimes\bm{I}_m)(\bm{I}_{n^2}\otimes\vecc{\bm{U}})\bm{\mathscr{D}}_{n}
	\label{dOmegaV} \\
	\frac{\partial \vech{\bm{\Omega}}}{\partial\vech{\bm{U}}^\tp}&=&\bm{\mathscr{L}}_{mn}(\bm{I}_n\otimes\bm{K}_{mn}\otimes\bm{I}_m)(\vecc{\bm{V}}\otimes \bm{I}_{m^2})\bm{\mathscr{D}}_{m}.
	\label{dOmegaU} 
\end{eqnarray} 
Routine algebra shows that 
\begin{eqnarray}
	\frac{\partial \vecc{\bm{\Psi}}}{\partial\vecc{\bm{B}^\top}^\tp}
	\vecc{\bm{B}^\top} - \frac{\partial \vecc{\bm{\Psi}}}{\partial\vecc{\bm{A}}^\tp}\vecc{\bm{A}} &\equiv& \bm0. \label{eq:dependence}\\\frac{\partial \vech{\bm{\Omega}}}{\partial\vech{\bm{V}}^\tp}\vech{\bm{V}} -\frac{\partial \vech{\bm{\Omega}}}{\partial\vech{\bm{U}}^\tp}\vech{\bm{U}}&\equiv& \bm0. 
	\label{eq:dependence-cov} 
\end{eqnarray}
It is readily seen that if  $\bm{A}\not=\bm0$, then $ {\partial \vecc{\bm{\Psi}}}/{\partial\vecc{\bm{B}^\top}^\tp}$ is of full-rank. Otherwise, it is of zero rank. Similarly, $ {\partial \vecc{\bm{\Psi}}}/{\partial\vecc{\bm{A}}^\tp}$ is of full-rank (zero-rank) if $\bm{B}\not = \bm0$ ($\bm{B} = \bm0$). Similar discussion applies to the partial derivative of $\bm{\Omega}$ w.r.t. $\bm{U}$ and that w.r.t. $\bm{V}$

Consequently, $\bm{\Delta}$ is a differentiable  function of $\bm{\theta}=(\bm{\theta}_1^\top,\cdots, \bm{\theta}_K^\top, \alpha_1, \cdots, \alpha_{K-1})^\top$. Let $\bm{J}=\frac{\partial \bm{\Delta}}{\partial \bm{\theta}^\top}$. Because of \eqref{eq:dependence} and \eqref{eq:dependence-cov}, $\bm{J}$  is of reduced rank and its rank cannot exceed dim$(\bm{\theta})-(p+1)K$, due to one such linear dependence relationship for each $\bm{\Psi}_{k,j}$ and each $\bm{\Omega}_k$. Below, we show that under some mild regularity conditions, the rank of $\bm{J}$ equals dim$(\bm{\theta})-(p+1)K$.

However $\bm{\Delta}(\bm{\theta})$ is generally a many-to-one function. As discussed earlier, under conditions \eqref{iden_c1_0}, \eqref{iden_c1}  and \eqref{iden_c1_1},  it is a one-to-one function, which we now verify. 
The recovery of $\bm{\theta}$ from $\bm{\Delta}$ can be illustrated in the setting of a non-zero $\bm{\Psi}$ that admits the Kronecker product representation: $\bm{\Psi}=\bm{B}\otimes \bm{A}$, where  $\bm{A}$ is $m\times  m$ and $\bm{B}$ is $n\times n$, $\|\bm{B}\|_F=1$, and the first non-zero element of $\vecc{\bm{B}}$ is positive.
Let $\bm{B}=(b_{u,v})$ and $\bm{A}=(a_{i,j})$.
Let $a_{i_0,j_0}$  and $b_{u_0,v_0}$ be the first non-zero entries of $\vecc{\bm{A}}$ and   $\vecc{\bm{B}}$, respectively, hence the first non-zero entry of   $\bm{\Psi}$ is $\psi_{(u_0-1)m+i_0, (v_0-1)m+j_0}$.  Now, $\psi_{(u-1)m+i_0, (v-1)m+j_0}=b_{u,v} a_{i_0,j_0}$, for all $1\le u,v \le n$. Therefore, $|a_{i_0, j_0}|=(\sum_{1\le u,v\le n} \psi_{(u-1)m+i_0, (v-1)m+j_0}^2)^{1/2}$ hence for $1\le u, v \le n$, $b_{u,v}=\mbox{sign}(a_{i_0,j_0}) \psi_{(u-1)m+i_0, (v-1)m+j_0}/(\sum_{1\le i, j\le n} \psi_{(i-1)m+i_0, (j-1)m+j_0}^2)^{1/2}$. Since $b_{u_0,v_0}$ is the first non-zero element in $B$, it must be positive, which implies that $\mbox{sign}(a_{i_0,v_0})=\mbox{sign}(\psi_{(u_0-1)m+i_0, (v_0-1)m+j_0})$, hence 
$$b_{u,v}=\frac{\mbox{sign}(\psi_{(u_0-1)m+i_0, (v_0-1)m+j_0}) \psi_{(u-1)m+i_0, (v-1)m+j_0}}{(\sum_{1\le i, j\le n} \psi_{(i-1)m+i_0, (j-1)m+j_0}^2)^{1/2}}.$$
With $\bm{B}$ recovered, $\bm{A}$ can be readily recovered by inverting the equation $\psi_{(u_0-1)m+i,(v_0-1)m+j}=b_{u_0,v_0} a_{i,j}$, for 
$1\le i, j\le m$, resulting in
$$
a_{i,j}=\psi_{(u_0-1)m+i,(v_0-1)m+j}/b_{u_0,v_0}.
$$
It is clear that these local one-to-one  relationships are continuous and differentiable iff $\bm{\Psi}\not =\bm0$. Consequently the collection of the $\bm{\Psi}$ having the Kronecker product representation constitutes a manifold of dimension $(n^2-1)+m^2$ around any non-zero $\bm{\Psi}$,
since $\vecc{\bm{B}}$ is constrained to lie on the unit sphere in the $n^2$-dimensional Euclidean space and $\vecc{\bm{A}}$ is an arbitrary vector in the $m^2$-dimensional Euclidean space.  

Condition on $\mathscr{F}_{t-1}$, the log-likelihood function of the MMAR can be expressed as the constrained log-likelihood of  the MAR model, specifically, $L_T(\bm{\theta})=L_T\{\bm{\Delta}(\bm{\theta})\}$ where the unconstrained MAR log-likelihood equals 
$L_T(\bm{\Delta})=\sum_{t=p+1}^Tl_t(\bm{\Delta})$,
with  $$
l_t(\bm{\Delta})=\log\left(\sum_{k=1}^K\alpha_kf_t(\bmy_{t}|\mathscr{F}_{t-1};\bm{\Delta}_k)\right).
$$
Let $\bmtt^0$ be the true parameter.

The first derivative of the log-likelihood equals
$\frac{\partial L_T(\bm{\theta})}{\partial \bm{\theta}^\top}=\sum_t \frac{\partial l_t(\bm{\theta})}{\partial \bm{\theta}^\top}$ where 
$\frac{\partial l_t(\bm{\theta})}{\partial \bm{\theta}^\top}= \frac{\partial l_t(\bm{\Delta})}{\partial \bm{\Delta}^\top} \frac{\partial \bm{\Delta}}{\partial \bm{\theta}^\top}$. Recall $\tau_{t,k}=\alpha_k f_t\left(\bmy_{t}|
\mathscr{F}_{t-1};\bm{\Delta}_k
\right)/ f_t(\bmy_{t}|\mathscr{F}_{t-1};\bm{\bm{\Delta}})$ is the probability that $\bmy_t$ is generated from the $k$th component and 
$\bm{\epsilon}_{t,k}=\bmy_t-\bm{\Psi}_{k,0}-\sum_{j=1}^p \bm{\Psi}_{k,j}\bmy_{t-j}$. 
Note that 
$\frac{\partial l_t}{\partial \bm{\Delta}^\top}=\left(\frac{\partial l_t}{\partial 
	\vecc{\bm{\Psi}_{1,1}}^\top}, \cdots,\frac{\partial l_t}{\partial 
	\vecc{\bm{\Psi}_{1,p}}^\top}, \right. $ $
\left. \frac{\partial l_t}{\partial 
	\vecc{\bm{\Psi}_{1,0}}^\top}, \frac{\partial l_t}{\partial 
	\vech{\bm{\Omega}_1}^\top}, \cdots,\frac{\partial l_t}{\partial 
	\vecc{\bm{\Psi}_{K,1}}^\top},\cdots, \frac{\partial l_t}{\partial 
	\vech{\bm{\Omega}_K}^\top}, \frac{\partial l_t}{\partial \alpha_1},\cdots,  \frac{\partial l_t}{\partial \alpha_{K-1}}\right)
$
where,  $\forall k,j$
\begin{eqnarray}
	\frac{\partial l_t}{\partial \vecc{\bm{\Psi}_{k,0}}^\top} &=& \tau_{t,k} \bm{\epsilon}_{t,k}^\top\bm{\Omega}_{k}^{-1} \label{eq:lt-d0} \\
	\frac{\partial l_t}{\partial 
		\vecc{\bm{\Psi}_{k,j}}^\top} &=& \tau_{t,k} \bmy_{t-j}^\top \otimes\bm{\epsilon}_{t,k}^\top\bm{\Omega}_{k}^{-1} \label{eq:lt-dkj}    \\
	\frac{\partial l_t}{\partial 
		\vech{\bm{\Omega}_k}^\top} &=& -1/2 \times \tau_{t,k} \vech{\bm{\Omega}_k^{-1} -\bm{\Omega}_k^{-1}\bm{\epsilon}_{t,k}\bm{\epsilon}_{t,k}^\top\bm{\Omega}_k^{-1}}^\top\bm{\mathscr{D}}_{mn}^\top \bm{\mathscr{D}}_{mn}, \nonumber \\
	\label{eq:lt-Omega}
\end{eqnarray}
and for $k=1,\cdots, K-1$,
\begin{eqnarray}
	\frac{\partial l_t}{\partial \alpha_k}&=& \frac{f_t\left(\bmy_{t}|
		\mathscr{F}_{t-1};\bm{\Delta}_k
		\right)-f_t\left(\bmy_{t}|
		\mathscr{F}_{t-1};\bm{\Delta}_K
		\right)}{f_t(\bmy_{t}|\mathscr{F}_{t-1};\bm{\bm{\Delta}})}.
	\label{eq:lt-alpha}
\end{eqnarray}
The following condition is a commonly used  identifiability constraints for $\bm{\Delta}$: 
\vskip .1 in
\noindent  ({\bf C1}) $K$ is known, the validity of \eqref{iden_c2}, i.e., the mixing probabilities $\alpha_k$'s are positive and in ascending order, and that the $\bm{\Delta}_k$'s are distinct. 
\vskip .1 in
\noindent Under ({\bf C1}), the true 
$\bm{\Delta}^0\in \bm{\Xi}$ is unique.  It is readily checked that the score vector evaluated at $\bm{\Delta}^0$ is a martingale difference sequence of zero mean, i.e., $\E{\frac{\partial l_t(\bm{\Delta}^0)}{\partial \bm{\Delta}}|\mathscr{F}_{t-1}}=\bm0$ where $\E{\cdot}$ is always evaluated under the true model, hence it is of zero mean. Furthermore, we claim that the Fisher information matrix $$\mathcal{I}(\bm{\Delta}^0)= \E{\frac{\partial l_t(\bm{\Delta}^0)}{\partial \bm{\Delta}}\frac{\partial l_t(\bm{\Delta}^0)}{\partial \bm{\Delta}^\top}}$$ is positive definite; see \emph{Supplementary Materials}. 

Define $\ell(\bm{\Delta})=\mathbb{E}\{l_t(\bm{\Delta})\}.$  Assuming the true model is stationary, with finite fourth stationary moments,  the usual trick enabled by the interchange of differentiation and expectation yields 
\begin{equation}
	\ell(\bm{\Delta})= \ell(\bm{\Delta}^0)   -\frac{1}{2}(\bm{\Delta}-\bm{\Delta}^0)^\top \mathcal{I}(\bm{\Delta}^0)(\bm{\Delta}-\bm{\Delta}^0)+o(|\bm{\Delta}-\bm{\Delta}^0|^2).
	\label{eq:M-asym1}
\end{equation}
Furthermore, 
\begin{equation}
	l_t(\bm{\Delta})-l_t(\bm{\Delta}^0)= \frac{\partial l_t(\bm{\Delta}^0)}{\partial \bm{\Delta}^\top}(\bm{\Delta}-\bm{\Delta}^0) +|\bm{\Delta}-\bm{\Delta}^0|r_t(\bm{\Delta})
	\label{eq:M-asym2}
\end{equation}
where it follows from the mean-value theorem that $r_t(\bm{\Delta})=0$ if $\bm{\Delta}=\bm{\Delta}^0$ and otherwise there exists a $\bm{\Delta}^*$ such that (i) $|\bm{\Delta}^*-\bm{\Delta}^0|\le |\bm{\Delta}-\bm{\Delta}^0|$ and (ii)  $r_t(\bm{\Delta})= \left(\frac{\partial l_t(\bm{\Delta}^*)}{\partial \bm{\Delta}^\top}- \frac{\partial l_t(\bm{\Delta}^0)}{\partial \bm{\Delta}^\top}\right)(\bm{\Delta}-\bm{\Delta}^0)/|\bm{\Delta}-\bm{\Delta}^0|$. 
By the martingale central limit theorem, $\frac{1}{\sqrt{T}}\sum_{t=p+1}^T \frac{\partial l_t(\bm{\Delta}^0)}{\partial \vecc{\bm{\Delta}}}$ converges to $\mathcal{N}(\bm{0}, \mathcal{I}(\bm{\Delta}^0)) $ in distribution as $T\to\infty$. 

The study of the large-sample properties of the ML estimator of the MMAR model is complicated by the fact that the MMAR model is not identifiable for two  reasons. First, the (conditional) normal components need to be distinct distributions otherwise the model is non-identifiable. Even under the distinct component assumption, the component labels are only identifiable up to permutation but it is more a nuisance than a real obstacle. Second, the Kronecker product representations of the VAR coefficient matrices and that of the innovation covariance matrix per each mixing component implies that the model is over-parameterized. Consequently, the same MMAR model admits multiple parametric representations. The $\bm{\theta}$'s can be partitioned into equivalence classes of parameters indexing the same MMAR model, through the equivalence relationship that $\bm{\theta}\sim \bm{\theta}'$ if and only if these two parameter vectors index the same MMAR model.   Let $[\bm{\theta}]$ denote the equivalence class to which $\bm{\theta}$ belongs, and $[\bm{\theta}^0]$ the set of parameters that index the true model. The model non-identifiability of the MMAR model can be removed by imposing suitable identifiability constraints such as \eqref{iden_c1_0}--\eqref{iden_c33}.  However, these identifiability constraints are not unique and their choice is subjective.  More importantly, deriving the large-sample properties of the estimator under a specific set of constraints may muddle both the theoretical investigation and the general feature of the asymptotic properties of the estimator. A more elegant approach is to leverage the constrained M-estimation framework \citep{geyer1994asymptotics} to   first derive the large-sample distribution of the constrained VAR estimates, i.e., the constrained MLE of $\bm{\Psi}$'s and $\bm{\Sigma}$'s, based on which the large-sample distribution of the estimate of any estimable functions of $\bm{\theta}$, including those of the constrained  $\bm{\theta}$ satisfying a set of identifiability constraints can be more readily derived; see the related work by \cite{shapiro1986asymptotic} for over-parameterized structural equation model estimation. 

We begin with deriving the strong consistency of the ML estimator of the MMAR model. \cite{redner1981note} has shown that under  mild regularity conditions, the MLE of the (non-identifiable) mixture multivariate normal model with independent and identically distributed data is strongly consistent in the sense that the MLE converges almost surely to the equivalence class of the true model. Because of the non-identifiability, the MLE $\hat{\bmtt}$ is any $\bm{\theta}$ that attains the maximum likelihood over the parameter space. Below we show the strong consistency of the MLE for the MMAR model. Notably the consistency results only requires known upper bounds for $K$ and $p$, hence the ML estimator is robust against over-specification of $K$ or $p$. 

\begin{Theorem}
	\label{mmar_con}
	Assume that 
	(i)  the true MMAR model is strictly stationary and ergodic, whose stationary distribution admits finite second-order moments and  (ii) $K$ is a known upper bound for the number of distinct normal components and $p$ is a known upper bound for the AR orders and (iii) the parameter space $\bm{C}\subset \bm{\Theta}$ is a compact set containing the true parameter $\bmtt^0$ (in other words,  $\bm{C}\cap [\bmtt^0] \not = \emptyset$). 
	Then, the MLE $\hat{\bmtt}$
	is a strongly consistent estimator of the true parameter $\bmtt^0$, i.e., $\hat{\bmtt}\to \bm{C}\cap [\bm{\theta}^0]$ a.s. or, in other words,  the Euclidean distance between $\hat{\bmtt}$ and $\bm{C}\cap [\bm{\theta}^0]$ approaches zero as $T\to\infty$ a.s.
\end{Theorem}

\begin{Remark}
	Recall under ({\bf C1}), the true 
	$\bm{\Delta}^0\in \bm{\Xi}$ is unique.
	Moreover, $\bm{\Delta}=\bm{\Delta}(\bm{\theta})$ is the same for all $\bm{\theta} \in [\bm{\theta}]$. Hence if assumptions (i)-(iii) hold and (C1) is valid,  the ML estimator $\hat{\bm{\Delta}}$ is strongly consistent. 
	
	\cite[Theorem 6]{redner1981note}  shows that under stronger conditions on $\bm{C}$ including prior information about the disparity between the $\bm{\Delta}_k$'s and a lower bound on the smallest mixing probability, the number of distinct mixing components can be consistently estimated by the MLE. This result can be similarly generalized to the MMAR model, but will not be pursued here as the needed prior information is generally unavailable in practice. 
\end{Remark}

The next result shows that under some mild regularity conditions, $\sqrt{T}(\hat{\bm{\Delta}}-\bm{\Delta}^0)$ is asymptotically normal. 

\begin{Theorem}\label{Thm: 2}
	Assume that \newline 
	(i)  the true MMAR model is strictly stationary and geometric ergodic, whose stationary distribution admits finite fourth-order moments, 
	\newline 
	(ii) the number of distinct normal components $K$ is known, the mixing probabilities $0<\alpha_1<\cdots<\alpha_K<1$,  the true $\bm{\Delta}_k$'s are distinct and none of the true $\bm{\Psi}_{k,j}$'s equal to the zero matrix, and \newline (iii) the parameter space $\bm{C}\subset \bm{\Theta}$ is a compact set that contains an open ball centered at the true $\bm{\theta}^0$, where $\bm{\theta}^0$ is one of $\bm{\theta}$'s that index the true model.  
	
	Then the constrained  MLE $\hat{\bm{\Delta}}=\bm{\Delta}(\hat{\bmtt})$ is asymptotically normal, specifically, as $T\to \infty$, $\sqrt{T}(\hat{\bm{\Delta}} -\bm{\Delta}^0)$ converges in distribution to $\mathcal{N}(\bm{0}, \bm{J}(\bm{J}^\top \mathcal{I}(\bm{\Delta}^0)\bm{J})^{-} \bm{J}^\top)$, where 
	$\bm{J}=\frac{\partial \bm{\Delta}}{\partial \bm{\theta}^\top}$ evaluated at $\bm{\theta}^0$, and for any matrix $\bm{M}$, $\bm{M}^{-}$ denotes its generalized inverse. 
	Also, $2\{L_T(\hat{\bm{\Delta}})-L_T(\bm{\Delta}^0)\}$ 
	is asymptotically $\chi^2$ distributed with degrees of freedom equal to the rank of $\bm{J}$. 
\end{Theorem}
\begin{Remark}
	The preceding asymptotic normal distribution is of the same form as that given in \citep[Proposition 4.1]{shapiro1986asymptotic} upon  (i) noting that the $\bm{\Delta}$ and $\bm{V}$ there equal $\bm{J}$ and $\mathcal{I}(\bm{\Delta}^0)$, respectively  and (ii)  setting $\bm{\Gamma}$ there to be $\mathcal{I}^{-1}(\bm{\Delta}^0)$. As shown in the proof, the limiting normal distribution for $\hat{\bm{\Delta}}$ is the same for any $\bm{\theta}^0\in [\bm{\theta}^0]$. Furthermore, the rank of 
	$\bm{J}(\bm{J}^\top \mathcal{I}(\bm{\Delta}^0)\bm{J})^{-} \bm{J}^\top$ equals that of $\bm{J}$, because $\mathcal{I}(\bm{\Delta}^0)$ is positive definite. Note $\bm{J}^\top \mathcal{I}(\bm{\Delta}^0)\bm{J}=\mathcal{I}(\bm{\theta}^0)$ which is the covariance matrix of $\frac{\partial l_t}{\partial \bm{\theta}}$ evaluated at $\bm{\theta}^0$. The following formulas, which can be derived by routine matrix calculus,  are handy for computing $\mathcal{I}(\bm{\theta}^0)$.
	\begin{eqnarray}
		\frac{\partial l}{\partial \vecc{\bm{A}_{k,j}}} &=& \tau_{t,k} (\bmY_{t-j} \bm{B}_{k,j}^\top \otimes \bmI_m)\bm{\Omega}_k^{-1} \bm{\epsilon}_{t,k} \label{partial-l-A} \\
		\frac{\partial l}{\partial \vecc{\bm{B}^\top_{k,j}}} &=& \tau_{t,k} (\bmI_n \otimes \bmY_{t-j}^\top \bm{A}_{k,j}^\top)\bm{\Omega}_k^{-1} \bm{\epsilon}_{t,k}  \label{partial-l-B}\\
		\frac{\partial l}{\partial \vech{\bm{V}_k}} &=&-\frac{1}{2} \bm{\mathscr{D}}_n^\top \bm{\mathscr{D}}_n \vech{\bm{V}_k^{-1}- \bm{V}_k^{-1} \bm{E}_{t,k}^\top \bmU_k^{-1}  \bm{E}_{t,k}\bm{V}_k^{-1}} \label{partial-l-V}\\
		\frac{\partial l}{\partial \vech{\bm{U}_k}} &=&-\frac{1}{2} 
		\bm{\mathscr{D}}_m^\top \bm{\mathscr{D}}_m \vech{\bm{U}_k^{-1}- \bm{U}_k^{-1} \bm{E}_{t,k} \bmV_k^{-1}  \bm{E}_{t,k}^\top\bm{U}_k^{-1}}.
		\label{partial-l-U}
	\end{eqnarray}
\end{Remark}
Next, we derive the asymptotic distribution of the constrained MLE of $\bm{\theta}$ subject to a set of identifiability constraints. 
\begin{Theorem} \label{Thm: ident}
	Suppose assumptions (i)--(iii) in Theorem~\ref{Thm: 2} hold. Furthermore, let $c_i(\bm{\theta})=0, i=1,\cdots,r$ be identifiability constraints on $\bm{\theta}$, where  $r=(p+1)K$. Let $\bm{c}(\bm{\theta})=(c_1(\bm{\theta}), \cdots, c_r(\bm{\theta}))^\top$. 
	Let $\bm{W}=\frac{\partial \bm{c}(\bm{\theta}^0)}{\partial \bm{\theta}^\top}$. Assume that under the identifiability constraints, $\bm{\Delta}(\bm{\theta})$ is a diffeomorphic function of $\bm{\theta}$, i.e., they are differentiable reparameterizations of each other. Then, the constrained MLE $\hat{\bm{\theta}}$ subject to the constraints $\bm{c}(\hat{\bm{\theta}})=0$ is asymptotic normal, specifically, $\sqrt{T}(\hat{\bm{\theta}}-\bm{\theta}^0)$ converges weakly to  $\mathcal{N}(\bm0,\bm{P})$, as $T\to\infty$, where $\bm{P}$ is the upper left block of
	$$\begin{pmatrix}
		\bm{P} & \bm{Q}^\top \\
		\bm{Q} & \bm{R}
	\end{pmatrix}= \begin{pmatrix} \bm{J}^\top  \mathcal{I}(\bm{\Delta}^0) \bm{J} +\bm{W}^\top \bm{W}& \bm{W}^\top \\
		\bm{W} & \bm0 
	\end{pmatrix}^{-1}.$$ 
	It is implicitly assumed that the matrix inverse on the RHS of the preceding equation exists. 
\end{Theorem}
\begin{Remark}\label{remark5}
	Corresponding to \eqref{iden_c1} and \eqref{iden_c1_1}, the identifiability constraints $\bm{c}(\bm{\theta})$ consist of $\|\bm{B}_{k,j}\|_F-1$ and $\|\bm{V}_k\|_F-1$, for $j=1,\cdots, p$ and $k=1,\cdots K$. Note that $\frac{\partial \|\bm{B}_{k,j}\|_F-1}{\partial \vecc{\bm{B}_{j,k}}}= \vecc{\bm{B}_{j,k}}$ and  $\frac{\partial \|\bm{V}_{k,j}\|_F-1}{\partial \vech{\bm{V}_{k}}}= \vech{\bm{V}_{k}}$.  Thus, $\bm{W}$ is a $(p+1)K \times$ dim$(\bm{\theta})$ matrix  whose row vector corresponding to the constraint $\|\bm{B}_{k,j}\|=1$ ($\|\bm{V}_k\|=1$) is a zero vector except its $\vecc{\bm{B}_{k,j}}$ ($\vech{\bm{V}_k}$) component equals $\vecc{\bm{B}_{k,j}}$ ($\vech{\bm{V}_k}$). Let $\tilde{\bm{W}}$ equal $\bm{W}$ except that for the row vector corresponding $\|\bm{B}_{k,j}\|=1$ ($\|\bm{V}_k\|=1$), its $\vecc{\bm{A}_{k,j}}$ ($\vech{\bm{U}_k}$) component equals $-\vecc{\bm{A}_{k,j}}$ ($-\vech{\bm{U}_k}$). Then $\bm{J}\tilde{\bm{W}}^\top=0$, due to \eqref{eq:dependence} and \eqref{eq:dependence-cov}. Let $\bm{H}=\bm{J}^\top  \mathcal{I}(\bm{\Delta}^0) \bm{J} +\bm{W}^\top \bm{W}$. Then, $\bm{H}\tilde{\bm{W}}^\top=\bm{W}^\top$. It is readily seen that $\bm{W}\tilde{\bm{W}}^\top=\bm{I}$. Now, $\bm{P}=\bmH^{-1}+\bmH^{-1}\bm{W}^\top(\bm{0}-\bm{W}\bm{H}^{-1} \bm{W}^\top)^{-1}\bm{W}\bm{H}^{-1}= \bmH^{-1}-\bmH^{-1}\bm{W}^\top(\bm{W}\tilde{\bm{W}}^\top)^{-1}\bm{W}\bm{H}^{-1}=\bm{H}^{-1}\bm{J}^\top  \mathcal{I}(\bm{\Delta}^0) \bm{J} \bm{H}^{-1}$ which reduces to 
	Theorem 3 of \cite{chen2021autoregressive} for the MAR(1) model with no intercept term, on noting that they adopted the constraint $\|\bm{A}\|=1$. Furthermore, the covariance matrix of the asymptotic normal distribution of $\sqrt{T}\bm{W}(\hat{\bm{\theta}}-\bm{\theta}^0)$ is the zero matrix, hence $\bm{W}(\hat{\bm{\theta}}-\bm{\theta}^0) =o_p(T^{-1/2})$. In particular, 
	\begin{equation}
		\vecc{\bm{B}_{k,j}^0}^\top (\vecc{\hat{\bm{B}}_{k,j}}-\vecc{\bm{B}_{k,j}^0})=0
		\label{eq:estimate constraint}
	\end{equation}
	up to an error of $o_p(T^{-1/2})$.
\end{Remark}

\section{Model Selection}
\label{sec: selection}
In this section, we discuss methods for selecting the number of components $K$ and the AR orders $(p_1,\dots,p_K)$. Although the asymptotic distribution of the MLE has been derived in the previous section, it remains challenging to implement likelihood based tests to select $K$, such as the Wald test, the score test, and the likelihood-ratio test. This is because these tests for the MMAR model contain nuisance parameters, which are absent under the null hypothesis \citep[see, e.g.,][]{davies1987hypothesis,chan1990likelihood}. Even if $K$ is given and the AR orders are to be selected, the challenges of implementing these tests persist due to some identifiability issues under the null hypothesis.

Therefore, we resort to using information criteria for model selection. The following criteria are taken into consideration: the Akaike information criterion (AIC), the Bayesian information criterion (BIC) and the Hannan–Quinn (HQ) information criterion, which are defined as,
\begin{align*}
	\text{AIC} &=
	-2L_T(\hat{\bmtt})+ 2\cdot\mathrm{dim}(\Theta),\\
	\text{BIC} &=
	-2L_T(\hat{\bmtt})+ \log(T-p_{\max})\cdot \mathrm{dim}(\Theta),\\
	\text{HQ} &=
	-2L_T(\hat{\bmtt})+2\log(\log(T-p_{\max}))\cdot \mathrm{dim}(\Theta),
\end{align*}
where $\Theta$ is the parameter space of $\bm{\theta}$ subject to a set of identifiability constraints.  
In addition, we consider the generalized information criterion (GIC), which was proposed by \cite{nishii1984asymptotic} for model selection in linear regressions. The GIC is given by,
\begin{equation*}\text{GIC} = -2L_T(\hat{\bmtt})+\nu_T\cdot \mathrm{dim}(\Theta),
\end{equation*}
where $\nu_T>0$ is a sequence such that $\lim_{T\to\infty}\nu_T=\infty$ and $\lim_{T\to\infty}\nu_T/T=0$. Obviously, both the BIC and the HQ are special cases of GIC. In our studies, we consider a particular GIC with
\begin{equation*}
	\nu_T=\log(\log(T-p_{\max}))\log(dim(\Theta)),
\end{equation*}
which has also been explored by \cite{meng2022penalized}. Empirical results reported by  \cite{wong2000mixture} and \cite{fong2007mixture} showed that for mixture autoregressive models the AIC is not suitable for selecting the number of components while the BIC is recommended. Since the theoretical properties of these information criteria for the MMAR model are unknown, simulations are used to check their performance in selecting both the number of mixture components $K$ and the AR orders.

The conditional expectation of $\bmY_{t}|\mathscr{F}_{t-1}$ can be used for prediction, which is defined as
\begin{equation*}
	\E{\bmY_{t}|\mathscr{F}_{t-1}}
	=\sum_{k=1}^K\alpha_k\left(\bmC_k+\sum_{i=1}^{p_k}\bmA_{k,i}\bmY_{t-i}\bmB_{k,i}^\tp\right).
\end{equation*}
However, the use of conditional expectations may not be ideal for predicting future values due to the potential presence of multimodal predictive distributions \citep{wong2000mixture}. For instance, the conditional mean may lie in the trough of a bimodal distribution, rendering misleading predictions.  Depending on the shape of the predictive distribution, other point  predictors such as conditional mode may be appropriate. However, for a multimodal predictive distribution, a prediction set is generally more reliable and informative than any point predictor.  

Moreover,  residuals can be used for diagnostic checks. Following \cite{fong2007mixture}, the fitted values take into account the estimated conditional expectation of 
$Z_{t,k}$. Let $\hat{k}(t)$ be the index of the largest value in $\{{\tau}_{t,1},\dots,{\tau}_{t,K}\}$, i.e., $\hat{k}(t)=k$ if and only if
${\tau}_{t,k} =\max\{{\tau}_{t,1},\dots,{\tau}_{t,K}\}$. That is to say, the observation at time $t$ is assumed to be generated by component $\hat{k}(t)$. The fitted values are defined as
\begin{equation*}
	\hat{\bmY}_t=\bmC_{\hat{k}(t)}+
	\hat{\bmA}_{\hat{k}(t)}\bmY_{t-1}\hat{\bmB}_{\hat{k}(t)}^\tp,
\end{equation*}
and the residuals are ${\bmY_t -\hat{\bmY}_t}$,
which can be used to evaluate the goodness of fit for the model. However, common tests for serial correlations among the residuals, such as the multivariate portmanteau tests, cannot be directly applied, as the null distributions of these tests are nontrivial for the MMAR models.

\section{Empirical Results}
\label{sec:empirical}
\subsection{Simulation Studies}
\subsubsection{Performance of the EM algorithm}	
We consider the following two scenarios:
\begin{itemize}
	\item{Scenario 1:} An MMAR(2;1,1) with $(m,n)=(2,3)$.
	\item{Scenario 2:}  An MMAR(2;1,1) with $(m,n)=(4,5)$. 
\end{itemize}
In each scenario, the mixing weights are set to be $(\alpha_1,\alpha_2) = (0.4 ,0.6)$.
The coefficient matrices $(\bmA_{k,i}, \bmB_{k,i},\bmC_{k})$ are generated from random normal matrices with mean $\bm{0}$. 
The variance-covariance matrix $\bmU_k$ is generated by $\bmU_k=\bmQ\bm{\Lambda}\bmQ^\tp$, where $\bmQ$ is a random orthogonal matrix, and $\bm{\Lambda}$ is a diagonal matrix whose elements are absolute values of i.i.d. standard normal random variables. $\bmV_k$ is generated in a similar way. For each scenario, those parameters are randomly generated once, and then remain fixed. In Scenario 1, both components are weakly stationary as $\rho(\bmB_{1,1}\otimes\bmA_{1,1})=0.766<1$ and $\rho(\bmB_{2,1}\otimes\bmA_{2,1})=0.952<1$. In Scenario 2, the first component is stationary, while the second is not as $\rho(\bmB_{1,1}\otimes\bmA_{1,1})=0.668<1$ and $\rho(\bmB_{2,1}\otimes\bmA_{2,1})=1.014>1$. However, both models can be easily verified to be stationary with finite sixth moments, which follows from Corollary~\ref{mmar_ws_cor}. For example, for the second simulation model,  $\sum_{k=1}^2\alpha_k(\rho(\bmB_{k,1})\times \rho(\bmA_{k,1}))^6=0.686 <1$.

For each scenario, 1000 independent realizations with length $T$ are generated, where $T\in\{200, 400, 800, 1600\}$. Then we use the proposed EM algorithm for parameter estimation. Initialization  for the EM algorithm was done as described at the end of Section \ref{sec:para est}. Furthermore, the identifiability constraints ~\eqref{iden_c1_0}-\eqref{iden_c1_1} were imposed for all reported numerical results. In particular, the 
average coverage rates of 95\% confidence intervals across all entries of each parameter matrix is computed. 

In practice, joint confidence regions for $\bm{B}_{k,j}$ and $\bm{A}_{k,j}$ may be of interest. Because of the identifiability constraints, the limiting joint distribution of the MLE of $\bm{B}_{k,j}$ is singular, although that of the MLE of any proper subset of $\bm{B}_{k,j}$ has a non-singular covariance matrix, thereby making it feasible to construct the joint confidence region for any proper subset of $\bm{B}_{k,j}$. Indeed, it follows from Eqn.~\eqref{eq:estimate constraint} that 
$
\vecc{\bm{B}_{k,j}}^\top\left(\vecc{\hat{\bm{B}}_{k,j}}-\vecc{\bm{B}_{k,j}}\right)=o_p(T^{-1/2})
$.    
Since $\vecc{\bm{B}_{k,j}}^\top\vecc{\bm{B}_{k,j}}=\vecc{\hat{\bm{B}}_{k,j}}^\top \vecc{\hat{\bm{B}}_{k,j}}=1$, the preceding asymptotic identity implies that
$
\vecc{\hat{\bm{B}}_{k,j}}^\top\left(\vecc{\hat{\bm{B}}_{k,j}}-\vecc{\bm{B}_{k,j}}\right)= o_p(T^{-1/2})
$, hence the true $\bm{B}_{k,j}(1,1)$ can be determined from the true $\veccq{\bm{B}_{k,j}}$ as follows:
\begin{eqnarray}
	&&
	\hat{\bm{B}}_{k,j}(1,1)\{\hat{\bm{B}}_{k,j}(1,1)-\bm{B}_{k,j}(1,1)\}\nonumber\\
	&=& -\veccq{\hat{\bm{B}}_{k,j}}^\top\left(\veccq{\hat{\bm{B}}_{k,j}}-\veccq{\bm{B}_{k,j}}\right)+o_p(T^{-1/2})  \label{eq:collinearity}
\end{eqnarray}   
where for any matrix $\bm{M}$, its $(1,1)$ element is denoted by $\bm{M}(1,1)$ and $\veccq{\cdot}$ denotes $\vecc{\cdot}$ with its first element removed.  This observation shows that one way to construct a joint confidence region for $\bm{B}_{k,j}$ is to first construct an elliptical joint confidence region for $\veccq{\bm{B}_{k,j}}$ and then compute the associated confidence interval for the $\bm{B}_{k,j}(1,1)$  by leveraging \eqref{eq:collinearity}. Note that $\hat{\bm{B}}_{k,j}(1,1) \not = 0$ a.s. The associated confidence interval for $\bm{B}_{k,j}(1,1)$ so computed is asymptotically the same as the confidence interval constructed based on the marginal distribution of $\hat{\bm{B}}_{k,j}(1,1)$. (We note that instead of the $(1,1)$-th element, any ${j,k}$-th element can be used in the preceding method.)
We adapted this approach to construct and derive the percentage of coverage of the 95\% joint confidence regions for $\bm{\xi}_1$, $\bm{\xi}_2$ and $(\bm{\xi}_1^\tp,\bm{\xi}_2^\tp)^\tp$, where
\begin{equation*}
	\bm{\xi}_k=(\vecc{\bmA_{k,1}}^\tp, \vecc{\bmB_{k,1}}^\tp,
	\dots,\vecc{\bmA_{k,p_k}}^\tp, \vecc{\bmB_{k,p_k}}^\tp,
	\vecc{\bmC_k}^\tp)^\tp,\quad k\in\{1,2\}.
\end{equation*}

The average coverage of the 95\% confidence intervals across all entries in each parameter matrix is given in 
Table \ref{sc12_avg_cov}, and the coverage of the 95\%  joint confidence region of each parameter vector is provided in Table \ref{sc12_er_cov}.
These tables clearly demonstrate that the coverage is precise, particularly when dealing with large sample sizes.

Table \ref{emp_merged} presents the performance of the EM algorithm for estimation. 
Specifically, we provide details for each element in matrix $\bmA_{1,1}$  within Scenario 1, including the true value, the mean of estimates, the theoretical standard error (se) and the empirical standard error. Here, $\bmA_{1,1}(g,h)$ denotes the $(g,h)$-th entry of the matrix $\bmA_{1,1}$, for $1\leq g,h\leq 2$.
In general, for each element, the mean of estimates is close to the true values, and the empirical standard error is closely aligned with the theoretical standard error. The performance of the EM algorithm for other $\bmA$'s and $\bmB$'s in both scenarios is similar and therefore not listed.

\renewcommand{\arraystretch}{0.6} 
\begin{table}[ht!]
	\begin{minipage}{0.5\textwidth}
		\centering
		\begin{tabular}{lllll}
			\hline
			\multicolumn{5}{c}{Scenario 1} \\
			\hline
			$T$&1600&800& 400 & 200 \\
			\hline
			$\bmA_{1,1}$&0.953&0.949  &0.935&0.934  \\
			$\bmB_{1,1}$&0.949	&  0.947& 0.944&0.939\\
			$\bmA_{2,1}$& 0.953& 0.946 & 0.936&0.885\\
			$\bmB_{2,1}$& 0.950&0.950& 0.931&0.877\\
			$\alpha_1$&0.946& 0.951&0.953&0.943\\
			\hline
		\end{tabular}
	\end{minipage}%
	\begin{minipage}{0.5\textwidth}
		\centering
		\begin{tabular}{lllll}
			\hline
			\multicolumn{5}{c}{Scenario 2} \\
			\hline
			$T$&1600&800& 400 & 200 \\
			\hline
			$\bmA_{1,1}$&0.945&0.948  &0.946&0.934  \\
			$\bmB_{1,1}$&0.952&  0.952& 0.952&0.942\\
			$\bmA_{2,1}$&0.951& 0.949 & 0.947&0.943\\
			$\bmB_{2,1}$&0.953& 0.952& 0.947&0.943\\
			$\alpha_1$&0.953& 0.942&0.955 &0.953\\	
			\hline
		\end{tabular}
	\end{minipage}
	\caption{\label{sc12_avg_cov}Empirical coverage rate of nominally 95\% CI.}
\end{table}

\begin{table}[ht!]
	\begin{minipage}{0.46\textwidth}
		\centering
		\begin{tabular}{lllll}
			\hline
			\multicolumn{5}{c}{Scenario 1} \\
			\hline
			$T$&1600	& 800 &400  &200  \\
			\hline
			$\bm{\xi}_1$&0.956 &0.942  &  0.930& 0.909 \\
			$\bm{\xi}_2$& 0.960&  0.954&0.912&0.827 \\
			$(\bm{\xi}_1^\tp,\bm{\xi}_2^\tp)^\tp$&0.965 & 0.955&0.913 &0.826\\
			\hline
		\end{tabular}
	\end{minipage}%
	\hfill
	\begin{minipage}{0.46\textwidth}
		\centering
		\begin{tabular}{lllll}
			\hline
			\multicolumn{5}{c}{Scenario 2} \\
			\hline
			$T$& 1600 & 800 &400  &200  \\
			\hline
			$\bm{\xi}_1$&0.948&0.941  &0.893  &0.817   \\
			$\bm{\xi}_2$&0.945&0.944  &0.921  & 0.894\\
			$(\bm{\xi}_1^\tp,\bm{\xi}_2^\tp)^\tp$&0.957& 0.934&0.906 &0.836\\
			\hline
		\end{tabular}
	\end{minipage}
	\caption{\label{sc12_er_cov}Empirical coverage rate of nominally 95\%  joint confidence regions.}
\end{table}

\begin{table}[ht!]
	\centering
	\begin{tabular}{llllll}
		\toprule
		& & $\bmA_{1,1}(1,1)$ & $\bmA_{1,1}(2,1)$ & $\bmA_{1,1}(1,2)$ & $\bmA_{1,1}(2,2)$ \\
		\midrule
		$T=200$ &true value & -0.752 & 0.694 & 0.662 & 0.844 \\ 
		&mean of estimates & -0.749 & 0.689 & 0.662 & 0.838 \\ 
		&empirical se & 0.048 & 0.165 & 0.053 & 0.088 \\ 
		&theoretical se & 0.024 & 0.132 & 0.013 & 0.071 \\ 
		\midrule
		$T=400$ &true value & -0.752 & 0.694 & 0.662 & 0.844 \\ 
		&mean of estimates & -0.750 & 0.691 & 0.662 & 0.843 \\ 
		&empirical se & 0.018 & 0.098 & 0.010 & 0.053 \\ 
		&theoretical se & 0.017 & 0.093 & 0.009 & 0.050 \\ 
		\midrule
		$T=800$ &true value & -0.752 & 0.694 & 0.662 & 0.844 \\ 
		&mean of estimates & -0.752 & 0.692 & 0.662 & 0.844 \\ 
		&empirical se & 0.012 & 0.067 & 0.007 & 0.036 \\ 
		&theoretical se & 0.012 & 0.066 & 0.007 & 0.036 \\
		\midrule
		$T=1600$ &true value & -0.752 & 0.694 & 0.662 & 0.844 \\ 
		&mean of estimates & -0.752 & 0.692 & 0.661 & 0.845 \\ 
		&empirical se & 0.008 & 0.045 & 0.005 & 0.025 \\ 
		&theoretical se & 0.008 & 0.047 & 0.005 & 0.025 \\ 
		\bottomrule
	\end{tabular}
	\caption{\label{emp_merged}
		Performance of the EM algorithm for scenario 1 with different values of $T$.}
\end{table}

\subsubsection{Comparison of the Information Criteria}
For each scenario, 500 independent realizations with length $T$ are generated, where $T\in\{200,400,800\}$. We then use the EM algorithm along with the proposed initial value selection method to estimate the parameters. For each estimation, the EM algorithm is repeated $m\times n$ times with different initial values, and the parameter estimate that results in the highest likelihood is selected.

We compare the models with $K\in\{1,2,3\}$.
For simplicity, only the models with $p_1=\dots=p_K=p_{\max}$ are considered.
For each scenario, we first selected the number of components with given AR orders. The percentages of correctly selecting $K$ are given in Table \ref{Sc12K}.
\begin{table}
	\begin{minipage}{0.48\textwidth}
		\begin{tabular}{lllll}
			\hline
			\multicolumn{5}{c}{Scenario 1} \\
			\hline
			$T$& AIC &BIC&HQ&GIC \\
			\hline
			$200$& 23.2\% &97.8\%& 83.8\%&99.8\%\\
			$400$&12.4\%  & 98.6\%&88.2\% &100.0\% \\
			$800$&11.6\%&99.4\%&92.4\%&100.0\%\\
			\hline
		\end{tabular}
	\end{minipage}%
	\hfill
	\begin{minipage}{0.48\textwidth}
		\centering
		\begin{tabular}{lllll}
			\hline
			\multicolumn{5}{c}{Scenario 2} \\
			\hline
			$T$& AIC &BIC&HQ&GIC  \\
			\hline
			$200$& 63.8\% &99.6\%&95.4\%&100.0\% \\
			$400$& 32.6\%&99.8\%&96.6\%&100.0\% \\
			$800$& 10.40\%&100.0\%&98.0\%&100.0\%\\
			\hline
		\end{tabular}
	\end{minipage}
	\caption{\label{Sc12K}Percentage of correctly selecting $K$ with $p_{\max}$ given.}
\end{table}
In general, the BIC and the GIC are highly effective in selecting both the number of components and the AR orders for the MMAR model, even when the AR orders are misspecified. In addition, their performance remains consistent for sequences of different lengths.
Generally, the GIC slightly outperforms the BIC. The HQ has a moderate performance in general.
But the AIC is not recommended for selecting the number of mixing components. 

We also compare the model selection performance for selecting the AR orders, with the number of components $K$ given. We select the models with AR orders up to 3. The results, which are presented in Table \ref{Sc12P}, demonstrate similar patterns as observed previously. Specifically, the BIC and the GIC are highly effective in selecting the AR orders, with the HQ exhibiting somewhat worse performance. In contrast, the AIC performed poorly hence not recommended. 

Model selection can be computationally intensive and time-consuming, particularly with moderate dimensional matrix-valued observations.
To speed up the calculation, we recommend a stepwise model selection approach using either the BIC or GIC by first selecting $K$ with $p_{\max}=1$, followed by choosing the AR orders with the selected $K$.  Table \ref{Sc3K_1} 
in the \emph{Supplemental Materials} demonstrate the effectiveness of this approach with the BIC or the GIC. 

In the \emph{Supplementary Materials}, we report  additional simulation results for assessing the empirical model selection performance in the case of three imbalanced mixing components and the case of an MAR model,  which shows  that the BIC and GIC perform well and can select the AR order and the number of components with increasing accuracy as the sample size increases, although in the imbalanced case, larger sample size is needed to achieve similar accuracy as compared to the balanced case. Furthermore, if the true model is an MAR model, BIC and GIC can pick $K=1$ with very high accuracy. This shows that the proposed model selection method provides an effective tool for specifying $K$ and the AR orders. Figs.~\ref{out-sample-prediction-boxplot} and \ref{out-sample-prediction-boxplot-2} show that fitting an MAR(2;1,1) model to data from an MAR(1) model results in slight inflation in the out-sample prediction performance, confirming  the conclusion of Theorem~\ref{mmar_con} that  the MLE is consistent even if the model is mis-specified to have more components than needed.

\begin{table}[ht!]
	\begin{minipage}{0.46\textwidth}
		\centering
		\begin{tabular}{lllll}
			\hline
			\multicolumn{5}{c}{Scenario 1} \\
			\hline
			$T$& AIC &BIC&HQ&GIC \\
			\hline
			$200$& 34.8\% &100.0\%& 99.8\%&100.0\%\\
			$400$&37.4\%  & 100.0\%&100.0\% &100.0\% \\
			$800$&42.4\%&100.0\%&100.0\%&100.0\%\\
			\hline
		\end{tabular}
	\end{minipage}%
	\hfill
	\begin{minipage}{0.46\textwidth}
		\centering
		\begin{tabular}{lllll}
			\hline
			\multicolumn{5}{c}{Scenario 2} \\
			\hline
			$T$& AIC &BIC&HQ&GIC \\
			\hline
			$200$& 0.4\% &100.0\%&100.0\%&100.0\%\\
			$400$&4.4\% &100.0\%&100.0\%&100.0\%\\
			$800$&7.4\%&100.0\%&100.0\%&100.0\%\\
			\hline
		\end{tabular}
	\end{minipage}
	\caption{\label{Sc12P}Percentage of correctly selecting $p_{\max}$ with $K$ given.}
\end{table}
\subsection{Real Data}	
The proposed MMAR model is applied to analyze the economic indicator dataset presented in  Fig.~\ref{fig_econo}. All the series are centered and normalized such that the pooled variance for each indicator across all the counties is 1.
We begin with model selection. The log-likelihood, the BIC, and the GIC for $K\in\{1,2,3,4\}$ and $p_{\max}\in\{1,2,3\}$ are given in Table \ref{econ_selection}. Again, only the models with $p_1=\dots=p_K=p_{\max}$ are considered.
According to BIC, an MMAR(3;1,1,1) model is selected while an MMAR(2;1,1) model is chosen by GIC.  We fit both models and evaluate their performance.  The MMAR(2;1,1) demonstrates better  interpretability and improved predictive accuracy on out-of-sample data. Accordingly, we present its results in this section, and the results for the MMAR(3;1,1,1) model are provided in Section  \ref{sec_K3_real} 
in the \emph{Supplementary Materials}.

For the fitted MMAR(2;1,1), convergence of the EM algorithm is monitored using the log-likelihood values across iterations. As shown in Fig.~\ref{K2_ll_trace}, the log-likelihood increases monotonically and eventually stabilizes. The algorithm stops when the increase in log-likelihood is less than the predefined threshold of $5\times10^{-4}$.
\begin{figure}[ht!]
	\centering
	\includegraphics[width=.8\textwidth]{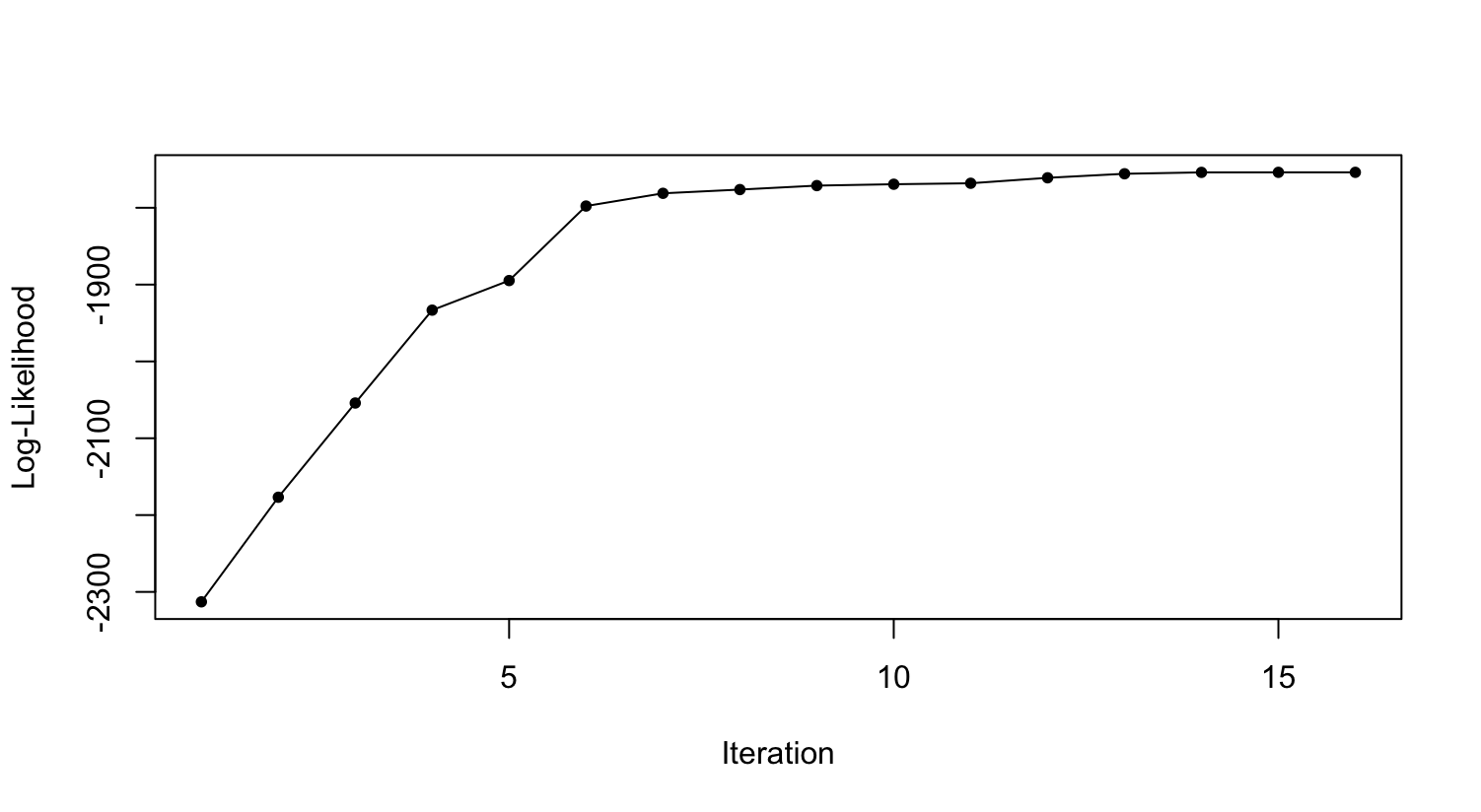}
	\caption[EM algorithm convergence plot for the fitted MMAR(2;1,1) model]{\label{K2_ll_trace} EM algorithm convergence plot for the fitted MMAR(2;1,1) model.}
\end{figure}
The standardized residuals of the fitted model (Fig.~\ref{k2_acf}) reveal no temporal patterns, suggesting a good fit. The estimated mixing weights are $\hat{\alpha}_1=0.183(0.001)$ and $\hat{\alpha}_2=0.817(0.001)$, where standard errors are shown in parentheses.
Since
\begin{equation*}
	\rho(\hat{\bmB}_{1,1}\otimes\hat{\bmA}_{1,1})=0.770<1\quad\text{and}\quad\rho(\hat{\bmB}_{2,1}\otimes\hat{\bmA}_{2,1})= 0.759<1,
\end{equation*}
both components of the mixture are weakly stationary. Moreover,
\begin{equation*}
	\sum_{k=1}^2\hat{\alpha}_k\log(\rho(\hat{\bmB}_{k,1})\rho(\hat{\bmA}_{k,1}))=-0.273<0,\qquad
	\sum_{k=1}^2\hat{\alpha}_k(\rho(\hat{\bmB}_{k,1})\rho(\hat{\bmA}_{k,1}))^q<1,
\end{equation*}
for any positive integer $q$. 
By Corollaries \ref{mmar_ss_cor} and \ref{mmar_ws_cor},
the overall model is strictly stationary, whose stationary distribution admits finite $q$-th moments, for any positive integer $q$. It follows from Proposition~\ref{mmar_gergod} that the fitted model is geometric ergodic and $\beta$-mixing with geometric decaying rate. Furthermore, we use the following decision rule \citep{mclachlanfinite} to classify the dataset into different regimes:
\begin{equation}
	\label{regime_decision}
	\bmY_t\in\text{regime}~i \quad\text{if}\quad
	\hat{\alpha}_if_t(\bmy_{t}|\mathscr{F}_{t-1};\hat{\bmtt}_i)\geq
	\hat{\alpha}_jf_t(\bmy_{t}|\mathscr{F}_{t-1};\hat{\bmtt}_j),\quad
	j\in\{1,2,\dots,K\},
\end{equation}
where $\hat{\bmtt}_j$ is the MLE of $\bmtt_j$, and $f_t(\bmy_{t}|\mathscr{F}_{t-1};{\bmtt_j})=f_t(\bmy_{t}|\mathscr{F}_{t-1};\bm{\Delta}_j({\bmtt_j}))$, i.e., the conditional pdf based on the $j$th regime. 
Based on the fitted model, the data are classified into two regimes, as shown in Fig.~\ref{K2_econ_c}, where regime 1 is shaded red, and regime 2 is unshaded. Here, regime 1 corresponds to the first mixture component, and regime 2 corresponds to the second.
It is worth noting that regime 1 displays relatively high volatility, while regime 2 exhibits comparatively low volatility.
\begin{figure}[h]
	\begin{center}
		\includegraphics[width=1\textwidth]{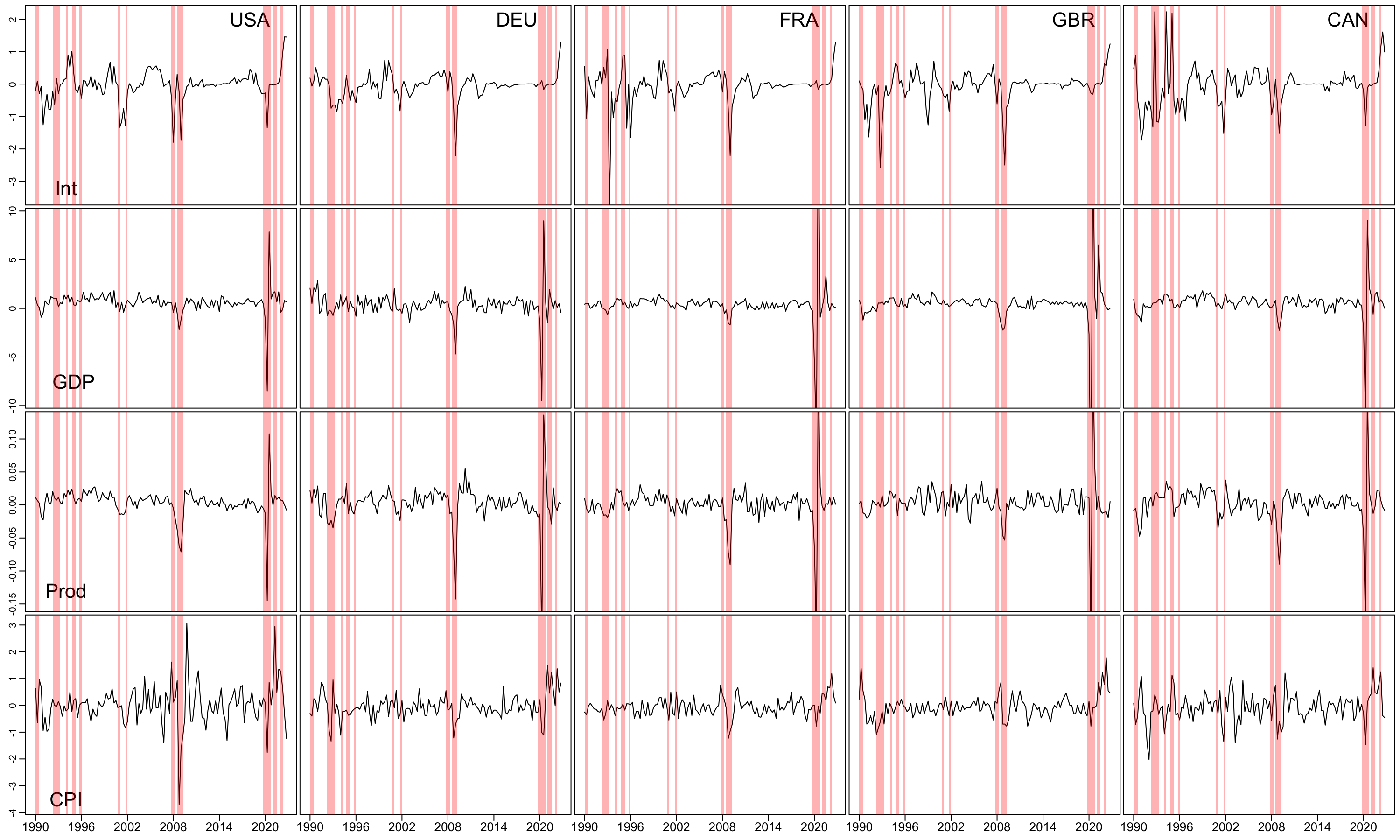}
	\end{center}
	\vspace{-0.2cm}
	\caption{\label{K2_econ_c} Time series plot of the
		economic indicators, with regime 1 shaded red and regime 2 unshaded using MMAR(2;1,1) model.}
\end{figure}

Tables \ref{K2_econ_A1} -- \ref{K2_econ_c2} show the MLE of the parameter matrices $\bmA_{k,1}, \bmB_{k,1}$ and $\bmC_{k,1}$ for $k\in\{1,2\}$, and the corresponding standard errors, respectively. Due to the identifiability constraints, the Frobenius norms of $\bmB$'s are scaled to 1. To facilitate model interpretation \citep{chen2021autoregressive}, the signs of the significant coefficient matrix elements, at the 5\% level,  are displayed on the right-hand side of each table, specifically, 
using symbols (+) for positively significant, (-) for negatively significant, and (0) for insignificant coefficients.
For instance, the first column of $\bmA_{k,1}$ can be understood as the impact of the previous quarter's interest rates on the current economic indicators, while the first column of $\bmB_{k,1}$ captures the influence of US's last quarter's indicators on the current quarter's indicators of all countries, for each $k\in\{1,2\}$. The estimated parameter matrices $\hat{\bmA}$'s and $\hat{\bmB}$'s demonstrate both differences and similarities among different regimes. Concerning the differences, one example is that the first column of $\hat{\bmA}_{1,1}$ indicates that an increase in the interest rate of the previous quarter has a significantly negative influence on the interest rate in the current quarter in regime 1. In contrast, the first column of $\hat{\bmA}_{2,1}$ indicates the opposite effect in regime 2.
Regarding the similarities, by checking the first columns of $\hat{\bmB}$'s, it is observed that
the US's previous quarter's indicators consistently have a positive effect on current quarter's indicators from all the countries across the two regimes. 

The out-of-sample prediction performance is also examined. We use the data from Q1 1990 to Q2 2021 ($1\leq t\leq 126$) to fit the model and derive the MLE of the parameter.
Subsequently, we derive the marginal predictive distributions for the period from Q3 2021 to Q4 2022 ($127\leq t\leq 132$),  based on \eqref{mixture density 1} with the parameter replaced by the MLE. 
The observed values along with the predictive values by the conditional mean are shown in Figs.~\ref{K2_t127} -- \ref{K2_t128} and Figs.~\ref{K2_t129} -- \ref{K2_t132}. In each plot, the shaded areas indicate the 95\% highest density region. The plots reveal several notable patterns. The marginal predictive distributions of interest rate and the CPI are generally unimodal, whereas those of GDP growth and industrial production growth tend to be either bimodal or skewed unimodal. 
In general, most of the 95\% highest density regions capture the true observations. 
\begin{figure}[ht!]
	\begin{center}
		\scalebox{0.125}{\includegraphics{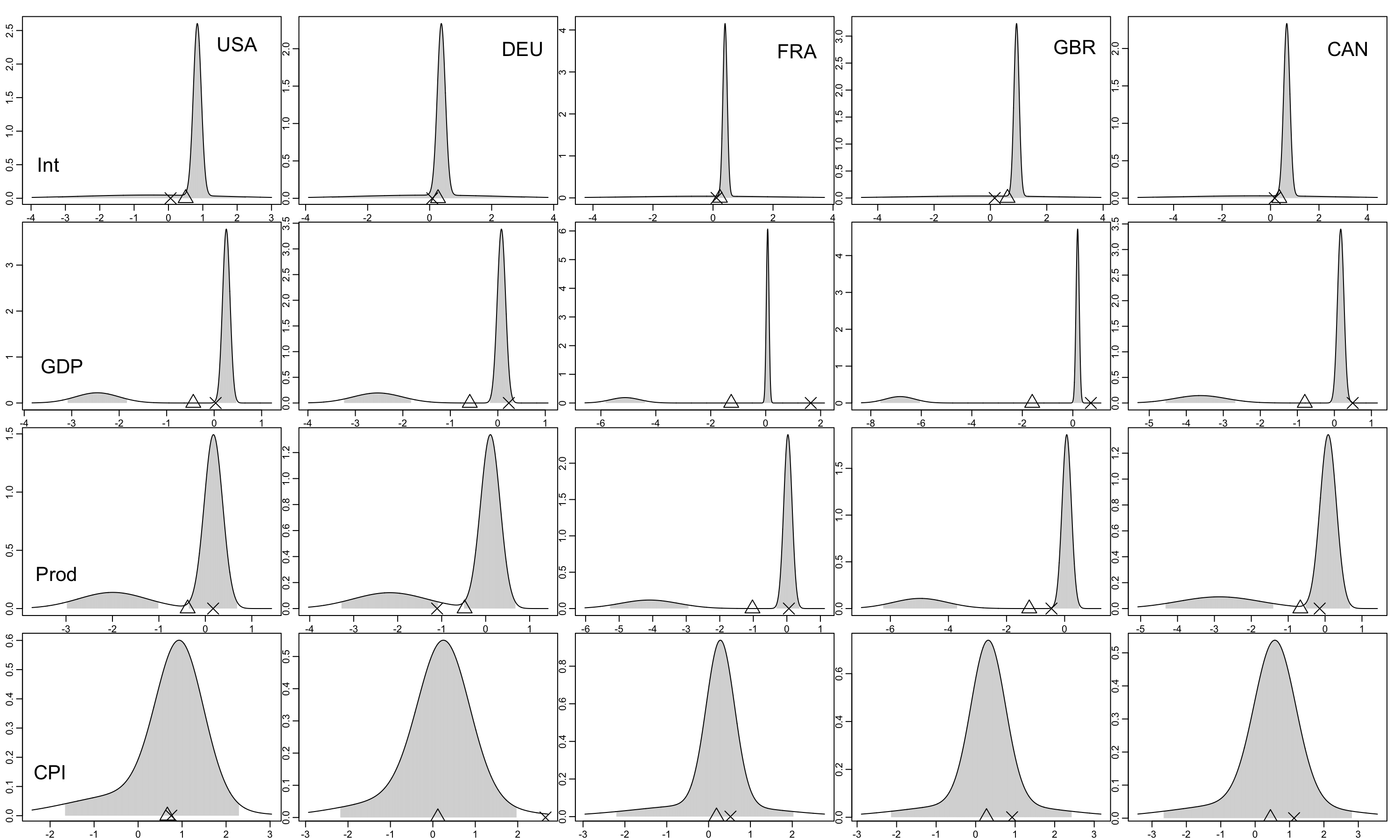}}
	\end{center}
	\vspace{-0.2cm}
	\caption{\label{K2_t127} One-step marginal predictive distribution for Q3 2021 under MMAR(2;1,1) model, with $\times$ representing the observed values and $\triangle$ the predicted values, and the shaded areas representing the 95\% highest density interval.}
\end{figure}

\begin{figure}[ht!]
	\begin{center}
		\scalebox{0.125}{\includegraphics{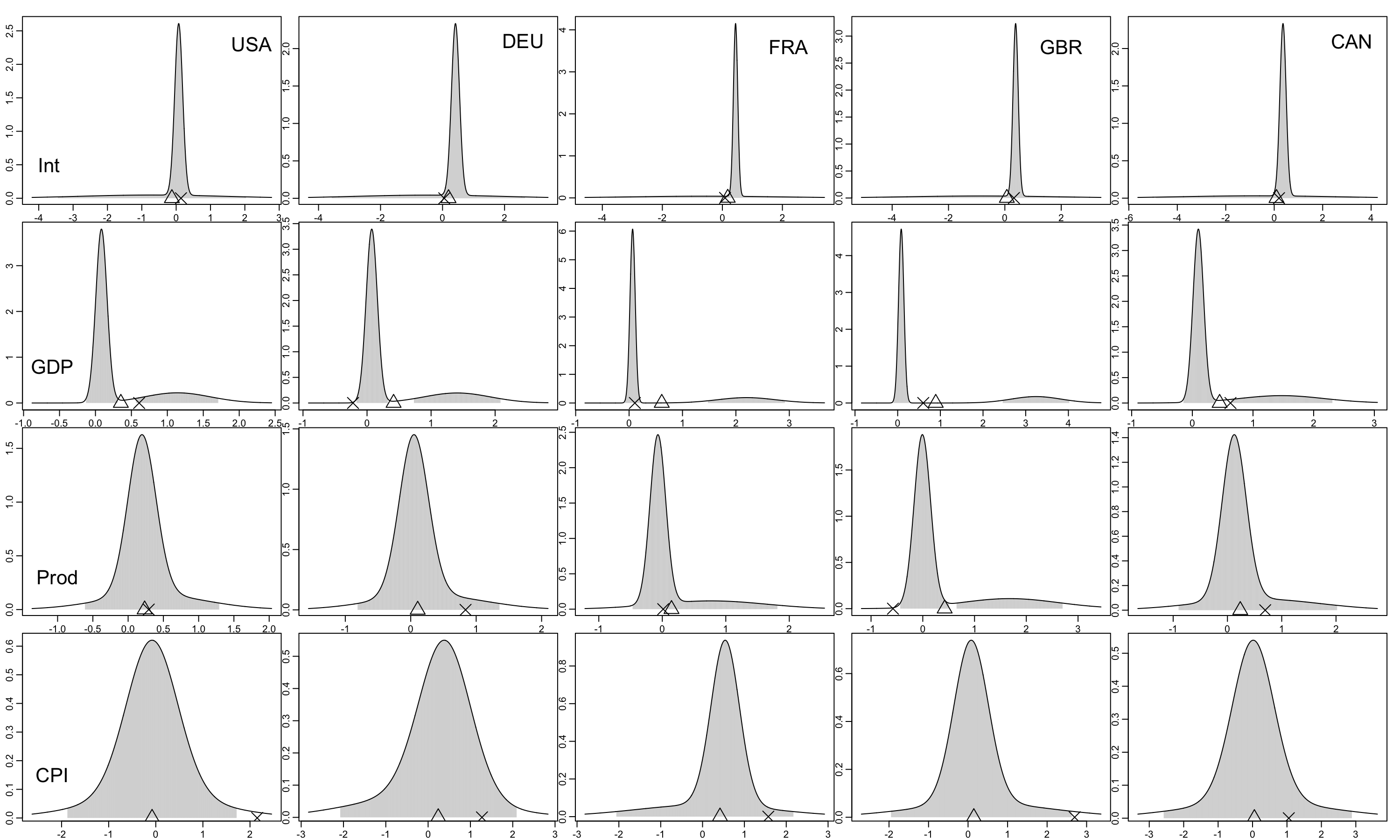}}
	\end{center}
	\vspace{-0.2cm}
	\caption{\label{K2_t128}One-step marginal predictive distribution for Q4 2021 under MMAR(2;1,1) model, with $\times$ representing the observed values and $\triangle$ the predicted values, and the shaded areas representing the 95\% highest density interval.}
\end{figure}

The one-step ahead out-of-sample prediction errors of the MMAR(2;1,1) model are compared with the following models:
\begin{enumerate}
	\item MAR($p$):  matrix autoregression, $p\in\{1,2\}$.
	\item VAR($p$): vector autoregression, $p\in\{1,2\}$.
\end{enumerate}
We have also attempted to implement the mixture VAR model \citep{fong2007mixture}. However, the estimation process using the EM algorithm did not converge due to a    singular variance-covariance matrix error.
Additionally, fitting the Gaussian mixture vector autoregressive model \citep{kalliovirta2016gaussian} using the {\tt gmvarkit} package\footnote{\url{https://cran.r-project.org/web/packages/gmvarkit/index.html}} resulted in errors. The estimation errors suggest that these two models may be inappropriate for analyzing high-dimensional data. We also include the MMAR(3;1,1,1) model for comparison.

Using the conditional mean for prediction, the mean squared prediction errors (MSPE) are given in Table \ref{out_pred}. Although for the mixture models, the conditional expectations may not be optimal for  predicting future values, 
the MMAR models still clearly outperform the MAR and VAR models. The MMAR(2;1,1) shows a modest improvement over the MMAR(3;1,1,1) in terms of out-of-sample prediction.
\begin{table}[ht!]
	\centering
	\begin{tabular}{cccccc}
		\hline
		MMAR(2;1,1)	&MMAR(3;1,1,1)& MAR(1)& MAR(2)&VAR(1)&VAR(2)\\
		\hline
		23.46&26.22	& 50.72 &  54.13& 73.10&134.64\\
		\hline
	\end{tabular}
	\caption{\label{out_pred}Mean of squared out-of-sample prediction errors.}
\end{table}

\section{Conclusion}
We have proposed a new mixture model for matrix-valued time series data, with the capability to effectively capture changing dynamics. We investigate both strict and weak stationarity conditions for the proposed model. An EM algorithm is implemented to estimate the MLE of the parameters, and the asymptotic properties of the MLE are derived. Based on our simulation results, we recommend using either the BIC or the GIC for model selection.

There are several directions to extend the proposed MMAR model. Fig.~\ref{K2_econ_c} shows that the regimes tend
to cluster over time, i.e., there is dependence over time in the regimes--a key stylized fact in many financial and economic data.  Thus, it is of interest to extend the MMAR model to incorporate Markov or dependent switching for matrix-valued time series. 
Note that \cite{pouzo2025robustness} showed that, under some regularity conditions, quasi-likelihood estimation of the parameters for the within-regime dynamics in a hidden Markov model is robust even if the true inhomogeneous Markov switching mechanism is misspecified as i.i.d. switching. However, their robustness result does not hold if within some regimes, the response depends on its lags, hence their robustness result is not applicable to the proposed MMAR model. 
The conditional matrix normal distribution in the model may be replaced by other distributions, such as matrix-valued t-distributions or even some skewed matrix-valued distributions \citep{gallaugher2018finite}. These models are potentially useful for modeling matrix-valued financial data with heavy tails, such as the Fama-French portfolios cross-classified by size and book-to-market ratio. 
Moreover, it is important to note that the proposed MMAR model may contain a large number of parameters, particularly with high-dimensional  matrix time series and numerous mixture components with high autoregressive orders. 
But not all coefficients need to vary across regimes. For instance, Fig.~\ref{fig_econo} suggests that in some regimes, all means are shifted by some large negative amount,
or the covariance in the shocks is much higher. Specifying a constrained model may provide a more parsimonious
specification.
Another promising solution to the aforementioned problem is to assume that the parameter matrices $\bmA$'s and $\bmB$'s are of low ranks, resulting in a reduced-rank MMAR model. In addition, when dealing with high-dimensional matrix-valued time series data, regularization methods can be applied to promote sparsity. These regularization methods can also be applied to the variance-covariance matrices. We may also assume that variance-covariance matrices admit some low rank structures, which can be represented by the sum of a diagonal matrix and a low rank matrix.
\bigskip
\begin{center}
	{\large\bf SUPPLEMENTAL MATERIALS}
\end{center}
The \emph{Supplemental Materials} contains the proofs of the theorems and  additional results for the simulation studies and the real application. 

\bibliographystyle{apalike}
\bibliography{Bibliography}

@article{pouzo2025robustness,
  title={On the robustness of mixture models in the presence of hidden Markov Regimes with covariate-dependent transition probabilities},
  author={Pouzo, Demian and Psaradakis, Zacharias and Sola, Martin},
  journal={Econometric Theory},
  pages={1--15},
  year={2025},
  publisher={Cambridge University Press}
}

@book{silvey2017statistical,
  title={Statistical inference},
  author={Silvey, Samuel David},
  year={2017},
  publisher={Routledge}
}

@book{pollard2012convergence,
  title={Convergence of stochastic processes},
  author={Pollard, David},
  year={2012},
  publisher={Springer Science \& Business Media}
}

@book{rao1971generalized,
  title={Generalized Inverse of Matrices and its Applications},
  author={Rao, C. Radhakrishna and Mitra, Sujit Kumar},
  year={1971},
  publisher={Wiley},
  address={New York, London},
  series={Wiley Series in Probability and Mathematical Statistics},
  note={xiv, 240 p.}
}

@article{yu1994rates,
  title={Rates of convergence for empirical processes of stationary mixing sequences},
  author={Yu, Bin},
  journal={The Annals of Probability},
  pages={94--116},
  year={1994},
  publisher={JSTOR}
}

@article{shapiro1986asymptotic,
  title={Asymptotic theory of overparameterized structural models},
  author={Shapiro, Alexander},
  journal={Journal of the American Statistical Association},
  volume={81},
  number={393},
  pages={142--149},
  year={1986},
  publisher={Taylor \& Francis}
}

@article{geyer1994asymptotics,
  title={On the asymptotics of constrained M-estimation},
  author={Geyer, Charles J},
  journal={The Annals of statistics},
  pages={1993--2010},
  year={1994},
  publisher={JSTOR}
}

@article{saikkonen2007stability,
  title={Stability of mixtures of vector autoregressions with autoregressive conditional heteroskedasticity},
  author={Saikkonen, Pentti},
  journal={Statistica Sinica},
  volume={17},
  number={1},
  pages={221--239},
  year={2007},
  publisher={JSTOR}
}

@incollection{doukhan1995mixing,
  title={Mixing},
  author={Doukhan, Paul},
  booktitle={Mixing: Properties and Examples},
  pages={15--23},
  year={1995},
  publisher={Springer}
}

@article{feigin1985random,
  title={Random coefficient autoregressive processes: a Markov chain analysis of stationarity and finiteness of moments},
  author={Feigin, Paul D and Tweedie, Richard L},
  journal={Journal of Time Series Analysis},
  volume={6},
  number={1},
  pages={1--14},
  year={1985},
  publisher={Wiley Online Library}
}

@article{magnus1985matrix,
  title={Matrix differential calculus with applications to simple, Hadamard, and Kronecker products},
  author={Magnus, Jan R and Neudecker, Heinz},
  journal={Journal of Mathematical Psychology},
  volume={29},
  number={4},
  pages={474--492},
  year={1985},
  publisher={Elsevier}
}

@article{redner1981note,
  title={Note on the consistency of the maximum likelihood estimate for nonidentifiable distributions},
  author={Redner, Richard},
  journal={The Annals of Statistics},
  pages={225--228},
  year={1981},
  publisher={JSTOR}
}

@article{peters1978iterative,
	title={An iterative procedure for obtaining maximum-likelihood estimates of the parameters for a mixture of normal distributions},
	author={Peters, Jr, B Charles and Walker, Homer F},
	journal={SIAM Journal on Applied Mathematics},
	volume={35},
	number={2},
	pages={362--378},
	year={1978},
	publisher={SIAM}
}

@article{redner1984mixture,
	title={Mixture densities, maximum likelihood and the EM algorithm},
	author={Redner, Richard A and Walker, Homer F},
	journal={SIAM review},
	volume={26},
	number={2},
	pages={195--239},
	year={1984},
	publisher={SIAM}
}

@article{chan1993asymptotic,
	title={Asymptotic behavior of the Gibbs sampler},
	author={Chan, Kung-Sik},
	journal={Journal of the American Statistical Association},
	volume={88},
	number={421},
	pages={320--326},
	year={1993},
	publisher={Taylor \& Francis}
}

@article{lu2005likelihood,
	title={The likelihood ratio test for a separable covariance matrix},
	author={Lu, Nelson and Zimmerman, Dale L},
	journal={Statistics \& Probability Letters},
	volume={73},
	number={4},
	pages={449--457},
	year={2005},
	publisher={Elsevier}
}

@article{meng2022penalized,
	title={Penalized quasi-likelihood estimation of generalized Pareto regression--consistent identification of risk factors for extreme losses},
	author={Meng, Jin and Chan, Kung-Sik},
	journal={Insurance: Mathematics and Economics},
	volume={104},
	pages={60--75},
	year={2022},
	publisher={Elsevier}
}

@article{nishii1984asymptotic,
	title={Asymptotic properties of criteria for selection of variables in multiple regression},
	author={Nishii, Ryuei},
	journal={The Annals of Statistics},
	pages={758--765},
	year={1984},
	publisher={JSTOR}
}

@article{gallaugher2018finite,
	title={Finite mixtures of skewed matrix variate distributions},
	author={Gallaugher, Michael PB and McNicholas, Paul D},
	journal={Pattern Recognition},
	volume={80},
	pages={83--93},
	year={2018},
	publisher={Elsevier}
}

@article{davies1987hypothesis,
	title={Hypothesis testing when a nuisance parameter is present only under the alternative},
	author={Davies, Robert B},
	journal={Biometrika},
	volume={74},
	number={1},
	pages={33--43},
	year={1987},
	publisher={Oxford University Press}
}

@article{chan1990likelihood,
	title={On likelihood ratio tests for threshold autoregression},
	author={Chan, Kung-Sik and Tong, Howell},
	journal={Journal of the Royal Statistical Society: Series B (Methodological)},
	volume={52},
	number={3},
	pages={469--476},
	year={1990},
	publisher={Wiley Online Library}
}

@book{hannan1970multiple,
	title={Multiple time series: Wiley series in probability and mathematical statistics},
	author={Hannan, Edward J},
	year={1970},
	publisher={John Wiley and Sons, Inc.(New York)}
}

@article{dempster1977maximum,
	title={Maximum likelihood from incomplete data via the EM algorithm},
	author={Dempster, Arthur P and Laird, Nan M and Rubin, Donald B},
	journal={Journal of the Royal Statistical Society: Series B (Methodological)},
	volume={39},
	number={1},
	pages={1--22},
	year={1977},
	publisher={Wiley Online Library}
}

@article{le1996modeling,
	title={Modeling flat stretches, bursts outliers in time series using mixture transition distribution models},
	author={Le, Nhu D and Martin, R Douglas and Raftery, Adrian E},
	journal={Journal of the American Statistical Association},
	volume={91},
	number={436},
	pages={1504--1515},
	year={1996},
	publisher={Taylor \& Francis}
}

@article{mai2022doubly,
	title={A doubly enhanced em algorithm for model-based tensor clustering},
	author={Mai, Qing and Zhang, Xin and Pan, Yuqing and Deng, Kai},
	journal={Journal of the American Statistical Association},
	volume={117},
	number={540},
	pages={2120--2134},
	year={2022},
	publisher={Taylor \& Francis}
}

@book{fan2003nonlinear,
	title={Nonlinear time series: nonparametric and parametric methods},
	author={Fan, Jianqing and Yao, Qiwei},
	volume={20},
	year={2003},
	publisher={Springer}
}

@book{tong1990non,
	title={Non-linear time series: a dynamical system approach},
	author={Tong, Howell},
	year={1990},
	publisher={Oxford university press}
}

@article{wang2021high,
	title={High-dimensional low-rank tensor autoregressive time series modeling},
	author={Wang, Di and Zheng, Yao and Li, Guodong},
	journal={arXiv preprint arXiv:2101.04276},
	year={2021}
}

@article{chang2022modelling,
	title={Modelling matrix time series via a tensor {CP}-decomposition},
	author={Chang, Jinyuan and Zhang, Henry and Yang, Lin and Yao, Qiwei},
	journal={Journal of the Royal Statistical Society. Series B: Statistical Methodology},
	year={2022}
}

@article{han2021cp,
	title={{CP} factor model for dynamic tensors},
	author={Han, Yuefeng and Zhang, Cun-Hui and Chen, Rong},
	journal={arXiv preprint arXiv:2110.15517},
	year={2021}
}

@article{nicholson2020high,
	title={High dimensional forecasting via interpretable vector autoregression},
	author={Nicholson, William B and Wilms, Ines and Bien, Jacob and Matteson, David S},
	journal={The Journal of Machine Learning Research},
	volume={21},
	number={1},
	pages={6690--6741},
	year={2020},
	publisher={JMLRORG}
}

@article{basu2015regularized,
	title={Regularized estimation in sparse high-dimensional time series models},
	author={Basu, Sumanta and Michailidis, George},
	journal={The Annals of Statistics},
	volume={43},
	number={4},
	pages={1535--1567},
	year={2015}
}

@article{fan2020factor,
	title={Factor-adjusted regularized model selection},
	author={Fan, Jianqing and Ke, Yuan and Wang, Kaizheng},
	journal={Journal of Econometrics},
	volume={216},
	number={1},
	pages={71--85},
	year={2020},
	publisher={Elsevier}
}

@article{pena2019forecasting,
	title={Forecasting multiple time series with one-sided dynamic principal components},
	author={Pe{\~n}a, Daniel and Smucler, Ezequiel and Yohai, Victor J},
	journal={Journal of the American Statistical Association},
	year={2019},
	publisher={Taylor \& Francis}
}

@article{hoff2015multilinear,
	title={Multilinear tensor regression for longitudinal relational data},
	author={Hoff, Peter D},
	journal={The Annals of Applied Statistics},
	volume={9},
	number={3},
	pages={1169},
	year={2015},
	publisher={NIH Public Access}
}

@article{chen2021autoregressive,
	title={Autoregressive models for matrix-valued time series},
	author={Chen, Rong and Xiao, Han and Yang, Dan},
	journal={Journal of Econometrics},
	volume={222},
	number={1},
	pages={539--560},
	year={2021},
	publisher={Elsevier}
}

@article{hamilton1989new,
	title={A new approach to the economic analysis of nonstationary time series and the business cycle},
	author={Hamilton, James D},
	journal={Econometrica: Journal of the Econometric Society},
	pages={357--384},
	year={1989},
	publisher={JSTOR}
}

@article{ding2018matrix,
	title={Matrix variate regressions and envelope models},
	author={Ding, Shanshan and Cook, R Dennis},
	journal={Journal of the Royal Statistical Society: Series B (Statistical Methodology)},
	volume={80},
	number={2},
	pages={387--408},
	year={2018},
	publisher={Wiley Online Library}
}

@article{fong2007mixture,
	title={On a mixture vector autoregressive model},
	author={Fong, Pak Wing and Li, Wai Keung and Yau, CW and Wong, Chun Shan},
	journal={Canadian Journal of Statistics},
	volume={35},
	number={1},
	pages={135--150},
	year={2007},
	publisher={Wiley Online Library}
}

@article{wong2000mixture,
	title={On a mixture autoregressive model},
	author={Wong, Chun Shan and Li, Wai Keung},
	journal={Journal of the Royal Statistical Society: Series B (Statistical Methodology)},
	volume={62},
	number={1},
	pages={95--115},
	year={2000},
	publisher={Wiley Online Library}
}

@book{lutkepohl2005new,
	title={New introduction to multiple time series analysis},
	author={L{\"u}tkepohl, Helmut},
	year={2005},
	publisher={Springer Science \& Business Media}
}

@article{lam2012factor,
	title={Factor modeling for high-dimensional time series: inference for the number of factors},
	author={Lam, Clifford and Yao, Qiwei},
	journal={The Annals of Statistics},
	pages={694--726},
	year={2012},
	publisher={JSTOR}
}

@article{hsu2021matrix,
	title={Matrix autoregressive spatio-temporal models},
	author={Hsu, Nan-Jung and Huang, Hsin-Cheng and Tsay, Ruey S},
	journal={Journal of Computational and Graphical Statistics},
	volume={30},
	number={4},
	pages={1143--1155},
	year={2021},
	publisher={Taylor \& Francis}
}

@article{kalliovirta2016gaussian,
	title={Gaussian mixture vector autoregression},
	author={Kalliovirta, Leena and Meitz, Mika and Saikkonen, Pentti},
	journal={Journal of Econometrics},
	volume={192},
	number={2},
	pages={485--498},
	year={2016},
	publisher={Elsevier}
}

@article{kalliovirta2015gaussian,
	title={A Gaussian mixture autoregressive model for univariate time series},
	author={Kalliovirta, Leena and Meitz, Mika and Saikkonen, Pentti},
	journal={Journal of Time Series Analysis},
	volume={36},
	number={2},
	pages={247--266},
	year={2015},
	publisher={Wiley Online Library}
}

@article{gao2021regularized,
	title={Regularized matrix data clustering and its application to image analysis},
	author={Gao, Xu and Shen, Weining and Zhang, Liwen and Hu, Jianhua and Fortin, Norbert J and Frostig, Ron D and Ombao, Hernando},
	journal={Biometrics},
	volume={77},
	number={3},
	pages={890--902},
	year={2021},
	publisher={Wiley Online Library}
}

@book{douc2014nonlinear,
	title={Nonlinear time series: Theory, methods and applications with R examples},
	author={Douc, Randal and Moulines, Eric and Stoffer, David},
	year={2014},
	publisher={CRC press}
}

@book{mclachlanfinite,
	title={Finite Mixture Models},
	author={McLachlan, Geoffrey and Peel, David},
	year={2000},
	publisher={Wiley Online Library}
}

@book{horn2012matrix,
	title={Matrix Analysis},
	author={Horn, Roger A and Johnson, Charles R},
	year={2012},
	publisher={Cambridge university press}
}

@article{li2021multi,
	title={Multi-linear Tensor Autoregressive Models},
	author={Li, Zebang and Xiao, Han},
	journal={arXiv preprint arXiv:2110.00928},
	year={2021}
}
\clearpage

\renewcommand{\thesection}{S.\arabic{section}}
\counterwithin{Lemma}{section}
\counterwithin{Theorem}{section}
\counterwithin{Proposition}{section}
\numberwithin{equation}{section}
\counterwithin{figure}{section}
\counterwithin{table}{section}
\setcounter{section}{0}
\setcounter{page}{1}
\begin{center}
	\Large{\textbf{Supplemental Materials of ``Mixture Matrix-valued Autoregressive Model"}}
\end{center}

The \emph{Supplemental Materials} are organized as follows. Section \ref{proofs_s3} gives the proofs of the propositions in Section 3 in the main text, which are related to the stationarity and ergodicity of the MMAR model. Section~\ref{a_positive} proves the positive definiteness of the Fisher information matrix for the unconstrained MMAR model. 
Sections~\ref{a-thm1}, \ref{a_thm2}  and \ref{a_thm3} contains the proofs for Theorems~\ref{mmar_con}, \ref{Thm: 2} and \ref{Thm: ident}.
Section \ref{a_simu} presents some additional simulations, and Section \ref{a_da} shows some additional results for the real data analysis.

\section{Proofs of the Propositions in Section 3 }\label{proofs_s3}
The strict and weak stationary conditions of the SDE model \eqref{sre} are given by the following two theorems, respectively.
\begin{Theorem}[Theorem 4.27 in \citealt{douc2014nonlinear}]\label{sre_str}
	In model \eqref{sre},
	let $\{(\bmD_t,\bm{\eta}_t)\}$ be a sequence of strictly stationary and ergodic sequence. Assume that,
	\begin{equation*}
		\E{\log^{+}\|\bmD_1\|}<\infty~~~~and~~~~
		\E{\log^{+}\|\bm{\eta}_0\|}<\infty.
	\end{equation*}
	Also assume that its top-Lyapunov exponent $\gamma$, defined in \eqref{top_lya}, 
	is strictly negative. Then, 
	\begin{equation}
		\tilde{\mathcal{X}_t}=\sum_{j=0}^\infty\left(\prod_{i=t-j+1}^t\bmD_i\right)\bm{\eta}_{t-j}\label{sre_solu_a}
	\end{equation}
	is the unique strictly stationary solution to Eqn.~\eqref{sre}.
\end{Theorem}

\begin{Theorem}[Theorem 4.30 in \citealt{douc2014nonlinear}]\label{sre_wk}
	\label{Thm_weak_D}
	Let $q\geq1$ and $\{(\bmD_t,\bm{\eta}_t)\}$ be a sequence of i.i.d. random elements, such that the $q$th norm Lyapunov coefficient
	is strictly negative, and $\E{\|\bm{\eta}_0\|^q}<\infty$. Then Eqn.~\eqref{sre} has a unique strictly stationary solution $\tilde{\mathcal{X}_t}$ given in \eqref{sre_solu_a}, such that $\E{\|\tilde{\mathcal{X}_t}\|^q}<\infty$. Moreover, the right-hand-side of  \eqref{sre_solu} converges in the $q$th norm.
\end{Theorem}
The Fekete's sub-additive lemma can be used to derive an equivalent expression of $\gamma_q$ given in Eqn.~\eqref{p_lya}.
\begin{Lemma}[Fekete's Subadditive Lemma]
	Let $\{a_t,t\geq1\}$ be a sequence, such that $\forall t_1,t_2\in\mathbb{N}^{*}$, $a_{t_1+t_2}\leq a_{t_1}+a_{t_2}$. Then,
	\begin{equation*}
		\lim_{t\to\infty}\frac{a_t}{t}=\inf_{t\in\mathbb{N}^{*}}\frac{a_t}{t}.
	\end{equation*}
\end{Lemma}

\begin{proof}[Proof of Eqn.~\eqref{p_lya}]
	By definition, matrix norms enjoys the property of sub-multiplicative \citep[pp.~341]{horn2012matrix}.
	That is to say, for any matrices $\bmM_1$ and $\bmM_2$ such that $\bmM_1\bmM_2$ is well-defined,
	\begin{equation*}
		\|\bmM_1\bmM_2\|\leq\|\bmM_1\|\|\bmM_2\|.
	\end{equation*}
	Assume $\{\bmD_t\}$ is a sequence of i.i.d. random matrices.
	Let $a_t=\log\left\{\E{\|\bmD_t\bmD_{t-1}\dots\bmD_1\|^q}\right\}^{1/q}$. 
	For any $t_1,t_2\in\mathbb{N}^{*}$,
	\begin{equation*}
		\|\bmD_{t_1+t_2}\bmD_{t_1+t_2-1}\dots\bmD_1\|^q\leq\|\bmD_{t_1+t_2}\dots\bmD_{t_1+1}\|^q\cdot\|\bmD_{t_1}\dots\bmD_1\|^q.
	\end{equation*}
	By independence,
	\begin{equation*}
		\E{\|\bmD_{t_1+t_2}\bmD_{t_1+t_2-1}\dots\bmD_1\|^q}
		\leq
		\E{\|\bmD_{t_1+t_2}\dots\bmD_{t_1+1}\|^q}\cdot\E{\|\bmD_{t_1}\dots\bmD_1\|^q},
	\end{equation*}
	and hence,
	\begin{equation*}
		a_{t_1+t_2}\leq a_{t_1}+a_{t_2}.
	\end{equation*}
	Therefore, 
	\begin{equation*}
		\gamma_q = \lim_{t\to\infty}\frac{1}{t}\log\mathbb{E}^{1/q}\left(\|\bmD_t\dots\bmD_1\|^q\right)
		=\inf_{t\in\mathbb{N}^{*}}\frac{1}{t}\log\mathbb{E}^{1/q}\left(\|\bmD_t\dots\bmD_1\|^q\right).
	\end{equation*}
\end{proof}
The proofs of the propositions in Section 3 are given below.
\begin{proof}[Proof of Proposition \ref{mmar_str_cond}]
	Under condition \eqref{dt_mmar}, it follows that,
	\begin{equation*}
		\E{\log^{+}\|\bmD_1\|}=
		\sum_{k=1}^K\alpha_k\log^{+}(\|\bm{\Phi}_k\|)<\infty,
	\end{equation*}
	and $\E{\log^{+}\|\bm{\eta}_0\|}<\infty$ because of normality. By Theorem \ref{sre_str}, the results holds.
\end{proof}
\begin{proof}[Proof of Proposition \ref{mmar_ergod}]
	First note that $\{\mathcal{X}_t\}$ is a time homogeneous Markov chain, as it 
	is strictly stationary and its unique stationary solution is given by \eqref{sre_solu}. Define the transition kernel by,
	\begin{equation*}
		\mathrm{P}(\mathcal{X}_t,\cdot)=\Pr(\mathcal{X}_{t+1}\in\cdot|\mathcal{X}_t),
	\end{equation*}
	The $1$-step transition density is,
	\begin{equation*}
		f(\mathcal{X}_{t+1}|\mathcal{X}_{t};\bmtt)
		=f(\bmy_{t+1}|\mathcal{X}_t;\bmtt)
		=\sum_{k=1}^Kf_{t+1}(\bmy_{t+1}|\mathcal{X}_{t};\bmtt_k),
	\end{equation*}
	indicating that $f(\mathcal{X}_{t+1}|\mathcal{X}_{t};\bmtt)>0$ for all $\mathcal{X}_{t+1}$ and $\mathcal{X}_{t}$. Therefore, the Markov chain $\{\mathcal{X}_t\}$ is irreducible and aperiodic. It can be seen that both the $1$-step transition probability and the stationary distribution are equivalent to a Lebesgue measure, hence 	$\mathrm{P}(\mathcal{X}_t,\cdot)$ is absolute continuous with respect to the stationary distribution. Also, the initial distribution is absolute continuous with respect to the stationary distribution. By Theorem 1.1 in \cite{chan1993asymptotic}, the Markov chain $\{\mathcal{X}_t\}$ is ergodic.
\end{proof}
\begin{proof}[Proof of Proposition \ref{mmar_gergod}]
	The proof makes use of Markov chain techniques; see \cite{feigin1985random} for a review with applications to random coefficient VAR models.  Let $q$ be the dimension of $\bm{\eta}_t$.  It follows from the normality assumption of $\bm{E}_t$ that for any subset, say, $A\subset R^q$ with positive Lebesgue measure, $P(\bm{\eta}_t\in A|\bm{\eta}_0=\bm{\nu})>0$ for $t\ge p_{\max}$ regardless of the fixed initial state $\bm{\nu}$. Thus, the MMAR process  is irreducible w.r.t. the Lebesgue measure on the state space of $\bm{\eta}_t$. Moreover, it is Feller continuous so that any compact subset of $R^q$ with positive Lebesgue measure is a small set. 
	
	Note that $\E{\bm{D}_t^\top \otimes \bm{L}_{0,t}^\top} $, $\E{\bm{L}_{0,t} \bm{L}_{0,t}^\top}$, $\E{\mathcal{E}_t^\top \otimes \mathcal{E}_t^\top}$, and $ \E{\bm{L}_t^\top \otimes \bm{L}_t^\top}$ are finite matrices that are independent of $t$.  Let $\bm{W}$ be a fixed $q\times q$ positive definite matrix. Define the $q\times q$ matrix $\bm{V}$ such that $\vecc{\bm{V}}=(\bmI_q -\mathscr{J}^\top)^{-1}\vecc{\bm{W}}$. By Lemma 1 in \cite{feigin1985random}, $\bm{V}$ is a positive-definite matrix. Define $g(\bm{\eta})=1+\bm{\eta}^\top \bm{V} \bm{\eta}$, for $\bm{\eta}\in R^q$. After some algebra (c.f. the proof of Lemma 1 in \cite{feigin1985random}), we have 
	\begin{eqnarray*}
		&&\E{g(\bm{\eta}_t)|\bm{\eta}_{t-1}=\bm{\eta}}\\
		&=&
		g(\bm{\eta})-\bm{\eta}^\top \bm{W} \bm{\eta}+ [2\bm{\eta}^\top \E{\bm{D}_t^\top \otimes \bm{L}_{0,t}^\top} \vecc{\bm{V}}+\\
		&& \quad  \mbox{tr}\{\bm{V} \E{\bm{L}_{0,t} \bm{L}_{0,t}^\top}\}+ 
		\E{\mathcal{E}_t^\top \otimes \mathcal{E}_t^\top} \E{\bm{L}_t^\top \otimes \bm{L}_t^\top} \vecc{\bm{V}}]
	\end{eqnarray*}
	which is bounded above by $g(\bm{\eta})\times (1-c_1)$ for some constant $0<c_1<1$ whenever $|\bm{\eta}|>c_2$ for some $c_2>0$. The rest of the proof is similar to that of Theorem 3 of \cite{feigin1985random}, and hence omitted. 
\end{proof}
\begin{proof}[Proof of Proposition \ref{mmar_1st}]
	Since the MMAR model is a special case of the mixture VAR model with parameter restrictions, let $\bmB_{k}\otimes\bmA_{k}$ play the role of $\Theta_{k1}$ in Theorem 1 of \cite{fong2007mixture} and this proposition is proved.
\end{proof}	
\begin{proof}[Proof of Proposition \ref{mmar_2nd}]
	Similar to the proof of Proposition \ref{mmar_1st}, let $\bmB_{k}\otimes\bmA_{k}$ play the role of $\Theta_{k1}$ in Theorem~3 of \cite{fong2007mixture} and this Proposition is proved.
\end{proof}	

\begin{proof}[Proof of Proposition \ref{mmar_wk_cond}]
	Under condition \eqref{dt_mmar}, it follows that,
	\begin{equation*}
		\E{\log^{+}\|\bmD_1\|}=
		\sum_{k=1}^K\alpha_k\log^{+}(\|\bm{\Phi}_k\|)<\infty,
	\end{equation*}
	and $\E{\log^{+}\|\bm{\eta}_0\|}<\infty$ because of normality. By Theorem \ref{sre_wk}, the results holds.
\end{proof}


\section{Proof of the positive-definiteness of the Fisher information matrix}\label{a_positive}


Partition the (conditional) score vector as follows: $$\frac{\partial l_t}{\partial \vecc{\bm{\Delta}}}= \left(\frac{\partial l_t}{\partial \vecc{\bm{\Delta}_1}^\top}, \cdots, \frac{\partial l_t}{\partial \vecc{\bm{\Delta}}_K^\top}, \frac{\partial l_t}{\partial \alpha_1}, \cdots,\frac{\partial l_t}{\partial \alpha_{K-1} }\right)^\top$$ and let $\bm{\kappa}=(\kappa_1^\top,\cdots,\kappa_K^\top, c_1,\cdots, c_{K-1})^\top $ be a similarly partitioned constant vector.   Then
$\bm{\kappa}^\top \mathcal{I}(\bm{\Delta}_0)\bm{\kappa} =\mathbb{E}\left\{\left( \bm{\kappa}^\top \frac{\partial l_t(\bm{\Delta}_0)}{\partial \vecc{\bm{\Delta}}} \right)^2 \right\}$ which is always non-negative. To verify the positive definiteness of $\mathcal{I}(\bm{\Delta}_0)$, it suffices to show that $\bm{\kappa}^\top \mathcal{I}(\bm{\Delta}_0)\bm{\kappa} =0$ implies that $\bm{\kappa}=0$. Assume then that $\bm{\kappa}^\top \mathcal{I}(\bm{\Delta}_0)\bm{\kappa} =0$.
Since the MMAR model is a VAR model whose  conditional pdf of $\bmy_t$ given $\mathscr{F}_{t-1}$ is positive everywhere, the  stationary joint density of $\bmy_{t-j}, j=0,\cdots, p$ is positive everywhere.  Therefore,  $\bm{\kappa}^\top \frac{\partial l_t(\bm{\Delta}_0)}{\partial \vecc{\bm{\Delta}}}=0$, almost everywhere, and hence $\bm{\kappa}^\top \frac{\partial l_t(\bm{\Delta}_0)}{\partial \vecc{\bm{\Delta}}}\equiv 0$ because  $\bm{\kappa}^\top \frac{\partial l_t(\bm{\Delta}_0)}{\partial \vecc{\bm{\Delta}}}$ is a continuous function as a function of $\bmy_{t-j}, j=0,\cdots, p$. Because $\alpha$'s are positive and $f_t(\bmy_t|\mathscr{F}_{-1};\Delta)>0$ everywhere,
$\bm{\kappa}^\top \frac{\partial l_t(\bm{\Delta}_0)}{\partial \vecc{\bm{\Delta}}}\equiv 0$ is equivalent to 
\begin{eqnarray}
	0&\equiv& 
	\sum_{k=1}^K Q_k\times f_t(\bmy_t|\mathscr{F}_{t-1}; \bm{\Delta}_k),
	\label{eq:indep-poly-normal}
\end{eqnarray}
where $c_K=-\sum_{j=1}^{K-1} c_j$, 
\begin{eqnarray}
	&&Q_k \nonumber\\
	&=&
	\alpha_k\left((\bmy_{t-1}^\top, \cdots, \bmy_{t-p}^\top,1) \otimes \bm{\epsilon}_{t,k}^\top\bm{\Omega}_k^{-1} , -1/2 \times \vech{\bm{\Omega}_k^{-1} -\bm{\Omega}_k^{-1} \bm{\epsilon}_{t,k}\bm{\epsilon}_{t,k}^\top\bm{\Omega}_k^{-1} }^\top \bm{\mathscr{D}}_{mn}^\top \bm{\mathscr{D}}_{mn}\right )\bm{\kappa}_k \nonumber \\
	&& \quad 
	+c_k
	\label{eq:defn of Q_k}
\end{eqnarray}
where $\bm{\epsilon}_{t,k}=\bmy_t-\bm{\Psi}_{k,0}-\sum_{j=1}^p \bm{\Psi}_{k,j}\bmy_{t-j}$. Hence, $Q_k$'s are polynomials of $\bmy_{t-j}, j=0,\cdots,p$ of degree at most two, and 
all parameters are evaluated at their true values. Note  the $f_t(\bmy_t|\mathscr{F}_{t-1}; \bm{\Delta}_k)$'s are conditional normal densities with distinct parameters. 

Next, we shall show that  \eqref{eq:indep-poly-normal} implies that  $Q_k\equiv 0$ for all $k$, which implies that the $Q_k$'s must be the zero polynomial from which it can be readily deduced that $\bm{\kappa}_k=0$ and $c_k=0$, for all $k$. Thus, the Fisher information matrix $\mathcal{I}(\bm{\Delta}_0)$ is positive definite. 
Finally, we prove that \eqref{eq:indep-poly-normal} entails that $Q_k\equiv 0$, for all $k$, by mathematical induction on $K$. Clearly, the claim is valid for $K=1$. Suppose that the claim is true for all $K\le K^*$. Consider the case that $K=K^*+1$. Rewrite \eqref{eq:indep-poly-normal} as follows
\begin{eqnarray}
	-Q_{K^*+1}&\equiv& 
	\sum_{k=1}^{K^*} Q_k\times f_t(\bmy_t|\mathscr{F}_{t-1}; \bm{\Delta}_k)/f_t(\bmy_t|\mathscr{F}_{t-1}; \bm{\Delta}_{K^*+1}).
	\label{eq:indep-poly-normal-rewrite}
\end{eqnarray}
Because $Q_{K^*+1}$ is a polynomial in $\bmy_{t-j}, j=0,\cdots, p$ of degree at most equal to 2, we can apply finitely many partial differentiations w.r.t.  $\bmy_{t-j}, j=0,\cdots, p$  to annihilate it. Because the $\bm{\Delta}_k$'s are distinct, the ratios \newline $f_t(\bmy_t|\mathscr{F}_{t-1}; \bm{\Delta}_k)/f_t(\bmy_t|\mathscr{F}_{t-1}; \bm{\Delta}_{K^*+1})$ equal to some multiples of an exponential function whose exponent is a non-zero polynomial (in  $\bmy_{t-j}, j=0,\cdots, p$) of degree at most equal to 2. Applying the said partial differentiation operators to annihilate $Q_{K^*+1}$ on both sides of \eqref{eq:indep-poly-normal-rewrite} then yields 
\begin{eqnarray}
	0&\equiv& 
	\sum_{k=1}^{K^*} (Q_k\times H_k+G_k) f_t(\bmy_t|\mathscr{F}_{t-1}; \bm{\Delta}_k)/f_t(\bmy_t|\mathscr{F}_{t-1}; \bm{\Delta}_{K^*+1}).
	\label{eq:indep-poly-normal-rewrite-2}
\end{eqnarray}
where $H_k$ is a non-zero polynomial and $G_k$ is a polynomial of degree strictly less than that of $Q_k\times H_k$. Consequently, we have $0\equiv \sum_{k=1}^{K^*} (Q_k\times H_k+G_k) \times f_t(\bmy_t|\mathscr{F}_{t-1}; \bm{\Delta}_k)$. Thus $Q_k\times H_k+G_k\equiv 0, \forall 1\le k\le K^*$, resulting in $Q_k\equiv 0, \forall k\le K^*$. But then $Q_{K^*+1}\equiv 0$. 


\section{Proof of Theorem 1}\label{a-thm1}
\begin{proof}
	It can be readily seen from the proof of \citep[Theorem 5]{redner1981note} that the strong consistency of the MLE holds for the conditional MLE with stationary ergodic vector data. Hence, we only need to verify the validity of the conditions required by Theorem 5 there. Assumption 4 there requires that for all $\bm{\Psi_0}^*=\bm{\Psi_0}(\bm{\theta}^*), \bm{\Psi}^*=\bm{\Psi}(\bm{\theta}^*), \bm{\Omega}^*=\bm{\Omega}(\bm{\theta}^*), \bm{\Psi_0}=\bm{\Psi_0}(\bm{\theta}), \bm{\Psi}=\bm{\Psi}(\bm{\theta}), \bm{\Omega}=\bm{\Omega}(\bm{\theta})$ where $\bm{\Omega}^*, \bm{\Omega}$ are positive definite matrices,
	$$
	\mathbb{E}_{\bm{\Psi_0}^*, \bm{\Psi}^*, \bm{\Omega}^*} \{|\log f_\mathcal{N}
	(\vecc{Y_t}| \bm{\Psi_0}+\bm{\Psi}\mathcal{X}_t,\bm{\Omega}\})|\}<\infty,$$
	where the subscripts of the expectation operator signify the true parameters. The finite second moment assumption implies that 
	$\mathbb{E}_{\bm{\Psi_0}^*, \bm{\Psi}^*, \bm{\Omega}^*} (\mathcal{X}_t \mathcal{X}_t^\tp)$ is a finite matrix. 
	Routine algebra shows that 
	\begin{eqnarray*}
		&&\mathbb{E}_{\bm{\Psi_0}^*, \bm{\Psi}^*, \bm{\Omega}^*} \{|\log f_\mathcal{N}
		(\vecc{\bm{Y}_t}| \bm{\Psi_0}+\bm{\Psi}\mathcal{X}_t,\bm{\Omega}\})|\} \\
		&\le & 1/2\times [K+ |\log |\bm{\Omega}^{-1}||+\mbox{trace}(\bm{\Omega}^{-1} \bm{\Omega}^*)+ \\
		&& \quad \mathbb{E}_{\bm{\Psi_0}^*, \bm{\Psi}^*, \bm{\Omega}^*} \{\bm{\Psi}_0^*-\bm{\Psi}_0+(\bm{\Psi}^*-\bm{\Psi})\mathcal{X}_t)^\tp \bm{\Omega}^{-1} (\bm{\Psi}_0^*-\bm{\Psi}_0+(\bm{\Psi}^*-\bm{\Psi})\mathcal{X}_t\} ]\\
		&<&\infty,
	\end{eqnarray*}
	where $K$ is a finite constant. Assumption 2a there can be similarly verified and assumptions 1 and 5 obviously hold. This completes the proof.
\end{proof}

\section{Proof of Theorem~\ref{Thm: 2}}  \label{a_thm2}
\begin{proof}
	Let $\bm{Z}\sim \mathcal{N}(\bm{0}, \mathcal{I}(\bm{\Delta}^0))$. 
	Theorem 4.4 in \cite{geyer1994asymptotics} states that under some conditions (Assumptions A-D listed there as well as the Chernoff regularity of the constrained set at the true parameter value), the followings hold with suitable modifications upon noting that the objective function in \cite{geyer1994asymptotics} is the negative conditional log-likelihood:
	\newline (1) The constrained MLE $\hat{\bm{\Delta}}$ is $\sqrt{T}$-consistent. 
	\newline (2) $\sqrt{T}( \hat{\bm{\Delta}} -\bm{\Delta}^0)$  converges in distribution to $\argmin_{\bm{\Delta}\in \mathscr{\bm{C}}} q_{\bm{Z}}(\bm{\Delta})  $ where $\mathscr{\bm{C}}$ is the tangent cone of $\bm{\Delta}(\bm{C})$ at $\bm{\Delta}^0$, 
	and $q_{\bm{Z}}(\bm{\Delta})=-\bm{\Delta}^\top \bm{Z} +\frac{1}{2} \bm{\Delta}^\top \mathcal{I}(\bm{\Delta}^0)\bm{\Delta}$. Let $\bm{\Delta}_{\bm{Z}}=\bm{J}(\bm{J}^\top \mathcal{I}(\bm{\Delta}^0) \bm{J})^{-}\bm{J}^\top \bm{Z} $. Note that $\bm{J}$ is a reduced-rank matrix of dimension $\mbox{dim}(\bm{\Delta})\times \mbox{dim}(\bm{\theta})$. Consider a rank factorization of $\bm{J}=\bm{F}\bm{Q}$  into a product of two full-rank matrices such that the ranks of $\bm{J}, \bm{F}, \bm{Q}$ are the same. Furthermore, $\bm{Q}$ can and will be chosen as an orthonormal matrix, i.e., $\bm{Q}\bm{Q}^\top=\bm{I}$. Thus,
	\begin{equation}
		\bm{J}(\bm{J}^\top \mathcal{I}(\bm{\Delta}^0) \bm{J})^{-}\bm{J}^\top= \bm{F}(\bm{F}^\top\mathcal{I}(\bm{\Delta}^0)\bm{F})^{-1} \bm{F}^\top. 
		\label{eq:reflexive}
	\end{equation}
	The proof of this claim is deferred to the end of the proof.   Below, we show that the solution of the constrained minimization problem is  $\bm{\Delta}_{\bm{Z}}$, hence the stated limiting normal distribution for $\sqrt{T}( \hat{\bm{\Delta}} -\bm{\Delta}^0)$ in Theorem~\ref{Thm: 2}.
	\newline 
	(3) $L_T(\hat{\bm{\Delta}})-L_T(\bm{\Delta}^0)$ converges in distribution to \newline $-q_{\bm{Z}}(\bm{\Delta}_{\bm{Z}})=\frac{1}{2}\bm{Z}^\top \bm{J}(\bm{J}^\top \mathcal{I}(\bm{\Delta}^0) \bm{J})^{-}\bm{J}^\top \bm{Z}=\frac{1}{2}\bm{Z}^\top \bm{F}(\bm{F}^\top\mathcal{I}(\bm{\Delta}^0)\bm{F})^{-1} \bm{F}^\top\bm{Z}$. Hence, $2\{L_T(\hat{\bm{\Delta}})-L_T(\bm{\Delta}^0)\}$ is asymptotically $\chi^2$ distributed with the degrees of freedom equal to the rank of $\bm{J}$, under the true model. 
	
	We now prove the claim on the solution of the constrained optimization problem in (2). 
	Because $\bm{C}$ contains an open ball centered at $\bm{\theta}^0$, the tangent cone $\mathscr{\bm{C}}=\{\bm{J}\bm{\theta}: \bm{\theta}\in R^q\}$.  It is clear that the tangent cone is same for all $\bm{\theta}\in [\bm{\theta}_0]$, even though $\bm{J}$ does depend on $\bm{\theta}\in [\bm{\theta}_0]$. Note that $\bm{J}$ is not of full-rank because of the over parameterization. The constrained minimization of  $q_{\bm{Z}}(\cdot)$ can be solved by first solving the unconstrained minimization problem:  
	$\argmin_{\bm{\theta}} q(\bm{\theta})$ where $q(\bm{\theta})=-\bm{\theta}^\top \bm{J}^\top \bm{Z}+\frac{1}{2} \bm{\theta}^\top \bm{J}^\top \mathcal{I}(\bm{\Delta}^0) \bm{J}\bm{\theta}$. The solution satisfies the optimality equation
	$$
	0=\frac{\partial q}{\partial \bm{\theta}}=-\bm{J}^\top \bm{Z}+\bm{J}^\top \mathcal{I}(\bm{\Delta}^0) \bm{J} \bm{\theta}.
	$$
	which admits multiple solutions, namely, $(\bm{J}^\top \mathcal{I}(\bm{\Delta}^0) \bm{J})^{-}\bm{J}^\top \bm{Z}$. Hence, the solution of the original constrained minimization problem is \newline $\bm{\Delta}_{\bm{Z}}=\bm{J}(\bm{J}^\top \mathcal{I}(\bm{\Delta}^0) \bm{J})^{-}\bm{J}^\top \bm{Z} $ which is unique. Hence, $\sqrt{T}( \hat{\bm{\Delta}} -\bm{\Delta}^0)$  converges in distribution to $\bm{\Delta}_{\bm{Z}}$ which is multivariate normal with zero mean and covariance matrix equal to $\bm{J}(\bm{J}^\top \mathcal{I}(\bm{\Delta}^0) \bm{J})^{-}\bm{J}^\top \mathcal{I}(\bm{\Delta}^0)\bm{J}(\bm{J}^\top \mathcal{I}(\bm{\Delta}^0) \bm{J})^{-}\bm{J}^\top$ which can be simplified to $\bm{J}(\bm{J}^\top \mathcal{I}(\bm{\Delta}^0) \bm{J})^{-}\bm{J}^\top$, owing to \eqref{eq:reflexive}.  This completes the proof of the claim. 
	
	It remains to verify Assumptions A-D. Assumption A holds, thanks to \eqref{eq:M-asym1} and upon noting that the objective function used by \cite{geyer1994asymptotics} is the negative log-likelihood. Assumption B concerns the  remainder term in \eqref{eq:M-asym2} and it requires that there exists a neighborhood $\mathcal{N}$ of $\bm{\Delta}^0$ such that the empirical process of $r_t(\bm{\Delta}), \bm{\Delta}\in \mathcal{N}$ is stochastically equicontinuous.  
	
	Recall $r_t(\bm{\Delta})= \left(\frac{\partial l_t(\bm{\Delta}^*)}{\partial \bm{\Delta}^\top}- \frac{\partial l_t(\bm{\Delta}^0)}{\partial \bm{\Delta}^\top}\right)(\bm{\Delta}-\bm{\Delta}^0)/|\bm{\Delta}-\bm{\Delta}^0|$, for $\bm{\Delta} \not = \bm{\Delta}^0$ where $|\bm{\Delta}^*-\bm{\Delta}^0|\le |\bm{\Delta}-\bm{\Delta}^0|$ and otherwise, it equals zero.  Thus, it suffices to show the stochastic equicontinuity of the empirical process of  $H(\nu)=\{\frac{\partial l_t(\bm{\Delta})}{\partial \bm{\Delta}^\top}: |\bm{\Delta}-\bm{\Delta}^0|<\nu\}$ for some $\nu>0$. 
	Note that the conditional mixture normal specification of the MMAR model implies that (i) it follows from \eqref{eq:lt-d0}-\eqref{eq:lt-alpha} that $H$ is a  V-C class of functions and (ii) for any fixed $\nu>0$, the functions in $H(\nu)$ is uniformly bounded above by some integrable function. See \cite{pollard2012convergence} for the definition and properties of the V-C class. Furthermore, the true MMAR model is strictly stationary and geometric ergodic, hence $\beta$-mixing with a geometric decaying mixing rate.  Consequently, it follows from \citep[Theorem 3.4]{yu1994rates} that for any fixed $\nu>0$, the empirical processes of $H(\nu)$ is stochastically equicontinuous. We conclude that Assumption B holds. 
	
	Assumption C trivially holds because of the validity of the Central Limit Theorem for the conditional score vector.  Assumption D holds as $\hat{\Delta}$ is strongly consistent, from Theorem 1. Chernoff regularity of the constrained set at the true value holds because  $\bm{\Delta}(\bm{C})$ is a manifold around $\bm{\Delta}^0$; see (2.1) and (2.2) in \cite{geyer1994asymptotics} for the definition of Chernoff regularity. 
	
	It remains to verify \eqref{eq:reflexive}. We have 
	\begin{equation*}
		\bm{J}(\bm{J}^\top \mathcal{I}(\bm{\Delta}^0) \bm{J})^{-}\bm{J}^\top= \bm{F}\bm{Q}(\bm{Q}^\top \bm{F}^\top\mathcal{I}(\bm{\Delta}^0)\bm{F}\bm{Q})^{-} \bm{Q}^\top\bm{F}^\top. 
	\end{equation*}
	Because $\mathcal{I}(\bm{\Delta}^0)$ is positive definite, the ranks of $\bm{F}^\top\mathcal{I}(\bm{\Delta}^0)\bm{F}$, $\bm{Q}$, $\bm{Q}^\top$ are the same. It follows  from \citep[Lemma 2.2.6 (g)]{rao1971generalized}  that $\bm{Q}(\bm{Q}^\top \bm{F}^\top\mathcal{I}(\bm{\Delta}^0)\bm{F}\bm{Q})^{-} \bm{Q}^\top$ is invariant w.r.t. the version of the generalized inverse. One version of $(\bm{Q}^\top \bm{F}^\top\mathcal{I}(\bm{\Delta}^0)\bm{F}\bm{Q})^{-}$ is $\bm{Q}^\top \{\bm{F}^\top\mathcal{I}(\bm{\Delta}^0)\bm{F}\}^{-1}\bm{Q}$. Thus 
	\begin{eqnarray*}
		&&\bm{J}(\bm{J}^\top \mathcal{I}(\bm{\Delta}^0) \bm{J})^{-}\bm{J}^\top\\
		&=& \bm{F}\bm{Q}\bm{Q}^\top \{\bm{F}^\top\mathcal{I}(\bm{\Delta}^0)\bm{F}\}^{-1}\bm{Q} \bm{Q}^\top\bm{F}^\top\\
		&=& \bm{F} \{\bm{F}^\top\mathcal{I}(\bm{\Delta}^0)\bm{F}\}^{-1}\bm{F}^\top,
	\end{eqnarray*}
	since $\bm{Q}\bm{Q}^\top=\bm{I}$. 
	This completes the proof of the theorem.
\end{proof}

\section{Proof of Theorem~\ref{Thm: ident}}\label{a_thm3}
\begin{proof}
	Recall from the proof of Theorem~\ref{Thm: 2} 
	that 
	$\sqrt{T}( \hat{\bm{\Delta}} -\bm{\Delta}^0)$  converges in distribution to $\argmin_{\bm{\Delta}\in \mathscr{\bm{C}}} q_{\bm{Z}}(\bm{\Delta})  $ where $\mathscr{\bm{C}}$ is the tangent cone of $\bm{\Delta}(\bm{C})$ at $\bm{\Delta}^0$, 
	and $q_{\bm{Z}}(\bm{\Delta})=-\bm{\Delta}^\top \bm{Z} +\frac{1}{2} \bm{\Delta}^\top \mathcal{I}(\bm{\Delta}^0)\bm{\Delta}$. Due to the identifiability constraints $\bm{c}(\bm{\theta})=0$, the tangent cone becomes $\mathscr{C}=\{\bm{J}\bm{\theta}: \bm{W}\bm{\theta}=0, \bm{\theta}\in R^q\}$
	where $\bm{W}=\frac{\partial \bm{C}(\bm{\theta}^0)}{\partial \bm{\theta}^\top}$. 
	Then we need to solve the following constrained optimization problem 
	$\argmin_{\bm{\theta}: \bm{W} \bm{\theta}=0} q(\bm{\theta})$ where $$q(\bm{\theta})=-\bm{\theta}^\top \bm{J}^\top \bm{Z}+\frac{1}{2} \bm{\theta}^\top \bm{J}^\top \mathcal{I}(\bm{\Delta}^0) \bm{J}\bm{\theta}.$$
	We solve the problem via the method of Lagrange multiplier. Define
	$$
	q(\bm{\theta}, \bm{\lambda})=-\bm{\theta}^\top \bm{J}^\top \bm{Z}+\frac{1}{2} \bm{\theta}^\top \bm{J}^\top \mathcal{I}(\bm{\Delta}^0) \bm{J}\bm{\theta} +\bm{\lambda}^\top \bm{W}\bm{\theta}. 
	$$
	Then 
	$$
	\begin{pmatrix}
		\frac{\partial q}{\partial \bm{\theta}} \\
		\frac{\partial q}{\partial \bm{\lambda}} 
	\end{pmatrix}
	= \begin{pmatrix} \bm{J}^\top  \mathcal{I}(\bm{\Delta}^0) \bm{J} & \bm{W}^\top \\
		\bm{W} & \bm0 
	\end{pmatrix} 
	\begin{pmatrix}
		\bm{\theta} \\
		\bm{\lambda}
	\end{pmatrix} -\begin{pmatrix}
		\bm{J}^\top \bm{Z} \\ \bm0
	\end{pmatrix}.
	$$
	Setting the partial derivatives to zero yields
	$$
	\begin{pmatrix} \bm{J}^\top  \mathcal{I}(\bm{\Delta}^0) \bm{J} & \bm{W}^\top \\
		\bm{W} & \bm0 
	\end{pmatrix} 
	\begin{pmatrix}
		\bm{\theta} \\
		\bm{\lambda}
	\end{pmatrix} =\begin{pmatrix}
		\bm{J}^\top \bm{Z} \\\bm0
	\end{pmatrix}.
	$$
	The solution must then satisfy the equation $\bm{W}\bm{\theta}=0$ hence, the preceding optimality equation is equivalent to 
	$$
	\begin{pmatrix} \bm{J}^\top  \mathcal{I}(\bm{\Delta}^0) \bm{J} +\bm{W} \bm{W}^\top& \bm{W}^\top \\
		\bm{W} & \bm0 
	\end{pmatrix} 
	\begin{pmatrix}
		\bm{\theta} \\
		\bm{\lambda}
	\end{pmatrix} =\begin{pmatrix}
		\bm{J}^\top \bm{Z} \\ \bm0
	\end{pmatrix},
	$$
	where the upper left block matrix in the $2\times 2$ block matrix is now invertible. Denote the inverse of the $2\times 2$ block matrix as $$\begin{pmatrix}
		\bm{P} & \bm{Q}^\top \\
		\bm{Q} & \bm{R}
	\end{pmatrix}.$$ Then the $\bm{\theta}$-component of the solution is 
	$\bm{P} \bm{J}^\top \bm{Z}$ which is $\mathcal{N}(\bm0,\bm{P})$, after routine algebra; c.f. \citep[Section 10 of Chapter 3]{silvey2017statistical}. Under the identifiability constraints, it is assumed that $\bm{\Delta}$ and $\bm{\theta}$ are locally diffeomorphic re-parameterization of each other, thus the constrained MLE $\hat{\bm{\theta}}$ is also $\sqrt{T}$-consistent. Thus, $\sqrt{T}(\hat{\bm{\theta}}-\bm{\theta}^0)$ converges weakly to  $\bm{P} \bm{J}^\top \bm{Z}$.
\end{proof}

\section{Additional Simulation Results}\label{a_simu}
The coefficient matrices for all simulation models can be viewed in the R code available to the public. 
Besides Scenarios 1 and 2, we add the following two simulation scenarios.
\begin{itemize}
	\item{Scenario 3:}  An MMAR(2;2,2) with $(m,n)=(2,3)$. 
	\item{Scenario 4:}  An MMAR(3;1,1,1) with $(m,n)=(4,5)$. 
\end{itemize}
In Scenario 3, the mixing weights are set to be $(\alpha_1,\alpha_2) = (0.4 ,0.6)$, and the parameter matrices are generated similarly to Scenarios 1 and 2. Both components are weakly stationary as $\rho(\bm{\Phi}_1)=0.8399<1$ and $\rho(\bm{\Phi}_2)=0.6691<1$, and so is the overall model. In Scenario 4, the mixing weights are set to be $(\alpha_1,\alpha_2,\alpha_3) = (0.1 ,0.2,0.7)$ with which we can compare the effect of small mixing weights. The first and the third components are stationary as $\rho(\bmB_{1,1}\otimes\bmA_{1,1})= 0.6682$ and $\rho(\bmB_{3,1}\otimes\bmA_{3,1})=0.6537$. The second component is not stationary as $\rho(\bmB_{2,1}\otimes\bmA_{2,1})=1.0136$. However, the overall model is stationary.
For Scenario 3, we also compare the performance of selecting $K$ when the AR orders are misspecified by setting $p_{\max}=1$. We select the models for $K\in\{1,2,3\}$ in Scenario 3, and $K\in\{1,2,3,4\}$ in Scenario 4. For both scenarios, we select the AR orders up to 3.

\begin{table}[ht!]
	\centering
	\begin{tabular}{lllll}
		\hline
		& AIC &BIC&HQ&GIC \\
		\hline
		$T=200$& 36.60\% &93.60\%& 78.60\%&99.80\%\\
		$T=400$&15.80\%  & 97.80\%&82.20\% &99.80\% \\
		$T=800$&16.20\%&98.40\%&86.40\%&100.00\%\\
		\hline
	\end{tabular}
	\caption{\label{Sc3K}Percentage of correctly selecting $K=2$ in Scenario 3 with the AR orders given.}
\end{table}

\begin{table}[ht!]
	\centering
	\begin{tabular}{lllll}
		\hline
		& AIC &BIC&HQ&GIC \\
		\hline
		$T=200$& 67.20\% &66.80\%& 70.20\%&61.20\%\\
		$T=400$&62.00\%  & 95.60\%&92.60\% &95.60\% \\
		$T=800$&45.40\%&99.20\%&97.60\%&99.20\%\\
		\hline
	\end{tabular}
	\caption{\label{Sc4K}Percentage of correctly selecting $K=3$ in Scenario 4 with the AR orders given.}
\end{table}

\begin{table}[ht!]
	\centering
	\begin{tabular}{lllll}
		\hline
		& AIC &BIC&HQ&GIC \\
		\hline
		$T=200$& 2.20\% &95.20\%& 68.20\%&99.60\%\\
		$T=400$& 0.20\%  & 97.40\%&56.00\% &100.00\% \\
		$T=800$&0.00\%&86.40\%&13.00\%&100.00\%\\
		\hline
	\end{tabular}
	\caption{\label{Sc3K_1}Percentage of correctly selecting $K=2$ in Scenario 3 with the AR orders misspecified ($p_{\max}$ is set to be 1).}
\end{table}
\begin{table}[ht!]
	\centering
	\begin{tabular}{lllll}
		\hline
		& AIC &BIC&HQ&GIC \\
		\hline
		$T=200$& 60.60\% &100.00\%& 99.80\%&100.00\%\\
		$T=400$&74.20\%  & 100.00\%&100.00\% &100.00\% \\
		$T=800$&79.60\%&100.00\%&100.00\%&100.00\%\\
		\hline
	\end{tabular}
	\caption{\label{Sc3P}Percentage of correctly selecting $p_{\max}=2$ in Scenario 3 with the number of components $K$ given.}
\end{table}
\begin{table}[ht!]
	\centering
	\begin{tabular}{lllll}
		\hline
		& AIC &BIC&HQ&GIC \\
		\hline
		$T=200$& 71.80\% &93.40\%& 87.80\%&95.60\%\\
		$T=400$&38.00\%  & 98.20\%&97.60\% &98.40\% \\
		$T=800$&11.20\%&99.60\%&99.40\%&99.60\%\\
		\hline
	\end{tabular}
	\caption{\label{Sc4P}Percentage of correctly selecting $p_{\max}=1$ in Scenario 4 with the number of components $K$ given.}
\end{table}
The GIC performs well in generally.
However, one exception is observed in Scenario 4 with small sample size ($N=200$). In this case, none of these methods achieve a desired level of performance as the percentage of correctly selecting $K$ falls between 60\% and 70\%.
This could be attributed to the influence of the components with small mixing weights. In this case the GIC has the worst performance, indicating that it could be conservative when the true model contains components with small mixing weights.

We also examine settings where the data are generated from MAR models. Specifically, we consider the following two scenarios:
\begin{itemize}
	\item{Scenario 5:}  An MAR(1) with $(m,n)=(2,3)$. 
	\item{Scenario 6:}  An MAR(1) with $(m,n)=(4,5)$. 
\end{itemize}
In Scenario 5, the data are generated from the first component of the model in Scenario 1. Similarly, in Scenario 6, the data are generated from the first component of the model in Scenario 2. We evaluate the performance of model selection criteria in correctly identifying the MAR(1) models over the MMAR models. In each of scenarios 5 and 6, 500 independent time series of length $T+20$ were generated, where $T\in \{200, 400, 800\}$. The first $T$ observations are used for model fitting, and the remaining 20 observations are reserved for out-of-sample prediction.
The model selection results for each information criterion are summarized in the following tables.
\begin{table}[ht!]
	\centering
	\begin{tabular}{lllll}
		\hline
		& AIC &BIC&HQ&GIC \\
		\hline
		$T=200$& 7.80\% &97.20\%& 82.60\%&98.80\%\\
		$T=400$& 9.20\%  & 97.20\%&88.40\% &100.00\% \\
		$T=800$&17.20\%&98.80\%&94.40\%&100.00\%\\
		\hline
	\end{tabular}
	\caption{Percentage of correct selection of MAR(1) model in Scenario 5.}
\end{table}
\begin{table}[ht!]
	\centering
	\begin{tabular}{lllll}
		\hline
		& AIC &BIC&HQ&GIC \\
		\hline
		$T=200$& 7.80\% &99.80\%& 78.80\%&100.00\%\\
		$T=400$& 3.20\%  & 100.00\%&87.80\% &100.00\% \\
		$T=800$&5.00\%&100.00\%&93.80\%&100.00\%\\
		\hline
	\end{tabular}
	\caption{Percentage of correct selection of MAR(1) model in Scenario 6.}
\end{table}
Overall, the BIC and GIC demonstrate strong performance, correctly selecting the true model with a high percentage across the considered scenarios.

Fig.~\ref{out-sample-prediction-boxplot} compares the MAR(1) with the MAR(2;1,1) model in terms of the out-sample mean squared 1-step prediction errors (MPE) for 20 observations generated by an MAR(1) model specified in Scenario 5. 
Fig.~\ref{out-sample-prediction-boxplot-2} shows the comparison when the data were generated according to Scenario 6. These plots show that mis-specifying the number of components to a higher number slightly inflates the prediction error, which is consistent with Theorem~\ref{mmar_con} that the MLE is consistent even if $K$ is higher than the true value. 

\begin{figure}[h]
	\begin{center}
		\scalebox{0.20}{\includegraphics{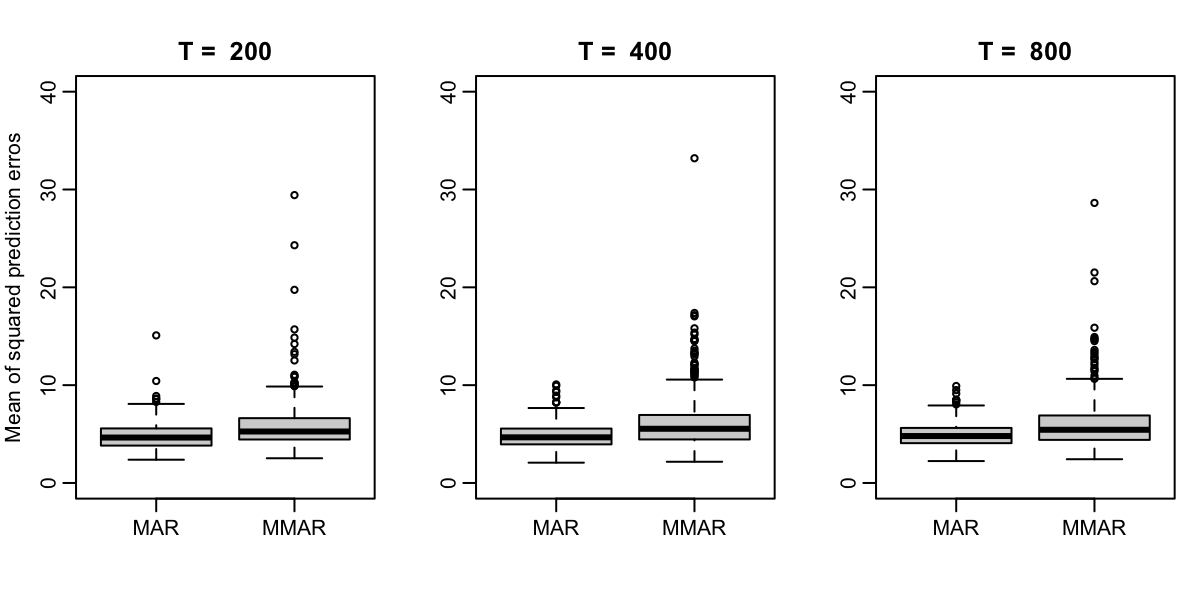}}
	\end{center}
	\vspace{-0.2cm}
	\caption{\label{out-sample-prediction-boxplot} Comparison  of the MAR(1) with the MAR(2;1,1) model in terms of the out-sample mean squared 1-step prediction errors (MPE) for 20 observations generated by an MAR(1) model specified in Scenario 5. In each sub-graph, left box-plot displays the distribution of the MPE from the fitted MAR(1) model and right box that of the fitted MMAR(2;1,1) model.}
\end{figure}
\begin{figure}[h]
	\begin{center}
		\scalebox{0.20}{\includegraphics{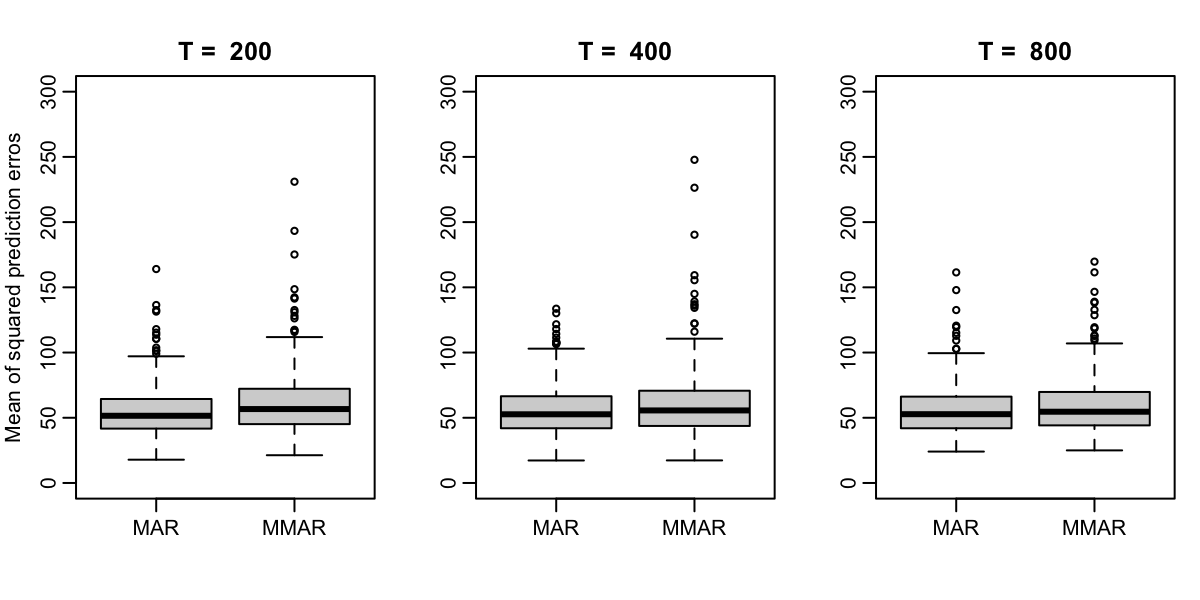}}
	\end{center}
	\vspace{-0.2cm}
	\caption{\label{out-sample-prediction-boxplot-2} Comparison  of the MAR(1) with the MAR(2;1,1) model in terms of the out-sample mean squared 1-step prediction errors (MPE) for 20 observations generated by an MAR(1) model specified in Scenario 6. In each sub-graph, left box-plot  displays the distribution of the MPE from the fitted MAR(1) model and right box that of the fitted MMAR(2;1,1) model.}
\end{figure}

\section{Additional Results of Real Data Analysis}\label{a_da}
\subsection{Additional Analysis Under MMAR(2;1,1) Model}
\begin{figure}[h]
	\begin{center}
		\scalebox{0.13}{\includegraphics{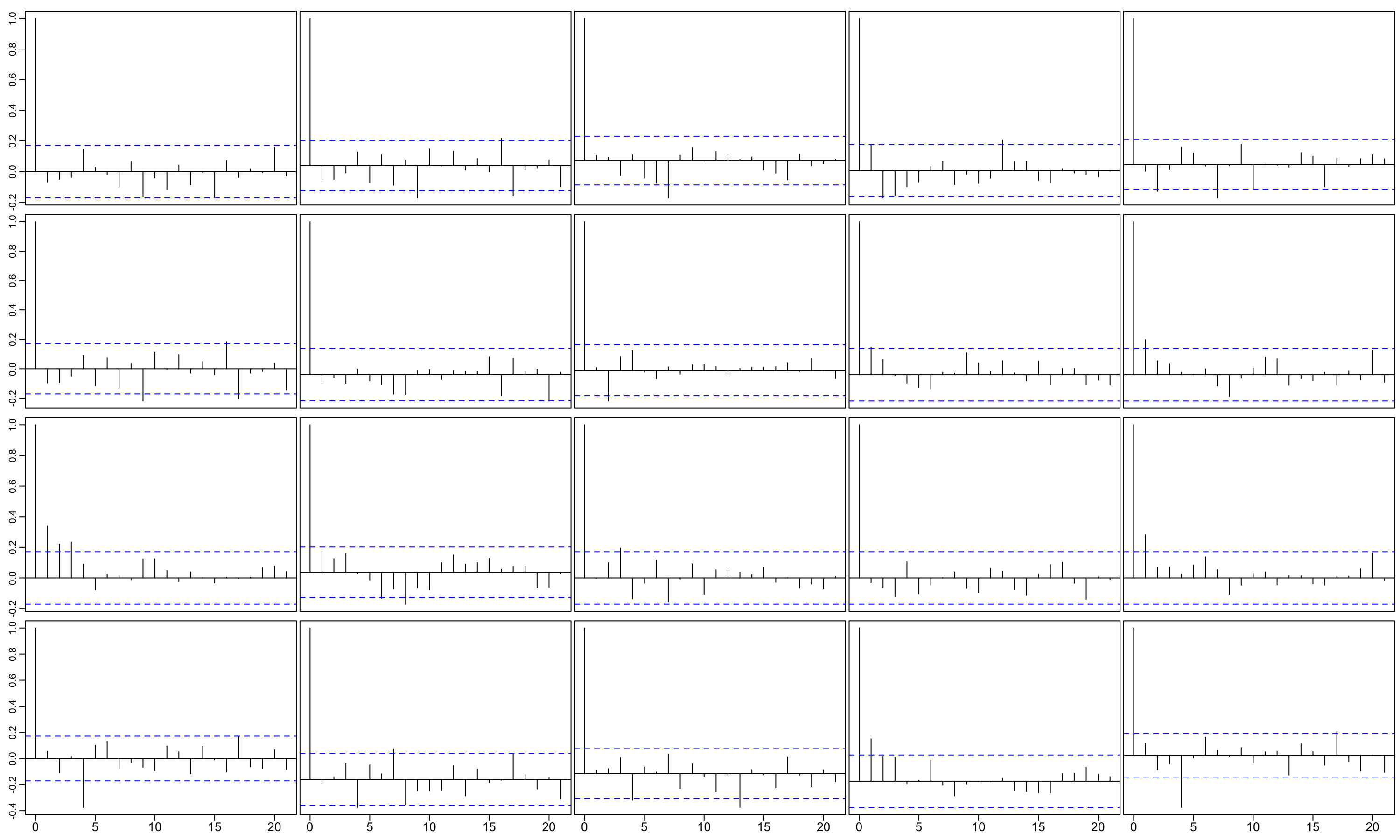}}
	\end{center}
	\vspace{-0.2cm}
	\caption{\label{k2_acf} ACF of standardized residuals of MMAR(2;1,1) model.}
\end{figure}
\begin{table}[ht]
	\centering
	\begin{tabular}{lllll||llll}
		\hline
		& Int & GDP & Prod & CPI & Int & GDP & Prod & CPI\\
		\hline
		Int&\makecell{-1.952\\(0.784)} & \makecell{-0.057\\(1.798)} & \makecell{0.699\\(1.508)} & \makecell{1.419\\(0.835)} & - & 0 & 0 & 0 \\ 
		GDP&\makecell{2.045\\(0.416)} & \makecell{4.156\\(0.968)} & \makecell{9.603\\(0.826)} & \makecell{-0.982\\(0.426)} & + & + & + & - \\ 
		Prod&\makecell{1.041\\(0.44)} & \makecell{3.644\\(1.034)} & \makecell{5.561\\(0.865)} & \makecell{0.296\\(0.466)} & + & + & + & 0 \\ 
		CPI&\makecell{-0.623\\(0.657)} & \makecell{-0.642\\(1.516)} & \makecell{3.761\\(1.268)} & \makecell{0.703\\(0.691)} & 0 & 0 & + & 0 \\ 
		\hline
	\end{tabular}
	\caption{\label{K2_econ_A1}The MLE of $\bmA_{1,1}$ of the MMAR(2;1,1) model, with standard errors given in parentheses.}
\end{table}

\begin{table}[ht]
	\centering
	\begin{tabular}{llllll||lllll}
		\hline
		&USA& DEU& FRA& GBR &CAN &USA& DEU& FRA& GBR &CAN \\ 
		\hline
		USA&\makecell{0.115\\(0.025)} & \makecell{-0.037\\(0.024)} & \makecell{0.198\\(0.022)} & \makecell{-0.14\\(0.017)} & \makecell{-0.131\\(0.021)} & + & 0 & + & - & - \\ 
		DEU& \makecell{0.132\\(0.029)} & \makecell{-0.034\\(0.028)} & \makecell{0.26\\(0.025)} & \makecell{-0.179\\(0.02)} & \makecell{-0.175\\(0.024)} & + & 0 & + & - & - \\ 
		FRA&\makecell{0.162\\(0.029)} & \makecell{-0.048\\(0.027)} & \makecell{0.356\\(0.023)} & \makecell{-0.248\\(0.02)} & \makecell{-0.247\\(0.024)} & + & 0 & + & - & - \\ 
		GBR &\makecell{0.16\\(0.03)} & \makecell{-0.076\\(0.027)} & \makecell{0.439\\(0.025)} & \makecell{-0.301\\(0.022)} & \makecell{-0.239\\(0.025)} & + & - & + & - & - \\ 
		CAN &\makecell{0.083\\(0.032)} & \makecell{-0.034\\(0.03)} & \makecell{0.242\\(0.028)} & \makecell{-0.17\\(0.021)} & \makecell{-0.124\\(0.027)} & + & 0 & + & - & - \\ 
		\hline
	\end{tabular}
	\caption{The MLE of $\bmB_{1,1}$ of the MMAR(2;1,1) model, with standard errors given in parentheses.}
\end{table}

\begin{table}[ht]
	\centering
	\begin{tabular}{llllll||lllll}
		\hline
		&USA& DEU& FRA& GBR &CAN &USA& DEU& FRA& GBR &CAN\\
		\hline
		Int&\makecell{-0.397\\(0.297)} & \makecell{-0.334\\(0.345)} & \makecell{-0.157\\(0.343)} & \makecell{-0.264\\(0.342)} & \makecell{0.033\\(0.374)} & 0 & 0 & 0 & 0 & 0 \\ 
		GDP&\makecell{-0.069\\(0.152)} & \makecell{0.014\\(0.176)} & \makecell{0.087\\(0.175)} & \makecell{0.066\\(0.177)} & \makecell{-0.06\\(0.191)} & 0 & 0 & 0 & 0 & 0 \\ 
		Prod&\makecell{-0.22\\(0.166)} & \makecell{-0.459\\(0.192)} & \makecell{-0.31\\(0.192)} & \makecell{0.108\\(0.191)} & \makecell{-0.25\\(0.207)} & 0 & - & 0 & 0 & 0 \\ 
		CPI&\makecell{-0.022\\(0.248)} & \makecell{0.009\\(0.288)} & \makecell{0.208\\(0.289)} & \makecell{0.553\\(0.285)} & \makecell{0.269\\(0.313)} & 0 & 0 & 0 & 0 & 0 \\ 
		\hline
	\end{tabular}
	\caption{The MLE of $\bmC_{1,1}$ of the MMAR(2;1,1) model, with standard errors given in parentheses.}
\end{table}

\begin{table}[ht]
	\centering
	\begin{tabular}{lllll||llll}
		\hline
		& Int & GDP & Prod & CPI & Int & GDP & Prod & CPI\\
		\hline
		Int& \makecell{1.426\\(0.053)} & \makecell{0.506\\(0.075)} & \makecell{0.075\\(0.083)} & \makecell{0.259\\(0.047)} & + & + & 0 & + \\ 
		GDP& \makecell{0.122\\(0.025)} & \makecell{0.138\\(0.048)} & \makecell{0.177\\(0.054)} & \makecell{-0.004\\(0.03)} & + & + & + & 0 \\ 
		Prod&	\makecell{0.111\\(0.037)} & \makecell{0.143\\(0.072)} & \makecell{0.372\\(0.081)} & \makecell{-0.016\\(0.045)} & + & + & + & 0 \\ 
		CPI	& \makecell{0.04\\(0.063)} & \makecell{0.184\\(0.12)} & \makecell{0.001\\(0.137)} & \makecell{0.532\\(0.077)} & 0 & 0 & 0 & + \\ 
		\hline
	\end{tabular}
	\caption{\label{K2_econ_A2}The MLE of $\bmA_{2,1}$ of the MMAR(2;1,1) model, with standard errors given in parentheses.}
\end{table}

\begin{table}[ht]
	\centering
	\begin{tabular}{llllll||lllll}
		\hline
		&USA& DEU& FRA& GBR &CAN &USA& DEU& FRA& GBR &CAN \\ 
		\hline
		USA& \makecell{0.497\\(0.03)} & \makecell{-0.139\\(0.039)} & \makecell{0.003\\(0.042)} & \makecell{0.145\\(0.032)} & \makecell{-0.039\\(0.035)} & + & - & 0 & + & 0 \\ 
		DEU& \makecell{0.111\\(0.039)} & \makecell{0.361\\(0.036)} & \makecell{0.063\\(0.042)} & \makecell{-0.004\\(0.033)} & \makecell{-0.006\\(0.036)} & + & + & 0 & 0 & 0 \\ 
		FRA&	\makecell{0.075\\(0.029)} & \makecell{0.465\\(0.028)} & \makecell{-0.117\\(0.03)} & \makecell{0.032\\(0.024)} &   \makecell{0.115\\(0.028)} & + & + & - & 0 & + \\ 
		GBR&	\makecell{0.223\\(0.036)} & \makecell{-0.042\\(0.037)} & \makecell{-0.008\\(0.039)} & \makecell{0.342\\(0.031)} & \makecell{0.026\\(0.034)} & + & 0 & 0 & + & 0 \\ 
		CAN&	\makecell{0.109\\(0.043)} & \makecell{-0.059\\(0.044)} & \makecell{-0.001\\(0.047)} & \makecell{0.261\\(0.036)} & \makecell{0.25\\(0.041)} & + & 0 & 0 & + & + \\ 
		\hline
	\end{tabular}
	\caption{The MLE of $\bmB_{2,1}$ of the MMAR(2;1,1) model, with standard errors given in parentheses.}
\end{table}

\begin{table}[ht]
	\centering
	\begin{tabular}{llllll||lllll}
		\hline
		&USA& DEU& FRA& GBR &CAN &USA& DEU& FRA& GBR &CAN\\
		\hline
		&Int \makecell{0.053\\(0.045)} & \makecell{0.051\\(0.045)} & \makecell{0\\(0.034)} & \makecell{0.055\\(0.042)} & \makecell{-0.041\\(0.05)} & 0 & 0 & 0 & 0 & 0 \\ 
		&GDP\makecell{0.039\\(0.029)} & \makecell{0.025\\(0.029)} & \makecell{0.005\\(0.022)} & \makecell{0.027\\(0.027)} & \makecell{0.042\\(0.033)} & 0 & 0 & 0 & 0 & 0 \\ 
		&Prod \makecell{0.052\\(0.044)} & \makecell{0.094\\(0.044)} & \makecell{0.051\\(0.033)} & \makecell{-0.012\\(0.041)} & \makecell{0.066\\(0.05)} & 0 & + & 0 & 0 & 0 \\ 
		&CPI \makecell{0.036\\(0.074)} & \makecell{0.041\\(0.075)} & \makecell{0.007\\(0.056)} & \makecell{-0.041\\(0.07)} & \makecell{-0.014\\(0.083)} & 0 & 0 & 0 & 0 & 0 \\ 
		\hline
	\end{tabular}
	\caption{\label{K2_econ_c2}The MLE of $\bmC_{2,1}$ of the MMAR(2;1,1) model, with standard errors given in parentheses.}
\end{table}

\begin{figure}[h]
	\begin{center}
		\scalebox{0.1}{\includegraphics{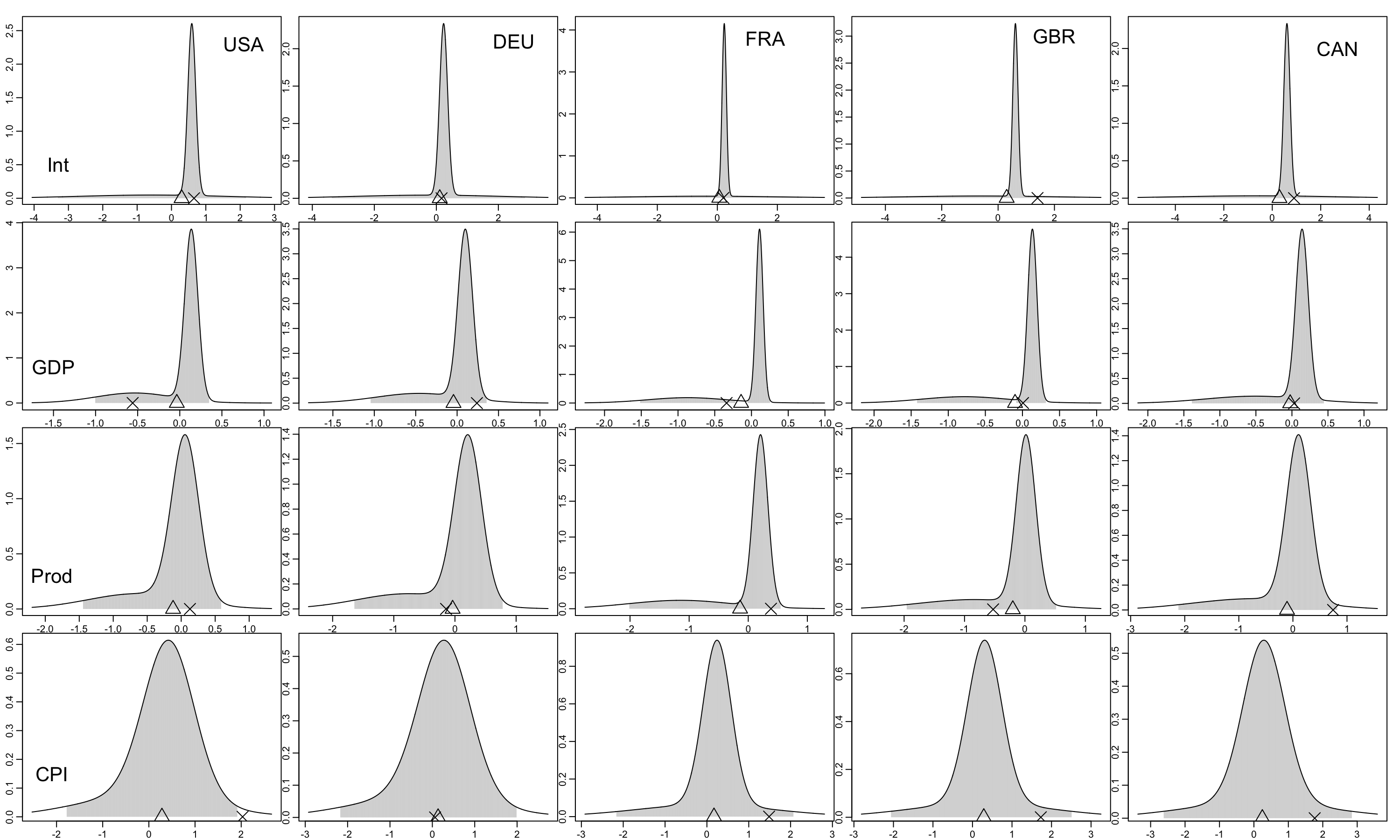}}
	\end{center}
	\vspace{-0.2cm}
	\caption{\label{K2_t129}One-step marginal predictive distribution for Q1 2022 under MMAR(2;1,1) model, with $\times$ representing the observed values and $\triangle$ the predicted values, and the shaded areas representing the 95\% highest density interval.}
\end{figure}

\begin{figure}[h]
	\begin{center}
		\scalebox{0.1}{\includegraphics{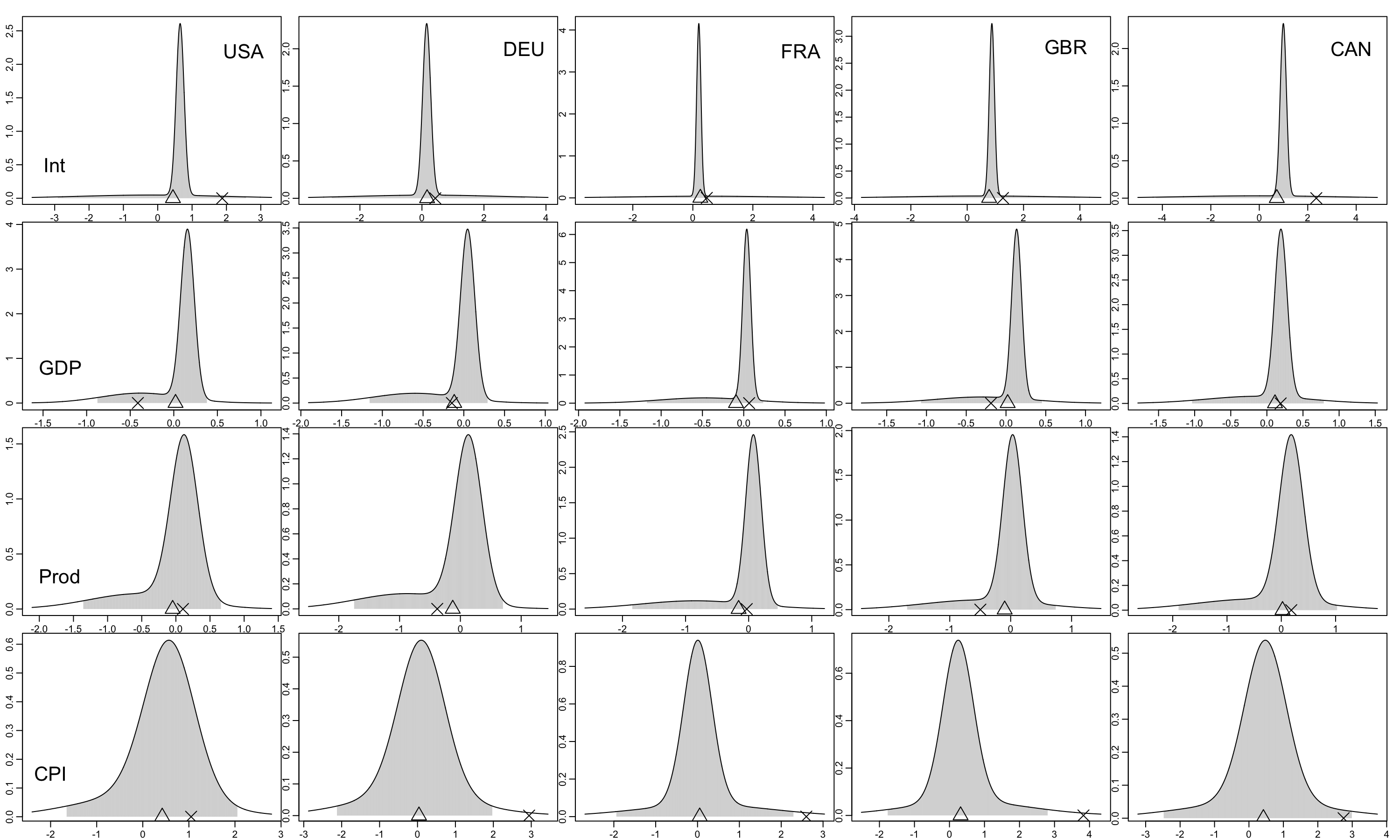}}
	\end{center}
	\vspace{-0.2cm}
	\caption{\label{K2_t130}One-step marginal predictive distribution for Q2 2022 under MMAR(2;1,1) model, with $\times$ representing the observed values and $\triangle$ the predicted values, and the shaded areas representing the 95\% highest density interval.}
\end{figure}

\begin{figure}[h]
	\begin{center}
		\scalebox{0.1}{\includegraphics{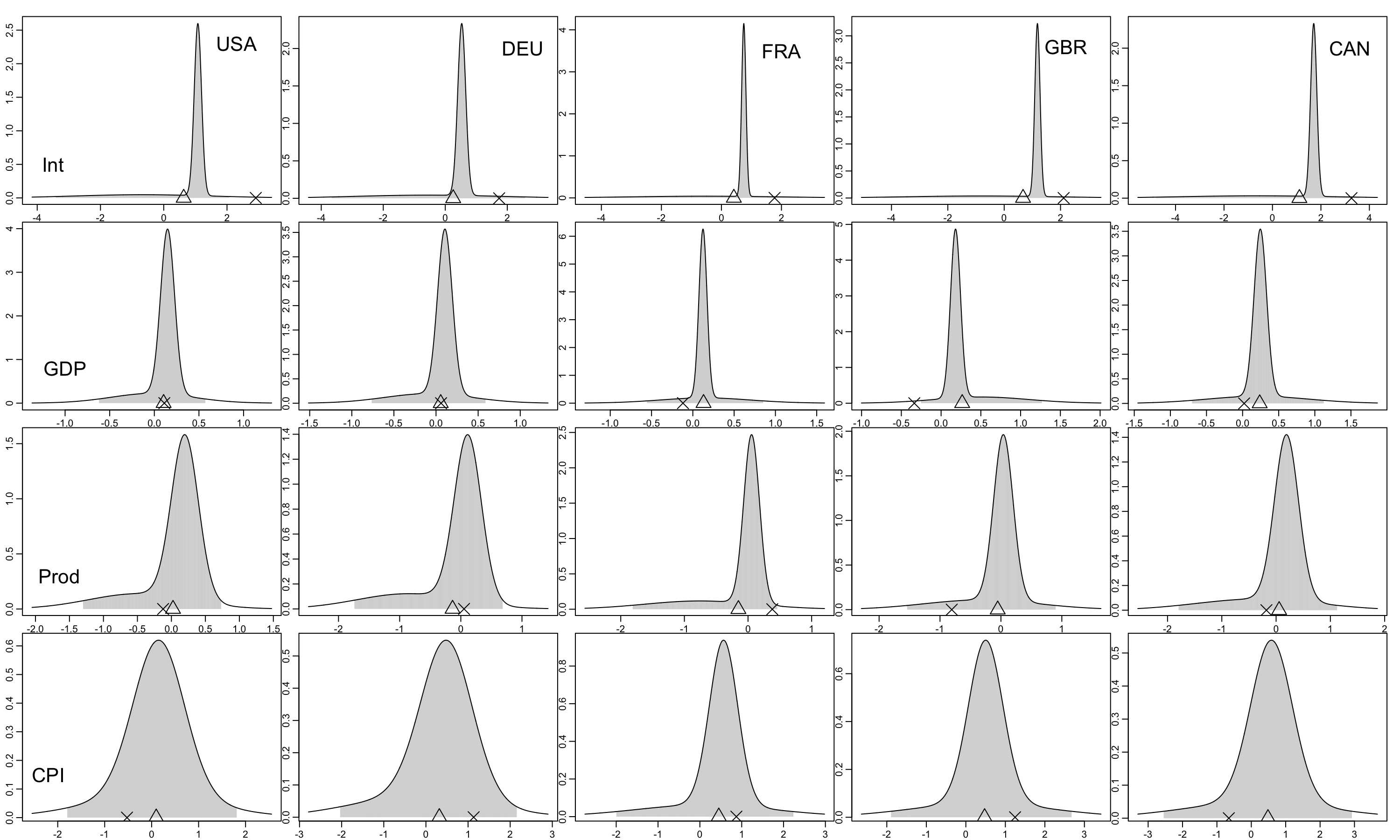}}
	\end{center}
	\vspace{-0.2cm}
	\caption{\label{K2_t131}One-step marginal predictive distribution for Q3 2022 under MMAR(2;1,1) model, with $\times$ representing the observed values and $\triangle$ the predicted values, and the shaded areas representing the 95\% highest density interval.}
\end{figure}

\begin{figure}[h]
	\begin{center}
		\scalebox{0.1}{\includegraphics{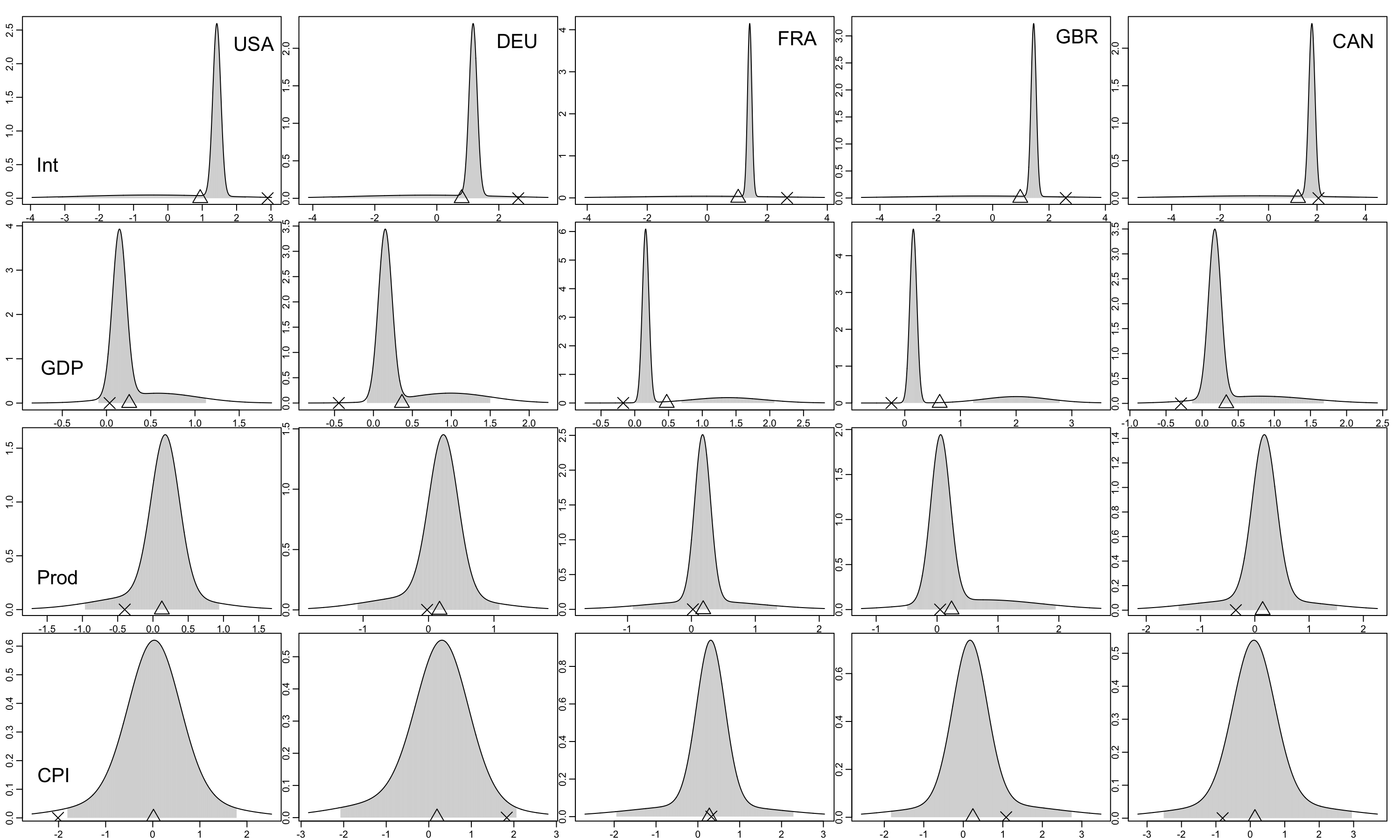}}
	\end{center}
	\vspace{-0.2cm}
	\caption{\label{K2_t132}One-step marginal predictive distribution for Q4 2022 under MMAR(2;1,1) model, with $\times$ representing the observed values and $\triangle$ the predicted values, and the shaded areas representing the 95\% highest density interval.}
\end{figure}

\clearpage
\subsection{Analysis Under MMAR(3;1,1,1) Model}
\label{sec_K3_real}
In this subsection, we present the real data analysis results under an MMAR(3;1,1,1) model. The convergence of the EM algorithm is monitored using the log-likelihood values across iterations. As shown in Fig.~\ref{ll_trace}, the log-likelihood increases monotonically and eventually stabilizes. The algorithm stops when the increase in log-likelihood is less than the predefined threshold of $5\times10^{-4}$.
\begin{figure}[ht!]
	\centering
	\includegraphics[width=.6\textwidth]{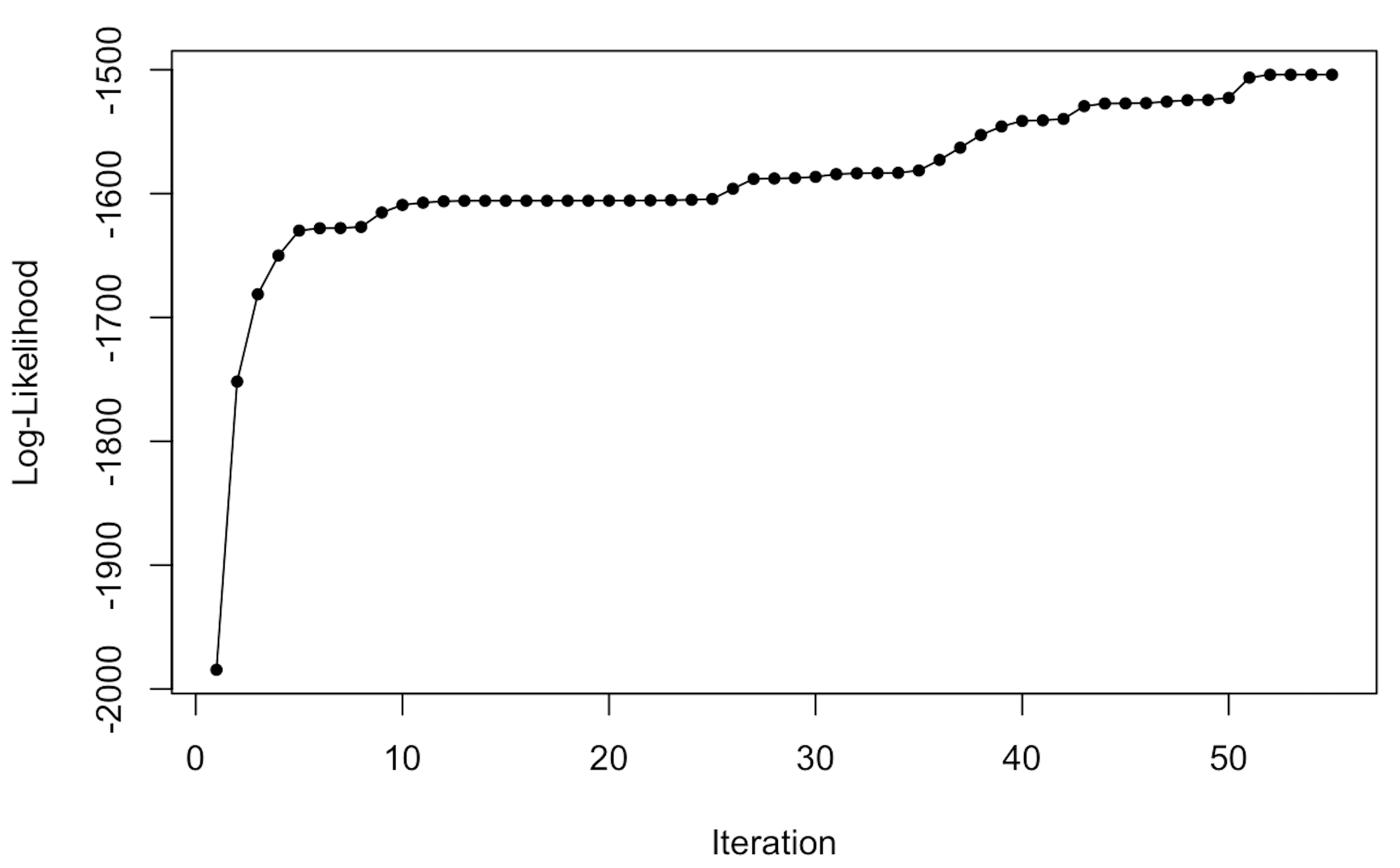}
	\caption[EM algorithm convergence plot for the fitted MMAR(3;1,1,1) model]{\label{ll_trace} EM algorithm convergence plot for the fitted MMAR(3;1,1,1) model.}
\end{figure}
The standardized residuals of the fitted model (Fig.~\ref{econ_acf_std}) reveal no temporal patterns, suggesting a good fit. The estimated mixing weights are $\hat{\alpha}_1=0.107(0.027)$, $\hat{\alpha}_2=0.272(0.039)$ and $\hat{\alpha}_3=0.621(0.043)$, where standard errors are shown in parentheses.
\begin{figure}[ht]
	\begin{center}
		\scalebox{0.5}{\includegraphics{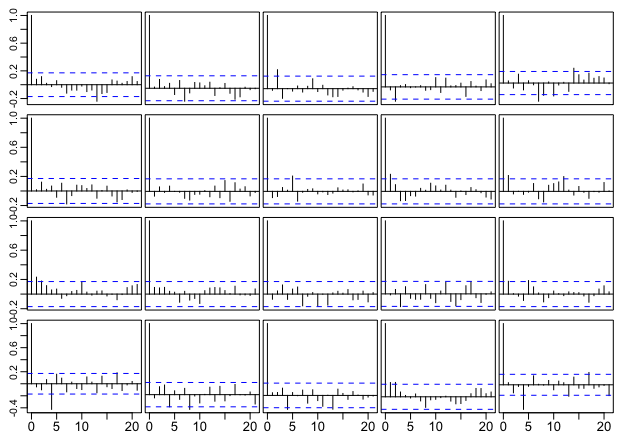}}
	\end{center}
	\vspace{-0.2cm}
	\caption{\label{econ_acf_std} ACF of standardized residuals of MMAR(3;1,1,1) model.}
\end{figure}
Since
\begin{equation*}
	\rho(\hat{\bmB}_{1,1}\otimes\hat{\bmA}_{1,1})=0.735<1,~\rho(\hat{\bmB}_{2,1}\otimes\hat{\bmA}_{2,1})= 1.229>1,~ \rho(\hat{\bmB}_{3,1}\otimes\hat{\bmA}_{3,1})=0.598<1,
\end{equation*}
both the first and the third component of the mixture are weakly stationary while the second component is not weakly stationary. Moreover,
\begin{equation*}
	\sum_{k=1}^3\hat{\alpha}_k\log(\rho(\hat{\bmB}_{k,1})\rho(\hat{\bmA}_{k,1}))=-0.296<0,\qquad
	\sum_{k=1}^3\hat{\alpha}_k(\rho(\hat{\bmB}_{k,1})\rho(\hat{\bmA}_{k,1}))^6=0.982<1.
\end{equation*}
By Corollaries \ref{mmar_ss_cor} and \ref{mmar_ws_cor},
the overall model is strictly stationary, whose stationary distribution has a finite sixth-order moment. Based on the fitted model, the data are classified into three regimes, as shown in Fig.~\ref{econ_c}, where regime 1 is shaded yellow, regime 2 is shaded red, and regime 3 is unshaded. Note that regime 1 generally exhibits the strongest volatility, regime 2 has moderately strong volatility, and regime 3 has relatively weak volatility.
\begin{figure}[h]
	\begin{center}
		\includegraphics[width=1\textwidth]{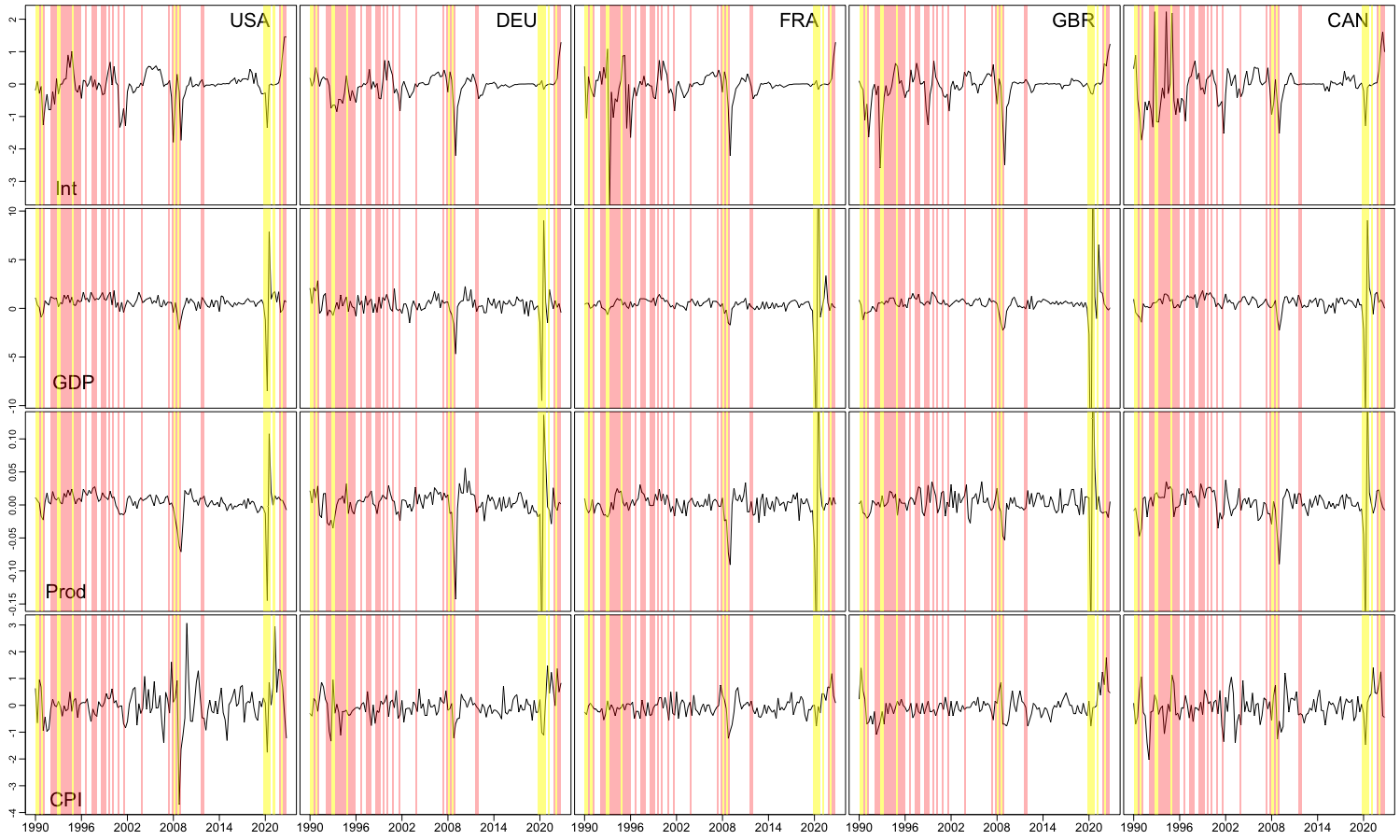}
	\end{center}
	\vspace{-0.2cm}
	\caption{\label{econ_c} Time series plot of the
		economic indicators, with regime 1 shaded yellow, regime 2 shaded red, and regime 3 unshaded.}
\end{figure}

Tables \ref{econ_A1} -- \ref{econ_c3} show the MLE of the parameter matrices $\bmA_{k,1}, \bmB_{k,1}$ and $\bmC_{k,1}$ for $k\in\{1,2,3\}$, and the corresponding standard errors, respectively. Due to the identifiability constraints, the Frobenius norms of $\bmB$'s are scaled to 1. 

The estimated parameter matrices $\hat{\bmA}$'s and $\hat{\bmB}$'s demonstrate both differences and similarities among different regimes. Concerning the differences, one example is that the second column of $\hat{\bmA}_{2,1}$ indicates that the GDP growth of the previous quarter  has a significantly positive influence on all the economic indicators in the current quarter. However, upon examining the second column of  $\hat{\bmA}_{3,1}$ , the previous quarter's GDP growth does not have a significant impact on all the indicators of the current quarter, except for itself.
Regarding the similarities, by checking the first columns of $\hat{\bmB}$'s, it is observed that
the US's previous quarter's indicators consistently have a positive effect on current quarter's indicators from all the countries across the three regimes, with only a few exceptions. 

Moreover, the out-of-sample prediction performance is also examined. We use the data from Q1 1990 to Q2 2021 ($1\leq t\leq 126$) to fit the model and derive the MLE of the parameter.
Subsequently, we derive the marginal predictive distributions for the period from Q3 2021 to Q4 2022 ($127\leq t\leq 132$),  based on \eqref{mixture density 1} with the parameter replaced by the MLE. 
The observed values along with the predictive values based on the conditional mean are shown in Figs.~\ref{t127} -- \ref{t132}. 
\begin{figure}[ht!]
	\begin{center}
		\scalebox{0.2}{\includegraphics{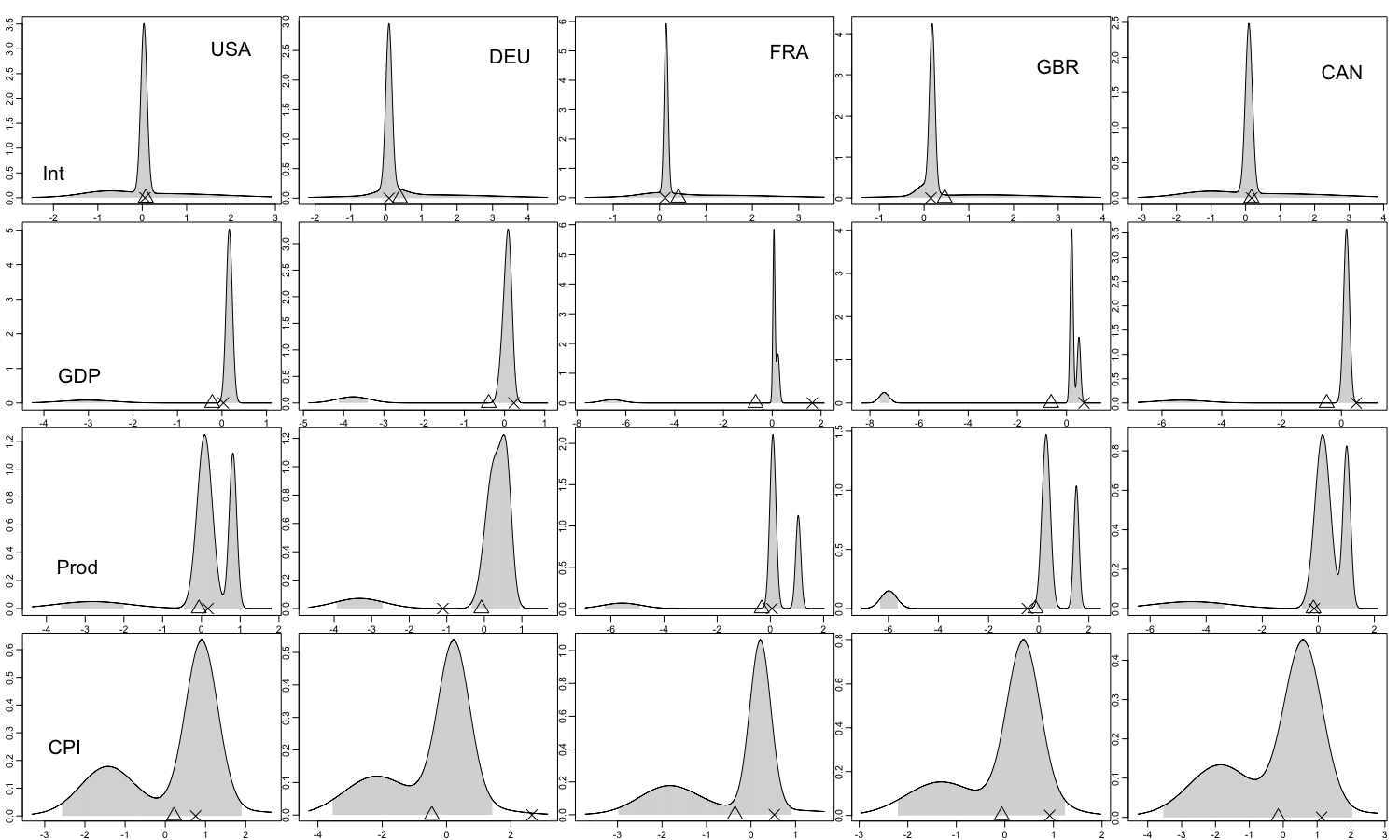}}
	\end{center}
	\vspace{-0.2cm}
	\caption{\label{t127} One-step marginal predictive distribution for Q3 2021 under MMAR(3;1,1,1) model, with $\times$ representing the observed values and $\triangle$ the predicted values, and the shaded areas representing the 95\% highest density interval.}
\end{figure}

\begin{figure}[ht!]
	\begin{center}
		\scalebox{0.2}{\includegraphics{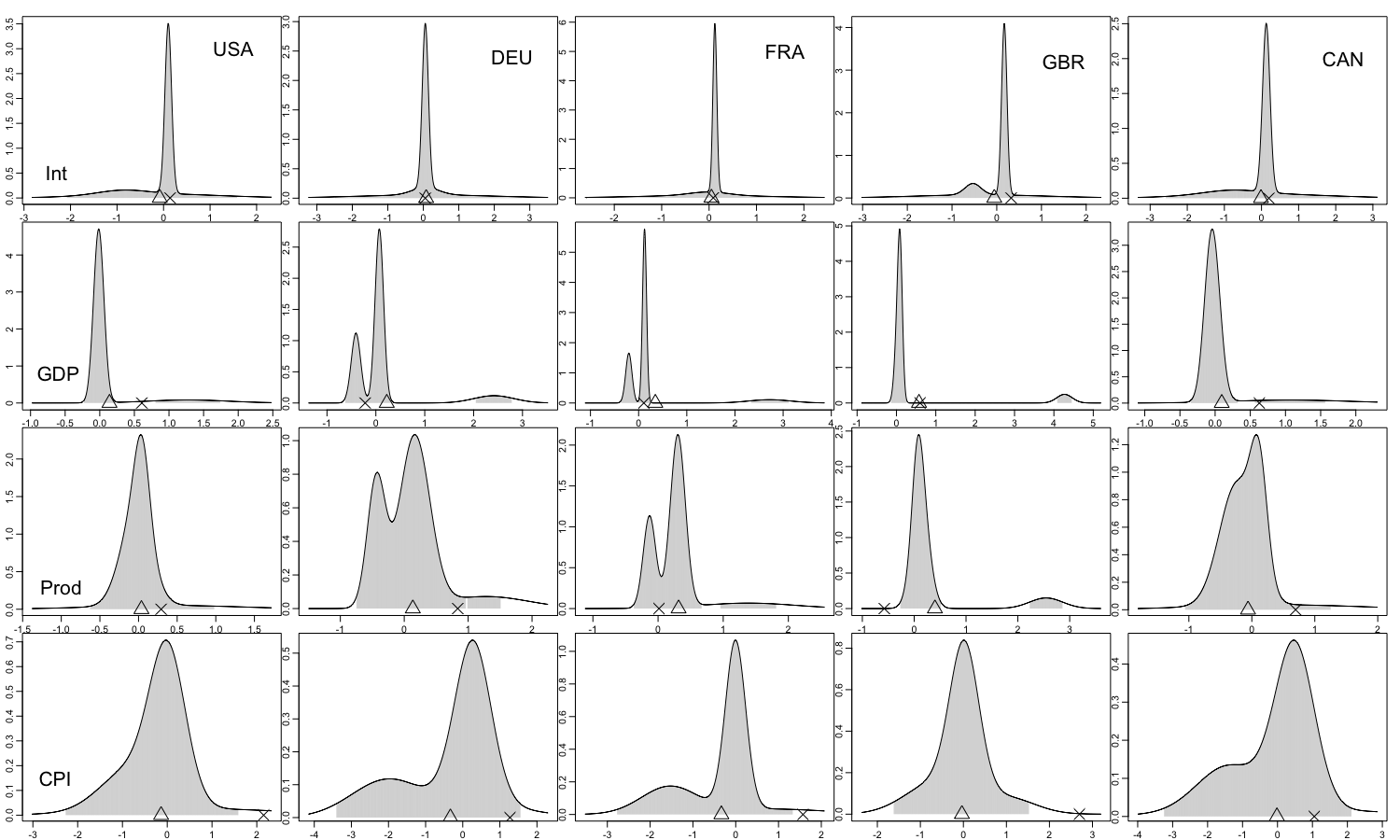}}
	\end{center}
	\vspace{-0.2cm}
	\caption{\label{t128}One-step marginal predictive distribution for Q4 2021 under MMAR(3;1,1,1) model, with $\times$ representing the observed values and $\triangle$ the predicted values, and the shaded areas representing the 95\% highest density interval.}
\end{figure}

\begin{table}[ht]
	\centering
	\begin{tabular}{lllll||llll}
		\hline
		& Int & GDP & Prod & CPI & Int & GDP & Prod & CPI\\
		\hline
		Int&\makecell{-1.253\\(0.515)} & \makecell{-0.539\\(1.202)} & \makecell{0.329\\(0.887)} & \makecell{1.351\\(0.495)} & - & 0 & 0 & + \\ 
		GDP &\makecell{0.886\\(0.239)} & \makecell{4.5\\(0.625)} & \makecell{8.251\\(0.536)} & \makecell{-0.849\\(0.235)} & + & + & + & - \\ 
		Prod&\makecell{0.46\\(0.333)} & \makecell{3.662\\(0.838)} & \makecell{4.607\\(0.617)} & \makecell{0.459\\(0.338)} & 0 & + & + & 0 \\ 
		CPI&\makecell{-1.045\\(0.612)} & \makecell{-1.438\\(1.479)} & \makecell{2.616\\(1.076)} & \makecell{0.7\\(0.607)} & 0 & 0 & + & 0 \\ 
		\hline
	\end{tabular}
	\caption{\label{econ_A1}The MLE of $\bmA_{1,1}$ of the MMAR(3;1,1,1) model, with standard errors given in parentheses.}
\end{table}

\begin{table}[ht]
	\centering
	\begin{tabular}{llllll||lllll}
		\hline
		&USA& DEU& FRA& GBR &CAN &USA& DEU& FRA& GBR &CAN \\ 
		\hline
		USA&\makecell{0.057\\(0.03)} & \makecell{-0.089\\(0.023)} & \makecell{0.196\\(0.025)} & \makecell{-0.12\\(0.021)} & \makecell{-0.068\\(0.026)} & 0 & - & + & - & - \\ 
		DEU&\makecell{0.122\\(0.029)} & \makecell{-0.104\\(0.023)} & \makecell{0.279\\(0.023)} & \makecell{-0.165\\(0.02)} & \makecell{-0.155\\(0.024)} & + & - & + & - & - \\ 
		FRA&\makecell{0.124\\(0.035)} & \makecell{-0.18\\(0.028)} & \makecell{0.384\\(0.029)} & \makecell{-0.232\\(0.025)} & \makecell{-0.162\\(0.03)} & + & - & + & - & - \\ 
		GBR &\makecell{0.154\\(0.024)} & \makecell{-0.236\\(0.02)} & \makecell{0.469\\(0.026)} & \makecell{-0.297\\(0.02)} & \makecell{-0.136\\(0.021)} & + & - & + & - & - \\ 
		CAN&\makecell{0.011\\(0.04)} & \makecell{-0.123\\(0.03)} & \makecell{0.241\\(0.032)} & \makecell{-0.142\\(0.026)} & \makecell{-0.032\\(0.034)} & 0 & - & + & - & 0 \\ 
		\hline
	\end{tabular}
	\caption{The MLE of $\bmB_{1,1}$ of the MMAR(3;1,1,1) model, with standard errors given in parentheses.}
\end{table}

\begin{table}[ht]
	\centering
	\begin{tabular}{llllll||lllll}
		\hline
		&USA& DEU& FRA& GBR &CAN &USA& DEU& FRA& GBR &CAN\\
		\hline
		Int&\makecell{-0.225\\(0.342)} & \makecell{0.046\\(0.324)} & \makecell{-0.088\\(0.4)} & \makecell{0.021\\(0.266)} & \makecell{-0.144\\(0.441)} & 0 & 0 & 0 & 0 & 0 \\ 
		GDP& \makecell{-0.35\\(0.162)} & \makecell{-0.073\\(0.154)} & \makecell{-0.378\\(0.19)} & \makecell{-0.32\\(0.126)} & \makecell{-0.495\\(0.208)} & - & 0 & - & - & - \\ 
		Prod  &\makecell{-0.389\\(0.23)} & \makecell{-0.465\\(0.216)} & \makecell{-0.63\\(0.267)} & \makecell{-0.198\\(0.177)} & \makecell{-0.43\\(0.296)} & 0 & - & - & 0 & 0 \\ 
		CPI  &\makecell{0.058\\(0.411)} & \makecell{0.173\\(0.388)} & \makecell{0.364\\(0.482)} & \makecell{0.906\\(0.319)} & \makecell{0.586\\(0.524)} & 0 & 0 & 0 & + & 0 \\ 
		\hline
	\end{tabular}
	\caption{The MLE of $\bmC_{1,1}$ of the MMAR(3;1,1,1) model, with standard errors given in parentheses.}
\end{table}

\begin{table}[ht]
	\centering
	\begin{tabular}{lllll||llll}
		\hline
		& Int & GDP & Prod & CPI & Int & GDP & Prod & CPI\\
		\hline
		Int&\makecell{1.695\\(0.142)} & \makecell{2.16\\(0.271)} & \makecell{-0.539\\(0.253)} & \makecell{0.207\\(0.146)} & + & + & - & 0 \\ 
		GDP& \makecell{0.094\\(0.033)} & \makecell{0.233\\(0.075)} & \makecell{0.662\\(0.077)} & \makecell{-0.076\\(0.04)} & + & + & + & 0 \\ 
		Prod&\makecell{0.063\\(0.048)} & \makecell{0.532\\(0.11)} & \makecell{1.331\\(0.117)} & \makecell{0.064\\(0.056)} & 0 & + & + & 0 \\ 
		CPI&\makecell{-0.153\\(0.07)} & \makecell{2.515\\(0.197)} & \makecell{-1.414\\(0.16)} & \makecell{0.767\\(0.09)} & - & + & - & + \\ 
		\hline
	\end{tabular}
	\caption{The MLE of $\bmA_{2,1}$ of the MMAR(3;1,1,1) model, with standard errors given in parentheses.}
\end{table}

\begin{table}[ht]
	\centering
	\begin{tabular}{llllll||lllll}
		\hline
		&USA& DEU& FRA& GBR &CAN &USA& DEU& FRA& GBR &CAN \\ 
		\hline
		USA&\makecell{0.503\\(0.036)} & \makecell{0.056\\(0.046)} & \makecell{0.036\\(0.053)} & \makecell{0.099\\(0.041)} & \makecell{-0.076\\(0.045)} & + & 0 & 0 & + & 0 \\ 
		DEU&\makecell{0.201\\(0.069)} & \makecell{0.398\\(0.056)} & \makecell{0.062\\(0.067)} & \makecell{0.213\\(0.053)} & \makecell{0.022\\(0.061)} & + & + & 0 & + & 0 \\ 
		FRA&\makecell{0.193\\(0.039)} & \makecell{0.291\\(0.035)} & \makecell{-0.013\\(0.039)} & \makecell{0.228\\(0.034)} & \makecell{0.063\\(0.036)} & + & + & 0 & + & 0 \\ 
		GBR&\makecell{0.207\\(0.06)} & \makecell{-0.024\\(0.052)} & \makecell{-0.015\\(0.059)} & \makecell{0.366\\(0.047)} & \makecell{0.061\\(0.054)} & + & 0 & 0 & + & 0 \\ 
		CAN &\makecell{0.29\\(0.074)} & \makecell{0.097\\(0.071)} & \makecell{0.027\\(0.082)} & \makecell{0.138\\(0.064)} & \makecell{0.08\\(0.075)} & + & 0 & 0 & + & 0 \\ 
		\hline
	\end{tabular}
	\caption{The MLE of $\bmB_{2,1}$ of the MMAR(3;1,1,1) model, with standard errors given in parentheses.}
\end{table}

\begin{table}[ht]
	\centering
	\begin{tabular}{llllll||lllll}
		\hline
		&USA& DEU& FRA& GBR &CAN &USA& DEU& FRA& GBR &CAN\\
		\hline
		Int&\makecell{-0.094\\(0.153)} & \makecell{-0.22\\(0.198)} & \makecell{-0.219\\(0.116)} & \makecell{-0.281\\(0.174)} & \makecell{-0.209\\(0.238)} & 0 & 0 & 0 & 0 & 0 \\ 
		GDP& \makecell{0.017\\(0.043)} & \makecell{-0.086\\(0.055)} & \makecell{-0.006\\(0.032)} & \makecell{-0.025\\(0.048)} & \makecell{0.008\\(0.067)} & 0 & 0 & 0 & 0 & 0 \\ 
		Prod&\makecell{0.014\\(0.06)} & \makecell{-0.168\\(0.078)} & \makecell{-0.005\\(0.045)} & \makecell{0.031\\(0.068)} & \makecell{-0.047\\(0.093)} & 0 & - & 0 & 0 & 0 \\ 
		CPI& \makecell{-0.16\\(0.09)} & \makecell{-0.306\\(0.116)} & \makecell{-0.118\\(0.068)} & \makecell{-0.142\\(0.101)} & \makecell{-0.187\\(0.139)} & 0 & - & 0 & 0 & 0 \\ 
		\hline
	\end{tabular}
	\caption{The MLE of $\bmC_{2,1}$ of the MMAR(3;1,1,1) model, with standard errors given in parentheses.}
\end{table}

\begin{table}[ht]
	\centering
	\begin{tabular}{lllll||llll}
		\hline
		& Int & GDP & Prod & CPI & Int & GDP & Prod & CPI \\
		\hline
		Int&\makecell{1.472\\(0.044)} & \makecell{0.347\\(0.076)} & \makecell{0.044\\(0.064)} & \makecell{0.145\\(0.036)} & + & + & 0 & + \\ 
		GDP  &\makecell{0.302\\(0.036)} & \makecell{-0.044\\(0.088)} & \makecell{-0.005\\(0.074)} & \makecell{-0.01\\(0.042)} & + & 0 & 0 & 0 \\ 
		Prod&\makecell{0.054\\(0.053)} & \makecell{-0.072\\(0.132)} & \makecell{0.122\\(0.11)} & \makecell{0.068\\(0.063)} & 0 & 0 & 0 & 0 \\ 
		CPI &\makecell{0.251\\(0.086)} & \makecell{-0.155\\(0.217)} & \makecell{-0.024\\(0.181)} & \makecell{0.676\\(0.104)} & + & 0 & 0 & + \\ 
		\hline
	\end{tabular}
	\caption{The MLE of $\bmA_{3,1}$ of the MMAR(3;1,1,1) model, with standard errors given in parentheses.}
\end{table}

\begin{table}[ht]
	\centering
	\begin{tabular}{llllll||lllll}
		\hline
		&USA& DEU& FRA& GBR &CAN &USA& DEU& FRA& GBR &CAN \\ 
		\hline
		USA&\makecell{0.511\\(0.019)} & \makecell{-0.25\\(0.021)} & \makecell{-0.016\\(0.027)} & \makecell{0.036\\(0.024)} & \makecell{-0.062\\(0.023)} & + & - & 0 & 0 & - \\ 
		DEU& \makecell{0.068\\(0.024)} & \makecell{0.366\\(0.023)} & \makecell{-0.058\\(0.026)} & \makecell{-0.115\\(0.023)} & \makecell{-0.028\\(0.022)} & + & + & - & - & 0 \\ 
		FRA&\makecell{0.066\\(0.019)} & \makecell{0.139\\(0.018)} & \makecell{0.21\\(0.02)} & \makecell{-0.115\\(0.018)} & \makecell{-0.066\\(0.017)} & + & + & + & - & - \\ 
		GBR&\makecell{0.062\\(0.022)} & \makecell{-0.143\\(0.019)} & \makecell{0.135\\(0.023)} & \makecell{0.236\\(0.021)} & \makecell{-0.064\\(0.019)} & + & - & + & + & - \\ 
		CAN&\makecell{0.189\\(0.026)} & \makecell{-0.43\\(0.02)} & \makecell{0.295\\(0.025)} & \makecell{0.104\\(0.025)} & \makecell{0.094\\(0.026)} & + & - & + & + & + \\ 
		\hline
	\end{tabular}
	\caption{The MLE of $\bmB_{3,1}$ of the MMAR(3;1,1,1) model, with standard errors given in parentheses.}
\end{table}

\begin{table}[ht]
	\centering
	\begin{tabular}{llllll||lllll}
		\hline
		&USA& DEU& FRA& GBR &CAN &USA& DEU& FRA& GBR &CAN\\
		\hline
		Int&\makecell{0.068\\(0.031)} & \makecell{0.089\\(0.029)} & \makecell{0.125\\(0.023)} & \makecell{0.122\\(0.026)} & \makecell{0.121\\(0.032)} & + & + & + & + & + \\ 
		GDP&\makecell{0.032\\(0.036)} & \makecell{0.038\\(0.035)} & \makecell{0.047\\(0.028)} & \makecell{0.058\\(0.031)} & \makecell{0.052\\(0.038)} & 0 & 0 & 0 & 0 & 0 \\ 
		Prod&\makecell{0.033\\(0.054)} & \makecell{0.131\\(0.052)} & \makecell{0.076\\(0.041)} & \makecell{0.001\\(0.046)} & \makecell{0.069\\(0.057)} & 0 & + & 0 & 0 & 0 \\ 
		CPI&\makecell{0.088\\(0.088)} & \makecell{0.111\\(0.084)} & \makecell{0.038\\(0.067)} & \makecell{0\\(0.074)} & \makecell{0.047\\(0.092)} & 0 & 0 & 0 & 0 & 0 \\ 
		\hline
	\end{tabular}
	\caption{\label{econ_c3}The MLE of $\bmC_{3,1}$ of the MMAR(3;1,1,1) model, with standard errors given in parentheses.}
\end{table}
\begin{figure}[h]
	\begin{center}
		\scalebox{0.2}{\includegraphics{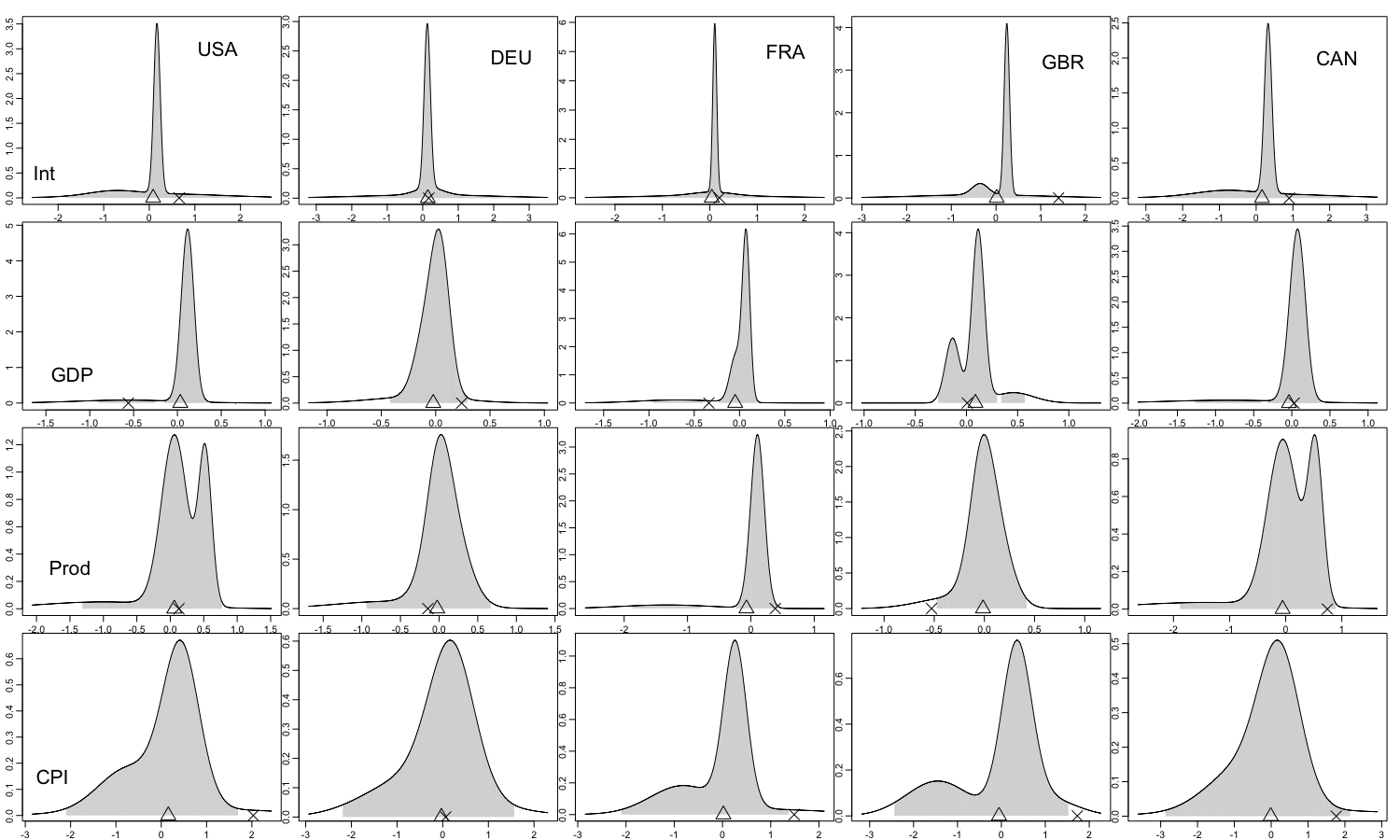}}
	\end{center}
	\vspace{-0.2cm}
	\caption{\label{t129}One-step marginal predictive distribution for Q1 2022 under MMAR(3;1,1,1) model, with $\times$ representing the observed values and $\triangle$ the predicted values, and the shaded areas representing the 95\% highest density interval.}
\end{figure}

\begin{figure}[h]
	\begin{center}
		\scalebox{0.2}{\includegraphics{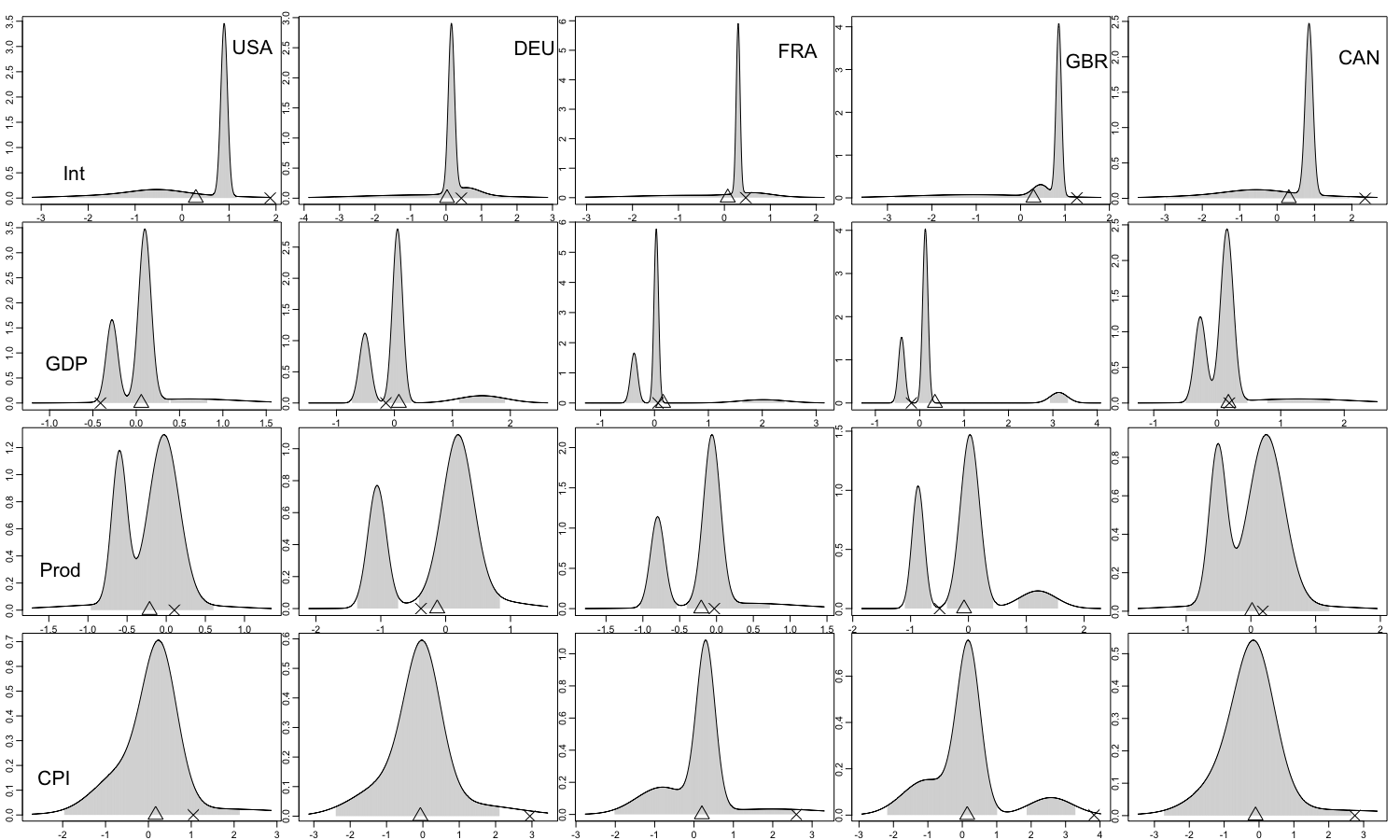}}
	\end{center}
	\vspace{-0.2cm}
	\caption{\label{t130}One-step marginal predictive distribution for Q2 2022 under MMAR(3;1,1,1) model, with $\times$ representing the observed values and $\triangle$ the predicted values, and the shaded areas representing the 95\% highest density interval.}
\end{figure}

\begin{figure}[h]
	\begin{center}
		\scalebox{0.2}{\includegraphics{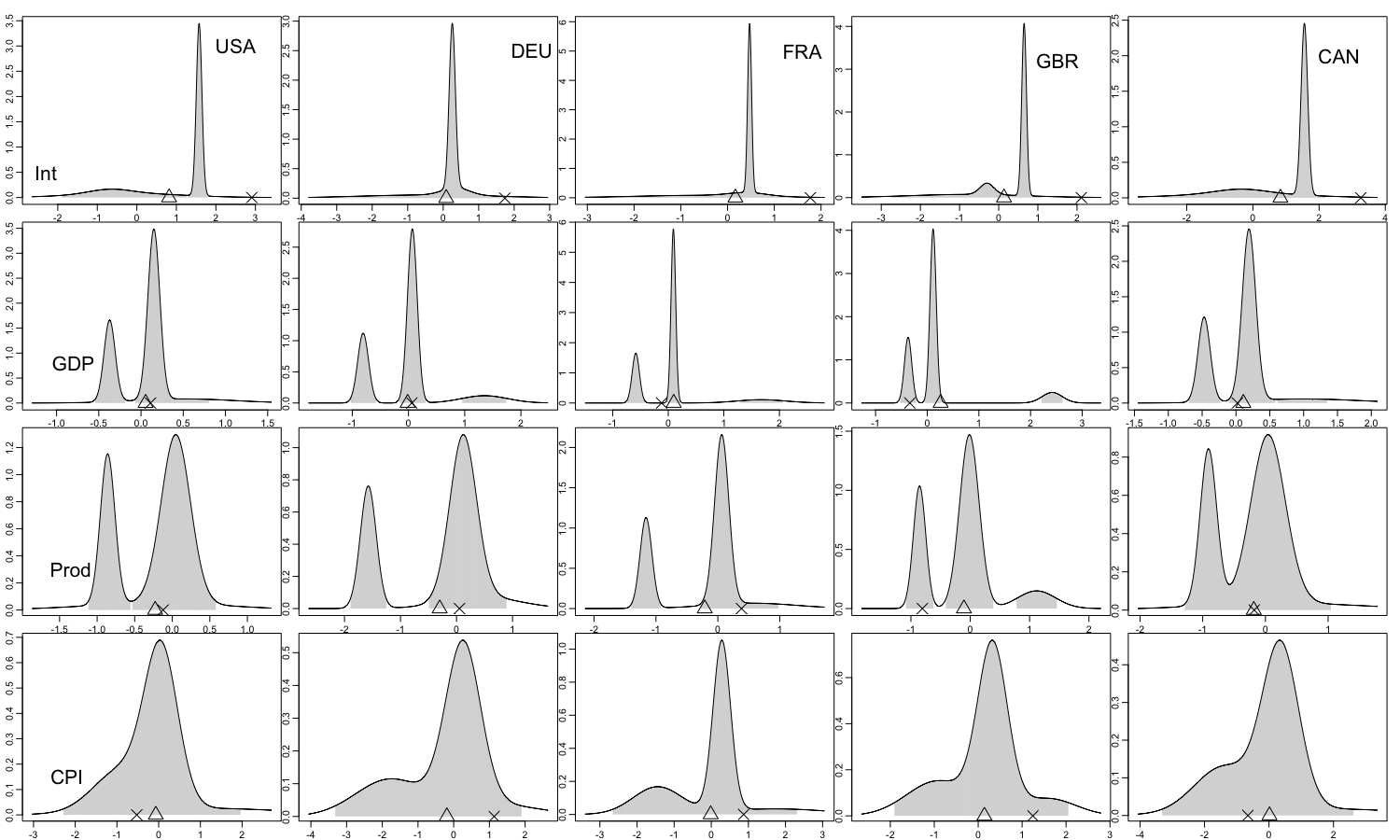}}
	\end{center}
	\vspace{-0.2cm}
	\caption{\label{t131}One-step marginal predictive distribution for Q3 2022 under MMAR(3;1,1,1) model, with $\times$ representing the observed values and $\triangle$ the predicted values, and the shaded areas representing the 95\% highest density interval.}
\end{figure}

\begin{figure}[h]
	\begin{center}
		\scalebox{0.2}{\includegraphics{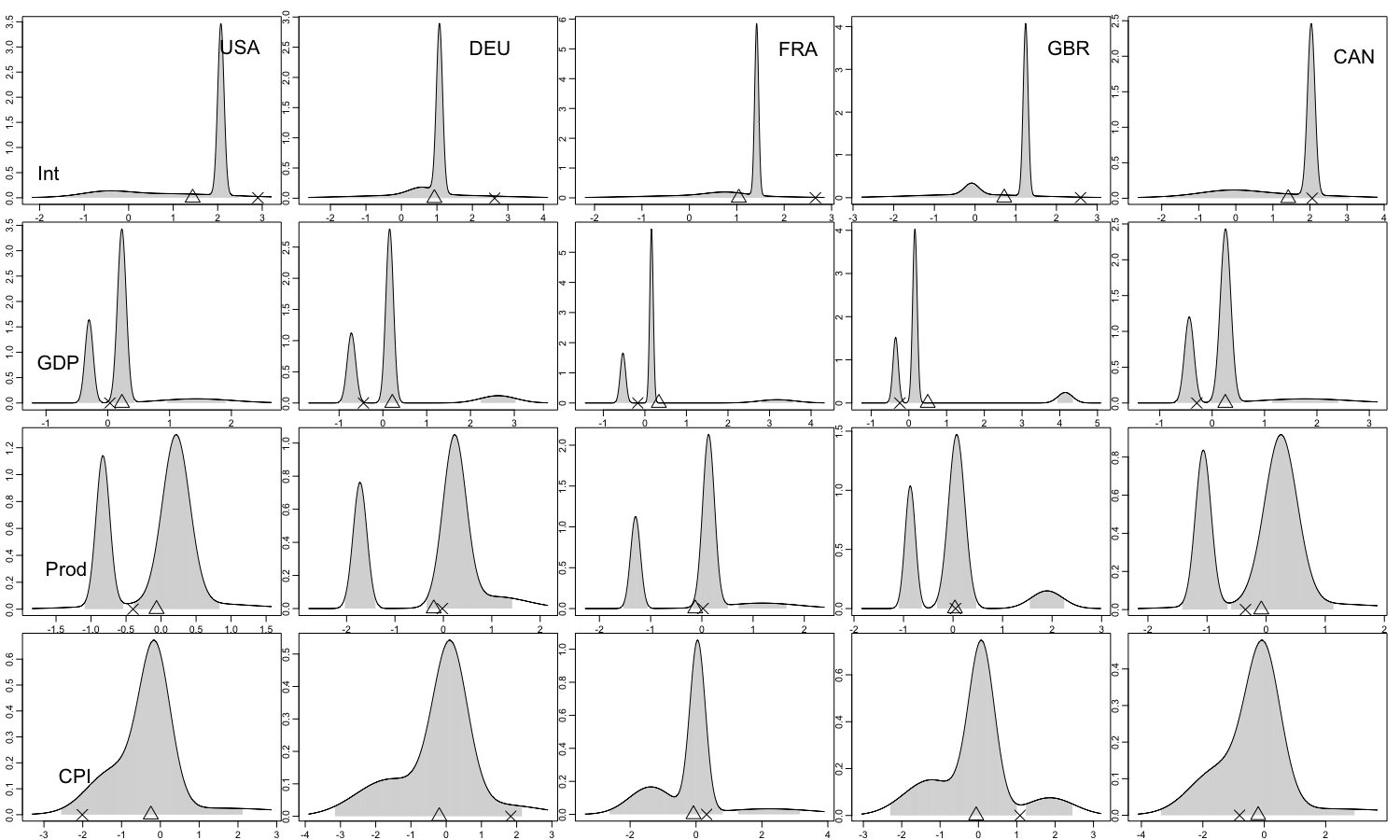}}
	\end{center}
	\vspace{-0.2cm}
	\caption{\label{t132}One-step marginal predictive distribution for Q4 2022 under MMAR(3;1,1,1) model, with $\times$ representing the observed values and $\triangle$ the predicted values, and the shaded areas representing the 95\% highest density interval.}
\end{figure}

\begin{table}[ht!]
	\centering
	\begin{tabular}{|l|l|l|l|l|l|l|}
		\hline
		$K$&$p_{\max}$&log-likehood&AIC&BIC & GIC&HQ  \\ \hline
		1& 1 & -2492.35 & 5152.70 & 5394.22 & 5574.31 & 5250.84 \\ \hline
		1& 2 & -2376.53 & 5001.07 & 5356.64 & 5699.01 & 5145.55 \\ \hline
		1 & 3 & -2286.74 & 4901.49 & 5370.50 & 5895.80 & 5092.06 \\ \hline
		2 & 1 & -1753.87 & 3845.75 & 4331.66 & \textbf{4881.14} & 4043.19 \\ \hline
		2& 2& -1728.56 & 3955.11 & 4669.13 & 5631.34 & 4245.24 \\ \hline
		2& 3 & -1743.21 & 4144.42 & 5085.30 & 6501.23 & 4526.72 \\ \hline
		3& 1& -1504.04 & \textbf{3516.09} & \textbf{4246.39} & 5236.18 & \textbf{3812.84} \\ \hline
		3& 2& -1421.84 & 3591.69 & 4664.15 & 6350.18 & 4027.46 \\ 
		\hline
	\end{tabular}
	\caption{Model selection for economic indicators dataset.}
	\label{econ_selection}
\end{table}

\clearpage
\end{document}